\documentclass[journal]{IEEEtran}
\usepackage{amsmath,amssymb}
\usepackage{xcolor}
\usepackage{bm}
\usepackage{array}
\usepackage{booktabs}
\usepackage{slashed}
\usepackage{algorithm}
\usepackage{algorithmic}
\usepackage{graphicx}
\usepackage[caption=false,font=footnotesize]{subfig}
\usepackage{multirow}
\usepackage{makecell}
\usepackage{threeparttable}
\usepackage{bibunits}
\usepackage{stfloats}
\usepackage{needspace}
\usepackage{amssymb}
\usepackage{relsize}
\usepackage{xspace}
\newcommand{\papertablespacing}{%
  \footnotesize
  \renewcommand{\arraystretch}{1.25}%
  \setlength{\extrarowheight}{0pt}%
}
\newtheorem{lemma}{\bfseries Lemma}
\newtheorem{corollary}{\bfseries Corollary}
\newtheorem{proposition}{\bfseries Proposition}
\newcommand{\citepair}[2]{\cite{#1}--\cite{#2}}
\newcommand{\citerange}[3]{\cite{#1}\nocite{#2}--\cite{#3}}
\newcommand{\DOAM}{DOAM\xspace}
\usepackage[hidelinks]{hyperref}

\ifCLASSINFOpdf
\else
\fi

\begin{document}
\bstctlcite{IEEEtran:BSTcontrol}
%
\title{Trajectory CPHD Filtering for Multiple Turning Vehicles With Decoupled Orientation and Axial-Scale Estimation}
%
%
%

\author{Yunhe Cao,~\IEEEmembership{Member,~IEEE,}
        Yu Wan,
        Yuanhao Cheng,
        Tat-Soon Yeo,~\IEEEmembership{Life~Fellow,~IEEE,}
        Jie Fu
\thanks{Yunhe Cao, Yu Wan, Yuanhao Cheng and Jie Fu are with the National Key Laboratory of Radar Signal Processing, Xidian University, Xi'an 710071, China, E-mail: caoyunhe@mail.xidian.edu.cn; chengyh@stu.xidian.edu.cn; 25021110918@stu.xidian.edu.cn; 21021210890@stu.xidian.edu.cn;}
\thanks{Tat-Soon Yeo is with the Department of Electrical and Computer Engineering, National University of Singapore, Singapore 119077, E-mail: eleyeots@nus.edu.sg;} \thanks{Corresponding author: Yuanhao Cheng.}
\thanks{Manuscript received A; revised August B.}}

%
%

\markboth{IEEE Transactions on Intelligent Transportation Systems}%
{Shell \MakeLowercase{\textit{et al.}}: x}
%



\maketitle

\begin{abstract}
\begingroup
Reliable estimation of vehicle orientation, centroid, and footprint is important for representing road-user occupancy during vehicle turns at intersections and other common road maneuvers. During a turn, the vehicle orientation and direction of motion may differ, and coupling the orientation with the axial scales can degrade both orientation and extent estimates. This paper proposes a trajectory cardinalized probability hypothesis density filter with a Decoupled Orientation and Axial-Scale Estimation (DOAM) model for multiple-vehicle tracking. The model represents the kinematic, orientation, and squared semi-axis-length states separately within each trajectory component, reducing the mutual interference between orientation variation and scale estimation without introducing an additional interacting multiple-model structure. A structured coordinate-ascent variational inference recursion jointly updates these state sequences and the measurement-source variables. The trajectory-component likelihood and weight update associated with the decoupled representation are derived while retaining the standard cardinality recursion. A fixed-lag trajectory implementation further uses current measurements to correct historical states within the smoothing window. Evaluation with signalized-intersection simulations and real onboard LiDAR measurements shows improved estimation of vehicle orientation, centroid, and extent, particularly during turns. The resulting vehicle-state and footprint estimates provide information for road-user occupancy perception and subsequent collision-risk assessment and motion planning.
\endgroup
\end{abstract}


%
\IEEEpeerreviewmaketitle

\section{Introduction}

\begingroup
\IEEEPARstart{A}{ccurate} estimation of multiple road vehicles supports environmental perception in connected and autonomous vehicles and intelligent traffic control. Radar and light detection and ranging (LiDAR) sensors can provide multiple spatially distributed measurements from each vehicle, allowing it to be represented as an extended target with kinematic and extent states~\citerange{Mihaylova2014SMCReviewGroupExtended}{Cao2024RectangularObjectTITS}{Granstrom2017Overview}. Its centroid, orientation, size, and shape characterize road-space occupancy and provide information for collision-risk assessment and motion planning. Elliptical models describe these quantities using two axial scales and can be divided into implicit matrix and explicit parametric representations~\citerange{Liu2025GMMVB}{Zheng2025HeavyTailed,Cao2021EDA}{Zhang-model-2025}. The random matrix model (RMM) jointly represents the orientation and axial scales using a symmetric positive-definite matrix~\citerange{Koch2008RM}{Jiao2024DistributedRM}{Pei2025ConstrainedRM}, whereas the multiplicative error model (MEM) represents the orientation and two semi-axis lengths as explicit extent states~\citerange{2019YangShishanMEM}{Li-model2023}{Zheng2025OrientationVector}.
\par
\endgroup

\begingroup
Vehicle turning maneuvers are common at intersections, during lane changes, and during obstacle avoidance~\cite{Kellner2016HighResolutionDoppler}. The vehicle’s immediate orientation may then differ from its instantaneous direction of motion, and an inaccurate representation of this relationship can degrade the estimated orientation and occupied road region. In the RMM, orientation is embedded with the axial scales and must be extracted through eigendecomposition. Kinematics-dependent prediction models and rotation transformations improve its adaptability~\citerange{Granstrom2014NewPrediction}{Lan-model1-2016}{Sahin2024TrajectoryAligned}, but retain the implicit orientation-scale representation. Explicit parametric models avoid indirect extraction but involve nonlinear measurement equations and require approximation. Particle filtering and interacting multiple-model methods improve maneuver adaptability~\citerange{Steuernagel2025RBPF}{IMM-SOEKF,Sarkar2026IMM}{Wang2026IMMPF}, but add computational cost and may require a predefined model set and transition probabilities. Expectation-maximization and constraint-based alternatives~\citepair{Li-model2023}{Zheng2025OrientationVector}, \citerange{Liu2021EMExtendedObject}{Liu2022ConstrainedEM}{Wen2024VelocityOrientation} generally rely on motion-orientation relationships or fixed shape parameters.
\par
\endgroup

To address the interaction between orientation and scale and the intractable nonlinear posterior under maneuvering conditions, Tuncer \emph{et al.} separated the target orientation from the two axial scales within the RMM framework and employed variational inference to iteratively update the kinematic, orientation, and scale states~\cite{Tuncer-model-2021}. Variational inference provides an effective means of handling nonconjugate posteriors in extended target tracking~\cite{Jiao2024DistributedRM}, \citepair{Orguner2012VB}{Yang2023AdaptiveVB}. This model explicitly estimates the orientation and axial scales while mitigating their mutual interference. Since it decouples the orientation from the axial scales, this model is referred to as the Decoupled Orientation and Axial-Scale Estimation (DOAM) model in this paper.

\begingroup
\renewcommand{\DOAM}{DOAM\xspace}
For multi-vehicle tracking in road scenes, existing methods can be reviewed
from two complementary perspectives: data-association-based tracking and
filtering that combines extent models with the random finite set (RFS)
framework. The first focuses on inferring uncertain measurement-to-object
correspondences while estimating object states. Recent developments include
particle-based and neural-enhanced belief propagation (BP) for association
inference and joint state estimation.
For example, \cite{Meyer2021ScalableEOT} uses particle-based BP for joint object detection,
measurement association, and geometric extent estimation,
but propagates current-object densities rather than a posterior over
trajectories. \cite{Liang2024NeuralBP} enhances BP with learned sensor
features to improve association and false-alarm rejection, but does not
directly estimate object orientation or size.
The second incorporates RMM or MEM into RFS filters for joint
estimation of object number and
states~\citerange{Li2023JDTCPMBMGGIW}{Xie2023MMPMBM,Yang2019NetworkFlow}{Cheng2025VGMCPHD}.
Variational RFS methods additionally address association uncertainty,
unknown parameters, and extent
states~\citerange{Li2022VBEMCPHD}{Cheng2025VBMultipleExtended}{Tuncer2022MultiEllipsoidalVB}.
However, implementations based on clustering and association-hypothesis
management are sensitive to ambiguous measurement partitions when nearby
vehicles generate overlapping returns. Trajectory RFS formulations enable
direct inference on vehicle state
sequences~\citerange{Wang2022RobustTrajectoryRFS}{Granstrom2025PMBMTrajectories,GarciaFernandez2019TPHDTCPHD}{GarciaFernandez2022TreeTrajectories}.
Combining these two approaches, \cite{Xia2023TPMBBP} integrates BP with
trajectory Poisson multi-Bernoulli (TPMB) filtering, enabling trajectory
inference without explicit enumeration of association hypotheses.
\cite{Lyu2026PLMBBP} combines a Poisson labeled multi-Bernoulli filter with
parameter-based BP and fixed-point iteration to reduce sampling cost and
jointly estimate the measurement rate, kinematics, and extent.
Nevertheless, both filters retain the implicit matrix representation of
orientation and axial scales. This coupling also persists in RMM-based
trajectory filters, while
explicit parametric trajectory models remain affected by nonlinear
approximation
errors~\citerange{Sjudin2021ExtendedTPHD}{2025WeiShaoxiuTAES}{Cheng2026ExplicitExtent}.
The remaining challenge is to extend the single-time \DOAM posterior to
the trajectory space while preserving the temporal dependence of the
axial-scale states and incorporating the resulting likelihood into a
multi-vehicle trajectory-filter update.
\par
\endgroup

\begingroup
To fill this gap, this paper proposes a trajectory cardinalized probability hypothesis density (TCPHD) algorithm based on the \DOAM, termed TCPHD-\DOAM, for multiple turning vehicles. The main contributions are summarized as follows:

\begin{enumerate}\item A TCPHD-\DOAM algorithm is developed for multi-vehicle extended-target tracking in road scenes. The kinematic, orientation, and squared semi-axis-length (SSAL) states are represented separately within each trajectory component, thereby mitigating the interaction between orientation variations and scale estimation during vehicle turns. This capability is achieved using a single-model formulation without introducing an additional interacting multiple-model structure.
	
	\item Structured variational updates are derived for the kinematic, orientation, and SSAL state sequences. The \DOAM-specific trajectory-component likelihood and weight update are derived, while the standard TCPHD cardinality recursion is retained. A fixed-lag $L$-scan implementation is incorporated to enable current measurements to correct historical vehicle states within the scan window.
	
	\item The proposed algorithm is evaluated using simulated measurements from a signalized-intersection scenario and real onboard LiDAR measurements from a multi-vehicle road scenario. The experiments assess vehicle orientation, centroid, and extent estimation under both steady motion and turns, together with computational efficiency.
	
\end{enumerate}
\endgroup

\begingroup
The remainder of this paper is organized as follows. Section~II introduces
the \DOAM and its variational filtering recursion~\cite{Tuncer-model-2021}.
Section~III presents the proposed TCPHD-\DOAM algorithm, including its
trajectory-state recursion, cardinality-distribution recursion, $L$-scan
implementation, and computational complexity. Section~IV presents the
experimental results, and Section~V concludes the paper. The principal notation and
typographical conventions are listed in Tab.~\ref{tab:notations}.
\endgroup

\begin{table*}[!h]
	\centering
	\caption{Notations}
	\label{tab:notations}
	\footnotesize
	\papertablespacing
	\begin{tabular}{>{\raggedright\arraybackslash}p{0.20\textwidth} >{\raggedright\arraybackslash}p{0.74\textwidth}}
		\hline
		\textbf{Notation} & \textbf{Description}\\
		\hline
		
		$\mathbb R$, $\mathbb R_{>0}$, $\mathbb Z_{\geq0}$ &
		The field of real numbers, the set of positive real numbers, and the set of nonnegative integers, respectively.\\
		
		$\mathbb R^n$, $\mathbb R^{m\times n}$ &
		The space of real column vectors of length~$n$ and the space of real matrices of size~$m\times n$, respectively.\\
		
		$\mathrm{KL}\!\left(q(x)\|p(x)\right)$ &
		The Kullback--Leibler divergence from~$q(x)$ to~$p(x)$, defined as
		$\mathrm{KL}\!\left(q(x)\|p(x)\right)=\int q(x)\ln\!\left(q(x)/p(x)\right)\mathrm dx$.\\

		$\mathcal N(\bm x;\bm\mu,\bm\varXi)$ &
		The multivariate Gaussian distribution with mean~$\bm\mu\in\mathbb R^n$ and positive-definite covariance matrix~$\bm\varXi\in\mathbb R^{n\times n}$, satisfying
		$\mathcal N(\bm x;\bm\mu,\bm\varXi)=(2\pi)^{-n/2}|\bm\varXi|^{-1/2}
		\exp[-(\bm x-\bm\mu)^{\mathrm T}\bm\varXi^{-1}(\bm x-\bm\mu)/2]$, where~$|\bm\varXi|$ denotes the determinant of~$\bm\varXi$.\\
	
	$\mathcal{IG}(x;a,b)$ &
	The inverse-Gamma distribution with shape parameter~$a\in\mathbb R_{>0}$ and scale parameter~$b\in\mathbb R_{>0}$, satisfying
	$\mathcal{IG}(x;a,b)=b^a x^{-a-1}\exp(-b/x)/\Gamma(a)$ for~$x>0$, where~$\Gamma(a)=\int_{0}^{\infty}t^{a-1}\exp(-t)\,\mathrm dt$ denotes the Gamma function.\\
	
	$\delta_{a\geq b}$ &
	The generalized indicator function, where~$\delta_{a\geq b}=1$ if~$a\geq b$, and~$\delta_{a\geq b}=0$ otherwise.\\
	
	$\operatorname{tr}(\bm A)$, $\bm a^{\mathrm T}$, $a!$, $|\mathcal A|$, $\operatorname{Dim}(\bm a)$ &
	The trace of matrix~$\bm A$, the transpose of~$\bm a$, the factorial of~$a$, the cardinality of set~$\mathcal A$, and the dimension of~$\bm a$, respectively.\\
	
	$\left[\bm A\right]_{mn}$, $\left[\bm a\right]_n$ &
	The $(m,n)$-th element of matrix~$\bm A$ and the $n$-th element of vector~$\bm a$, respectively.\\

	$\mathbb E[a]$, $\operatorname{Var}(a)$, $\operatorname{Cov}(a,b)$ &
	The expectation and variance of~$a$, and the covariance between~$a$ and~$b$, respectively.\\
	
	$\bm I_n$, $\bm 0_n$ &
	The identity matrix and zero matrix of order~$n$, respectively.\\
	
	$\bm 1_{a\times b}$, $\bm 0_{a\times b}$ &
	The all-ones matrix and all-zeros matrix of size~$a\times b$, respectively.\\
	
	$\operatorname{diag}(a_1,\ldots,a_n)$ &
	The diagonal matrix with diagonal entries~$a_1,\ldots,a_n$. \\
	
	$\operatorname{blkdiag}(\bm A_1,\ldots,\bm A_n)$ & The block-diagonal matrix with diagonal blocks~$\bm A_1,\ldots,\bm A_n$. \\
	
	\hline
	\multicolumn{2}{l}{\textbf{Typographical conventions}}\\
	\hline
	
	\multicolumn{2}{l@{}}{
		\begin{tabular*}{0.987\textwidth}
			{@{\extracolsep{\fill}}llll@{}}
			
			Lowercase italic, e.g.,~$a$
			& Scalar variable.
			& Lowercase bold italic, e.g.,~$\bm a$
			& Vector variable.\\
			
			Uppercase bold italic, e.g.,~$\bm A$
			&  Matrix or stacked-vector variable.
			& Upright typewriter, e.g.,~$\mathtt{a}$
			& Constant.\\
			
			Uppercase calligraphic, e.g.,~$\mathcal A$
			& Function or set.
			& Upright sans serif, e.g.,~$\mathsf{RMM}$
			& Model, algorithm, or descriptive identifier.\\
			
		\end{tabular*}
	}\\
	
	\hline
\end{tabular}
\end{table*}

\section{Decoupled Orientation and Axial-Scale Estimation}
\label{Sec:DOAM_model}

\begingroup
Consider a two-dimensional road scene in which each extended target represents the elliptical footprint of a vehicle.
\endgroup
At time~$k$, the kinematic state, orientation state, and two SSAL states of a single extended target are denoted by~$\bm r_k\in\mathbb R^{d_{\bm r}}$, $\theta_k\in\mathbb R$, and~$\ddot l_{u,k}\in\mathbb R_{>0}$, respectively, where~$u\in\{1,2\}$. Taking the constant-velocity model as an example, the dimension of the kinematic state is~$d_{\bm r}=4$, and the state can be written as
\begin{equation}
\bm r_k=
\begin{bmatrix}
p_{k}^{\mathsf x} &
p_{k}^{\mathsf y} &
\dot p_{k}^{\mathsf x} &
\dot p_{k}^{\mathsf y}
\end{bmatrix}^{\mathrm T}
\in\mathbb R^4
\label{eq:rk_def}
\end{equation}

\noindent where $p_{k}^{\mathsf x}$ and~$p_{k}^{\mathsf y}$ denote the position components of the target centroid, while~$\dot p_{k}^{\mathsf x}$ and~$\dot p_{k}^{\mathsf y}$ denote the corresponding velocity components. The orientation~$\theta_k$ is measured counterclockwise from the positive~$\mathsf x$-axis of the global coordinate system to the first semi-axis of the target. When $l_{1,k}=l_{2,k}$, the orientation does not
affect the extent matrix. The~$u$-th SSAL state is defined as
\begin{equation}
\ddot l_{u,k}=l_{u,k}^{2}\in\mathbb R_{>0},
\qquad u\in\{1,2\}
\label{eq:axial_square_def}
\end{equation}

The target extent matrix in the global coordinate system is a function of the orientation state and the two SSAL states, given by
\begin{equation}
\bm S_k
=
\mathcal R\left(\theta_k\right)\mathcal{W}\left(\ddot l_{1,k},\ddot l_{2,k}\right)\mathcal R(-\theta_k)
:=\mathcal{S}\left(\theta_k,\ddot l_{1,k},\ddot l_{2,k}\right)
\label{eq:Sk_def}
\end{equation}
\noindent where~$\mathcal R(\theta_k)$ is the rotation function and~$\mathcal{W}\left(\ddot l_{1,k},\ddot l_{2,k}\right)$ is the diagonal SSAL mapping function, respectively given by
\begin{equation}
\mathcal R(\theta_k)
=
\begin{bmatrix}
\cos\theta_k & -\sin\theta_k\\
\sin\theta_k & \cos\theta_k
\end{bmatrix}
\label{eq:O_theta_def}
\end{equation}
and
\begin{equation}
\mathcal W\left(\ddot l_{1,k},\ddot l_{2,k}\right)
=
\operatorname{diag}\!\left(
\ddot l_{1,k},\ddot l_{2,k}
\right)
\label{eq:Lk_def}
\end{equation}

Following \cite[Eq.~(4)]{Tuncer-model-2021}, the initial kinematic state, orientation state, and two SSAL states are assumed to be mutually independent. Their joint prior density can therefore be expressed as
\begin{equation}
\begin{aligned}
&p\!\left(
\bm r_0,\theta_0,\ddot l_{1,0},\ddot l_{2,0}
\right)\\
&\quad =
\mathcal N\!\left(
\bm r_0;
\hat{\bm r}_0,
\bm\varXi_0^{\bm r}
\right)
\mathcal N\!\left(
\theta_0;
\hat{\theta}_0,
\varXi_0^{\theta}
\right)
\prod_{u=1}^{2}
\mathcal{IG}\!\left(
\ddot l_{u,0};
a_{u,0},b_{u,0}
\right)
\end{aligned}
\label{eq:DOAM_joint_prior}
\end{equation}

\noindent where $\hat{\bm r}_0\in\mathbb R^{d_{\bm r}}$ and~$\bm\varXi_0^{\bm r}\in\mathbb R^{d_{\bm r}\times d_{\bm r}}$ denote the prior mean and covariance matrix of the initial kinematic state, respectively; $\hat{\theta}_0\in\mathbb R$ and~$\varXi_0^{\theta}\in\mathbb R_{>0}$ denote the prior mean and variance of the initial orientation state, respectively; and~$a_{u,0}\in\mathbb R_{>0}$ and~$b_{u,0}\in\mathbb R_{>0}$ denote the shape and scale parameters of the inverse-Gamma distribution corresponding to the $u$-th initial SSAL state, respectively, where~$u\in\{1,2\}$.

According to \cite[Eq.~(1)]{Tuncer-model-2021}, conditioned on the target state, the $m$-th measurement satisfies

\begin{equation}
\begin{aligned}
p\!\left(
\bm z_k^{\mathsf O,m}
\mid
\bm r_k,\theta_k,\ddot l_{1,k},\ddot l_{2,k}
\right)=
\mathcal N\!\left(
\bm z_k^{\mathsf O,m};
\bm H_k\bm r_k,
\bm R_k+
\mathtt{s}\bm S_k
\right)
\end{aligned}
\label{eq:z_measurement_model}
\end{equation}

\noindent where~$\bm z_k^{\mathsf O,m}\in\mathbb R^{d_{\bm z}}$ denotes the $m$-th measurement generated by the target at time~$k$, and~$d_{\bm z}=2$ is the measurement dimension; $\bm H_k\in\mathbb R^{d_{\bm z}\times d_{\bm r}}$ is the measurement matrix; $\bm R_k\in\mathbb R^{d_{\bm z}\times d_{\bm z}}$ is the measurement-noise covariance matrix; and~$\mathtt{s}\in\mathbb R_{>0}$ is the measurement-spread scale constant. Conditioned on the target state and the number of target-generated measurements, these measurements are mutually independent. The measurement set generated by the target at time~$k$ is denoted by
\begin{equation}
\mathcal Z_k^{\mathsf O}
=
\left\{
\bm z_k^{\mathsf O,m}
\right\}_{m=1}^{M_k^{\mathsf O}},
\qquad
M_k^{\mathsf O}
=
\left|
\mathcal Z_k^{\mathsf O}
\right|
\label{eq:object_measurement_set}
\end{equation}
\noindent where~$M_k^{\mathsf O}$ denotes the number of measurements generated by the target at time~$k$. Under the conditional-independence assumption, the corresponding likelihood follows \cite[Eq.~(7)]{Tuncer-model-2021} and is given by
\begin{equation}
\begin{aligned}
p\!\left(
\mathcal Z_k^{\mathsf O}
\mid
\bm r_k,\theta_k,\ddot l_{1,k},\ddot l_{2,k}
\right)=
\prod_{m=1}^{M_k^{\mathsf O}}
\mathcal N\!\left(
\bm z_k^{\mathsf O,m};
\bm H_k\bm r_k,
\bm R_k+
\mathtt{s}\bm S_k
\right)
\end{aligned}
\label{eq:Zk_likelihood}
\end{equation}

\begingroup
For the single-time \DOAM, the predictions for the kinematic
and SSAL states are given in
\cite[Eqs.~(5) and~(40)]{Tuncer-model-2021}, respectively.
Let $\omega_k$ and~$\pi_k$ denote predicted and posterior quantities,
respectively.
The scalar Beta--Bartlett evolution of the SSAL states
\cite[Eqs.~(4.9a) and~(4.9b)]{Kartal-2022-VS-ETT} is given by
\begin{equation}
a_{u,\omega_k}
=
\mathtt g a_{u,\pi_{k-1}}+(1-\mathtt g),
\qquad
u\in\{1,2\}
\label{eq:l_shape_predict}
\end{equation}
\begin{equation}
b_{u,\omega_k}
=
\mathtt g b_{u,\pi_{k-1}},
\qquad
u\in\{1,2\}
\label{eq:l_scale_predict}
\end{equation}

\noindent where~$\mathtt g\in(0,1)$ is the forgetting factor of the SSAL
states.

The nonconjugate update in \cite[Eqs.~(8)--(12)]{Tuncer-model-2021}
introduces $\mathcal Y_k^{\mathsf O}=\{\bm y_k^{\mathsf O,m}\}_{m=1}^{M_k^{\mathsf O}}$
and applies coordinate-ascent variational inference (CAVI) to the kinematic,
orientation, SSAL, and measurement-source factors. The closed forms are
given in \cite[Eqs.~(17), (21), (25), and~(32)]{Tuncer-model-2021}; only
the two expectations used below are retained.

For concise notation, the scale-weighted reciprocal expectation of the~$u$-th SSAL state at the~$\ell$-th variational iteration is defined as
\begin{equation}
\iota_{u,k}^{[\ell]}
=
\mathbb E_{q_k^{\ddot l_u,[\ell]}}
\!\left[
(\mathtt s\ddot l_{u,k})^{-1}
\right]
=
\frac{a_{u,\pi_k}}
{\mathtt s b_{u,\pi_k}^{[\ell]}},
\qquad
u\in\{1,2\}
\label{eq:iota_u}
\end{equation}
Furthermore, the expectation of the inverse extent-spread covariance at
the~$\ell$-th variational iteration is
\begin{equation}
\bar{\bm G}_k^{[\ell]}
=
\mathbb E_{q_k^{\theta,[\ell]}}
\!\left[
\mathcal R(\theta_k)
\operatorname{diag}\!\left(
\iota_{1,k}^{[\ell]},\iota_{2,k}^{[\ell]}
\right)
\mathcal R(-\theta_k)
\right]
\label{eq:Gbar_def}
\end{equation}
\noindent where the closed-form expression of
$\bar{\bm G}_k^{[\ell]}$ follows directly from
\cite[Corollary~1, Eq.~(36)]{Tuncer-model-2021}.

After convergence, the kinematic and orientation posteriors are Gaussian,
whereas the two SSAL posteriors are inverse-Gamma. Their approximate joint
posterior is
\begin{equation}
\resizebox{0.98\columnwidth}{!}{$
\begin{aligned}
q_k(\bm r_k,\theta_k,\ddot l_{1,k},\ddot l_{2,k})
&=
\mathcal N\!\left(
\bm r_k;
\hat{\bm r}_{\pi_k},
\bm\varXi_{\pi_k}^{\bm r}
\right)\\
&\quad\times
\mathcal N\!\left(
\theta_k;
\hat{\theta}_{\pi_k},
\varXi_{\pi_k}^{\theta}
\right)
\prod_{u=1}^{2}
\mathcal{IG}\!\left(
\ddot l_{u,k};
a_{u,\pi_k},
b_{u,\pi_k}
\right)
\end{aligned}
$}
\label{eq:DOAM_DGDIG_posterior}
\end{equation}

Since Eq.~\eqref{eq:DOAM_DGDIG_posterior} consists of two Gaussian factors and two inverse-Gamma factors, this distribution structure is referred to as the Double-Gaussian Double-Inverse-Gamma (DGDIG) density. The DGDIG density is a structured variational approximation to the nonconjugate joint posterior of the \DOAM.
\endgroup

\section{TCPHD-\DOAM Filtering Algorithm}
\label{C6_6}

This section presents the TCPHD-\DOAM filtering algorithm based on the DGDIG density obtained by \DOAM variational filtering and the TCPHD filtering framework. Unlike single-target \DOAM variational filtering, the single-time-step DGDIG density is extended to the trajectory space. Each trajectory component stores the historical statistical parameters of the kinematic, orientation, and two SSAL states, which are jointly corrected using the current measurements.

\subsection{Single-Trajectory DGDIG Representation}
\label{C6_6_1}

Consider a surviving trajectory at time~$k$ with start time~$t$ and length~$n$, where~$t+n-1=k$. The single-trajectory state is represented as
\begin{equation}
\mathcal X
=
\left(
t,\bm r^{1:n},\bm\theta^{1:n},
\ddot l_1^{1:n},\ddot l_2^{1:n}
\right)
\label{eq:C6_traj_state}
\end{equation}
\noindent where~$\bm r^{1:n}$ and~$\bm\theta^{1:n}$ denote the kinematic and orientation state sequences, respectively, and~$\ddot l_u^{1:n}$ denotes the $u$-th SSAL state sequence, where~$u\in\{1,2\}$.

After applying assumed-density projection to the structured variational
trajectory posterior, the trajectory-level extension of the single-time
DGDIG density is referred to as the trajectory DGDIG density. For compact
notation, let
$\bm\varOmega\triangleq(\hat{\bm\varGamma},
\bm\varXi^{\bm\varGamma},\hat{\bm\varTheta},
\bm\varXi^{\bm\varTheta},\bm A_1,\bm B_1,\bm A_2,\bm B_2)$
denote its statistical-parameter tuple. Any subscripts and superscripts of
$\bm\varOmega$ are inherited by all entries in this tuple. The trajectory
start time~$t$ is retained as a separate argument because it is not a
statistical parameter. The trajectory DGDIG density is defined as
\begin{equation}
\begin{aligned}
\mathcal{DGDIG}\!\left(\mathcal X;t,\bm\varOmega\right)
&\triangleq
\mathcal N\!\left(
\bm r^{1:n};
\hat{\bm\varGamma},
\bm\varXi^{\bm\varGamma}
\right)\\
&\quad\times
\mathcal N\!\left(
\bm\theta^{1:n};
\hat{\bm\varTheta},
\bm\varXi^{\bm\varTheta}
\right)\\
&\quad\times
\prod_{h=1}^{n}\prod_{u=1}^{2}
\mathcal{IG}\!\left(
\ddot l_u^h;a_u^h,b_u^h
\right)
\end{aligned}
\label{eq:C6_single_component}
\end{equation}

\noindent where~$\hat{\bm\varGamma}\in\mathbb R^{n d_{\bm r}}$
stacks the kinematic means as
$\hat{\bm\varGamma}=[(\hat{\bm r}^{1})^{\mathrm T},\ldots,
(\hat{\bm r}^{n})^{\mathrm T}]^{\mathrm T}$, with
$\hat{\bm r}^{h}=\mathbb E[\bm r^{h}]$. The covariance
$\bm\varXi^{\bm\varGamma}\in
\mathbb R^{n d_{\bm r}\times n d_{\bm r}}$ consists of blocks
$\bm\varXi_{h,h'}^{\bm\varGamma}
=\operatorname{Cov}(\bm r^{h},\bm r^{h'})$. Similarly,
$\hat{\bm\varTheta}=[\hat\theta^{1},\ldots,\hat\theta^{n}]^{\mathrm T}$
and~$\bm\varXi^{\bm\varTheta}\in\mathbb R^{n\times n}$ are the
orientation mean vector and covariance matrix, with
$\hat\theta^{h}=\mathbb E[\theta^{h}]$ and
$\varXi_{h,h'}^{\bm\varTheta}
=\operatorname{Cov}(\theta^{h},\theta^{h'})$. For~$u\in\{1,2\}$,
$\bm A_u=[a_u^1,\ldots,a_u^n]^{\mathrm T}$ and
$\bm B_u=[b_u^1,\ldots,b_u^n]^{\mathrm T}$ collect the shape and scale
parameters of the inverse-Gamma factors for the $u$-th SSAL state
sequence.

The product of inverse-Gamma factors in
Eq.~\eqref{eq:C6_single_component} is retained as an assumed-density
representation of the smoothed SSAL marginals. Before this projection,
the sequence-level evolution of each SSAL state used in the backward
smoothing step follows the scalar Beta--Bartlett Markov model
\cite[Eq.~(4.33)]{Kartal-2022-VS-ETT}, i.e.,
\begin{equation}
p\!\left(\ddot l_u^{1:n}\right)
=
p\!\left(\ddot l_u^1\right)
\prod_{h=2}^{n}
p\!\left(
\ddot l_u^h\mid\ddot l_u^{h-1};\mathtt g
\right),
\qquad u\in\{1,2\}
\label{eq:C6_SSAL_Markov_joint}
\end{equation}

After the backward recursion, the resulting inverse-Gamma marginals are
projected back onto the product form in
Eq.~\eqref{eq:C6_single_component}. Thus, the product representation does
not remove the Beta--Bartlett dependence used during smoothing.

The density in Eq.~\eqref{eq:C6_single_component} is nonzero only when the trajectory start time is~$t$ and~$n=\operatorname{Dim}(\hat{\bm\varGamma})/d_{\bm r}=\operatorname{Dim}(\hat{\bm\varTheta})$; otherwise, it is zero. Accordingly, the DGDIG-mixture form of the multitrajectory probability hypothesis density (PHD) at time~$k$ is
\begin{equation}
\begin{aligned}
\mathcal{D}_k(\mathcal X)
&=
\sum_{j=1}^{J_k}w_k^{(j)}
\mathcal{DGDIG}\!\left(
\mathcal X;t_k^{(j)},
\bm\varOmega_k^{(j)}
\right)
\end{aligned}
\label{eq:C6_tphd_mix}
\end{equation}
\noindent where~$J_k$ and~$w_k^{(j)}$ denote the number of trajectory
components and the weight of the $j$-th component, respectively. For a
surviving component,
$n_k^{(j)}=\operatorname{Dim}(\hat{\bm\varGamma}_k^{(j)})/d_{\bm r}
=\operatorname{Dim}(\hat{\bm\varTheta}_k^{(j)})$ and
$t_k^{(j)}+n_k^{(j)}-1=k$.

\subsection{TCPHD-\DOAM Prediction}
\label{C6_6_2}

Given the posterior multitrajectory PHD~$\mathcal{D}_{\pi_{k-1}}(\mathcal X)$ at time~$k-1$ and the birth multitrajectory PHD~$\mathcal{D}_{\beta_k}(\mathcal X)$ at time~$k$, both represented by the DGDIG mixture in Eq.~\eqref{eq:C6_tphd_mix}, the predicted multitrajectory PHD at time~$k$ is
\begin{equation}
\begin{aligned}
\mathcal{D}_{\omega_k}(\mathcal X)
=
\mathcal{D}_{\beta_k}(\mathcal X)+
\sum_{j=1}^{J_{\pi_{k-1}}}
p^{\mathsf S}w_{\pi_{k-1}}^{(j)}
\mathcal{DGDIG}\!\left(
\mathcal X;t_{\omega_k}^{(j)},
\bm\varOmega_{\omega_k}^{(j)}
\right)
\end{aligned}
\label{eq:C6_pred_mix}
\end{equation}
\noindent where~$p^{\mathsf S}\in[0,1]$ is the target survival probability. For the $j$-th surviving trajectory component, $w_{\omega_k}^{(j)}=p^{\mathsf S}w_{\pi_{k-1}}^{(j)}$, $t_{\omega_k}^{(j)}=t_{\pi_{k-1}}^{(j)}$, and~$n_{\omega_k}^{(j)}=n_{\pi_{k-1}}^{(j)}+1$. The predicted mean vector and covariance matrix of the kinematic state sequence are respectively calculated as
\begin{equation}
\hat{\bm\varGamma}_{\omega_k}^{(j)}
=
\begin{bmatrix}
(\hat{\bm\varGamma}_{\pi_{k-1}}^{(j)})^{\mathrm T}&
(\dot{\bm F}_k^{\bm\varGamma,(j)}
\hat{\bm\varGamma}_{\pi_{k-1}}^{(j)})^{\mathrm T}
\end{bmatrix}^{\mathrm T}
\label{eq:C6_r_mean_pred}
\end{equation}
\begin{equation}
\bm\varXi_{\omega_k}^{\bm\varGamma,(j)}
=
\begin{bmatrix}
\bm\varXi_{\pi_{k-1}}^{\bm\varGamma,(j)}
&
\bm\varXi_{\pi_{k-1}}^{\bm\varGamma,(j)}
(\dot{\bm F}_k^{\bm\varGamma,(j)})^{\mathrm T}\\
\dot{\bm F}_k^{\bm\varGamma,(j)}
\bm\varXi_{\pi_{k-1}}^{\bm\varGamma,(j)}
&
\dot{\bm F}_k^{\bm\varGamma,(j)}
\bm\varXi_{\pi_{k-1}}^{\bm\varGamma,(j)}
(\dot{\bm F}_k^{\bm\varGamma,(j)})^{\mathrm T}
+\bm Q_k^{\bm r}
\end{bmatrix}
\label{eq:C6_r_cov_pred}
\end{equation}
\begin{equation}
\dot{\bm F}_k^{\bm\varGamma,(j)}
=
\begin{bmatrix}
\bm 0_{d_{\bm r}\times(n_{\pi_{k-1}}^{(j)}-1)d_{\bm r}}
&
\bm F_k^{\bm r}
\end{bmatrix}
\label{eq:C6_terminal_transition1}
\end{equation}

The predicted mean vector and covariance matrix of the orientation state sequence are respectively calculated as
\begin{equation}
\hat{\bm\varTheta}_{\omega_k}^{(j)}
=
\begin{bmatrix}
(\hat{\bm\varTheta}_{\pi_{k-1}}^{(j)})^{\mathrm T}&
\dot{\bm F}_k^{\bm\varTheta,(j)}
\hat{\bm\varTheta}_{\pi_{k-1}}^{(j)}
\end{bmatrix}^{\mathrm T}
\label{eq:C6_theta_mean_pred}
\end{equation}
\begin{equation}
\bm\varXi_{\omega_k}^{\bm\varTheta,(j)}
=
\begin{bmatrix}
\bm\varXi_{\pi_{k-1}}^{\bm\varTheta,(j)}
&
\bm\varXi_{\pi_{k-1}}^{\bm\varTheta,(j)}
(\dot{\bm F}_k^{\bm\varTheta,(j)})^{\mathrm T}\\
\dot{\bm F}_k^{\bm\varTheta,(j)}
\bm\varXi_{\pi_{k-1}}^{\bm\varTheta,(j)}
&
\dot{\bm F}_k^{\bm\varTheta,(j)}
\bm\varXi_{\pi_{k-1}}^{\bm\varTheta,(j)}
(\dot{\bm F}_k^{\bm\varTheta,(j)})^{\mathrm T}
+Q_k^{\theta}
\end{bmatrix}
\label{eq:C6_theta_cov_pred}
\end{equation}
\begin{equation}
\dot{\bm F}_k^{\bm\varTheta,(j)}
=
\begin{bmatrix}
\bm 0_{1\times(n_{\pi_{k-1}}^{(j)}-1)}
&
F_k^{\theta}
\end{bmatrix}
\label{eq:C6_terminal_transition2}
\end{equation}

The predicted parameters of the SSAL state sequences are respectively calculated as
\begin{equation}
\bm A_{u,\omega_k}^{(j)}
=
\begin{bmatrix}
(\bm A_{u,\pi_{k-1}}^{(j)})^{\mathrm T}&
a_{u,\omega_k}^{(j),n_{\omega_k}^{(j)}}
\end{bmatrix}^{\mathrm T}
\label{eq:C6_a_pred_hist}
\end{equation}
\begin{equation}
a_{u,\omega_k}^{(j),n_{\omega_k}^{(j)}}
=
\mathtt g
a_{u,\pi_{k-1}}^{(j),n_{\pi_{k-1}}^{(j)}}
+(1-\mathtt g),
\qquad u\in\{1,2\}
\label{eq:C6_a_terminal_pred}
\end{equation}
\begin{equation}
\bm B_{u,\omega_k}^{(j)}
=
\begin{bmatrix}
(\bm B_{u,\pi_{k-1}}^{(j)})^{\mathrm T}&
b_{u,\omega_k}^{(j),n_{\omega_k}^{(j)}}
\end{bmatrix}^{\mathrm T}
\label{eq:C6_b_pred_hist}
\end{equation}
\begin{equation}
b_{u,\omega_k}^{(j),n_{\omega_k}^{(j)}}
=
\mathtt g
b_{u,\pi_{k-1}}^{(j),n_{\pi_{k-1}}^{(j)}},
\qquad u\in\{1,2\}
\label{eq:C6_b_terminal_pred}
\end{equation}

\subsection{TCPHD-\DOAM Update}
\label{C6_6_3}

Given the predicted multitrajectory PHD at time~$k$
\begin{equation}
\begin{aligned}
\mathcal{D}_{\omega_k}(\mathcal X)
&=
\sum_{j=1}^{J_{\omega_k}}w_{\omega_k}^{(j)}
\mathcal{DGDIG}\!\left(
\mathcal X;t_{\omega_k}^{(j)},
\bm\varOmega_{\omega_k}^{(j)}
\right)
\end{aligned}
\label{eq:C6_update_pred_phd}
\end{equation}
let~$\mathcal Z_k=\{\bm z_k^m\}_{m=1}^{M_k}$ denote the total measurement set acquired by the sensor at time~$k$, where~$M_k=|\mathcal Z_k|$ is the total number of measurements at time~$k$. The notation~$\mathcal P\angle\mathcal Z_k$ indicates that~$\mathcal P$ is a partition of~$\mathcal Z_k$, and~$\mathcal C\in\mathcal P$ denotes a measurement cell in this partition. Let~$M_k^{\mathcal P,\mathcal C}=|(\mathcal P,\mathcal C)|$ denote the number of measurements in this cell. The corresponding measurement set and measurement-source-variable set are respectively denoted by~$\mathcal Z_k^{(\mathcal P,\mathcal C)}=\{\bm z_k^{\mathcal P,\mathcal C,m}\}_{m=1}^{M_k^{\mathcal P,\mathcal C}}$ and~$\mathcal Y_k^{(\mathcal P,\mathcal C)}=\{\bm y_k^{\mathcal P,\mathcal C,m}\}_{m=1}^{M_k^{\mathcal P,\mathcal C}}$. The posterior multitrajectory PHD can then be written as
\begin{equation}
\begin{aligned}
\mathcal{D}_{\pi_k}(\mathcal X)
&=
\kappa\left[
1-(1-\mathrm e^{-\gamma^{\mathsf O}})p^{\mathsf D}
\right]
\mathcal{D}_{\omega_k}(\mathcal X)\\
&\quad+
\sum_{\mathcal P\angle\mathcal Z_k}
\sum_{\mathcal C\in\mathcal P}
\mathcal{D}_{\pi_k}^{\mathsf D}
(\mathcal X,\mathcal P,\mathcal C)
\end{aligned}
\label{eq:C6_post_phd_update}
\end{equation}
\begin{equation}
\begin{aligned}
&\mathcal{D}_{\pi_k}^{\mathsf D}
(\mathcal X,\mathcal P,\mathcal C)=
\sum_{j=1}^{J_{\omega_k}}
w_{\pi_k}^{(j,\mathcal P,\mathcal C)}
\mathcal{DGDIG}\!\left(
\mathcal X;t_{\omega_k}^{(j)},
\bm\varOmega_{\pi_k}^{(j,\mathcal P,\mathcal C)}
\right)
\end{aligned}
\label{eq:C6_DpiD_expand}
\end{equation}
\noindent where~$p^{\mathsf D}$ is the target detection probability, $\gamma^{\mathsf O}$ is the target measurement rate, and~$\kappa$ is the TCPHD update coefficient, whose calculation is given later with the cardinality-distribution recursion. For the $j$-th predicted trajectory component with length~$n_{\omega_k}^{(j)}$, the structured variational factorization is
\begin{equation}
\begin{aligned}
&q_k^{(j,\mathcal P,\mathcal C)}\!\left(
\bm r^{1:n_{\omega_k}^{(j)}},\bm\theta^{1:n_{\omega_k}^{(j)}},
\ddot l_1^{1:n_{\omega_k}^{(j)}},\ddot l_2^{1:n_{\omega_k}^{(j)}},
\mathcal Y_k^{(\mathcal P,\mathcal C)}
\right)\\
&=
q_k^{\bm r,(j,\mathcal P,\mathcal C)}
\!\left(\bm r^{1:n_{\omega_k}^{(j)}}\right)
\prod_{u=1}^{2}
q_k^{\ddot l_u,(j,\mathcal P,\mathcal C)}
\!\left(\ddot l_u^{1:n_{\omega_k}^{(j)}}\right)\\
&\quad\times
q_k^{\theta,(j,\mathcal P,\mathcal C)}
\!\left(\bm\theta^{1:n_{\omega_k}^{(j)}}\right)
\prod_{m=1}^{M_k^{\mathcal P,\mathcal C}}
q_k^{\bm y^m,(j,\mathcal P,\mathcal C)}
\!\left(\bm y_k^{\mathcal P,\mathcal C,m}\right)
\end{aligned}
\label{eq:C6_traj_vb_factor}
\end{equation}
\noindent where~$q_k^{\bm r,(j,\mathcal P,\mathcal C)}(\cdot)$, $q_k^{\ddot l_u,(j,\mathcal P,\mathcal C)}(\cdot)$, $q_k^{\theta,(j,\mathcal P,\mathcal C)}(\cdot)$, and~$q_k^{\bm y^m,(j,\mathcal P,\mathcal C)}(\cdot)$ are the variational posterior factors of the kinematic state sequence, the $u$-th SSAL state sequence, the orientation state sequence, and the $m$-th measurement-source variable in measurement cell~$(\mathcal P,\mathcal C)$, respectively.

To extract the terminal statistics of both predicted and posterior trajectory components in a unified form, define the terminal-selection vector~$\bm e_n=\begin{bmatrix}\bm 0_{1\times(n-1)}&1\end{bmatrix}^{\mathrm T}$. For a trajectory component of length~$n$, let~$\hat{\bm\varGamma}$, $\bm\varXi^{\bm\varGamma}$, $\hat{\bm\varTheta}$, and~$\bm\varXi^{\bm\varTheta}$ denote the statistical parameters of its kinematic and orientation state sequences, and let~$\bm A_u$ and~$\bm B_u$ denote the statistical parameters of its two SSAL state sequences. The corresponding terminal statistical parameters are
\begin{equation}
\hat{\bm r}^{n}
=
\left(
\bm e_n^{\mathrm T}\otimes\bm I_{d_{\bm r}}
\right)
\hat{\bm\varGamma}
\label{eq:C6_terminal_r_mean}
\end{equation}
\begin{equation}
\bm\varXi^{\bm r,n}
=
\left(
\bm e_n^{\mathrm T}\otimes\bm I_{d_{\bm r}}
\right)
\bm\varXi^{\bm\varGamma}
\left(
\bm e_n\otimes\bm I_{d_{\bm r}}
\right)
\label{eq:C6_terminal_r_statistics}
\end{equation}
\begin{equation}
\hat{\theta}^{n}
=
\bm e_n^{\mathrm T}\hat{\bm\varTheta}
\label{eq:C6_terminal_theta_mean}
\end{equation}
\begin{equation}
\varXi^{\theta,n}
=
\bm e_n^{\mathrm T}
\bm\varXi^{\bm\varTheta}
\bm e_n
\label{eq:C6_terminal_theta_statistics}
\end{equation}
\begin{equation}
a_u^{n}
=
\bm e_n^{\mathrm T}\bm A_u,
\qquad u\in\{1,2\}
\label{eq:C6_terminal_axis_shape}
\end{equation}
\begin{equation}
b_u^{n}
=
\bm e_n^{\mathrm T}\bm B_u,
\qquad u\in\{1,2\}
\label{eq:C6_terminal_axis_statistics}
\end{equation}

For the $j$-th predicted trajectory component, $n$ in Eqs.~\eqref{eq:C6_terminal_r_mean}--\eqref{eq:C6_terminal_axis_statistics} is set to~$n_{\omega_k}^{(j)}$, and all statistical parameters are set to the corresponding predicted parameters. For the posterior trajectory component obtained from this predicted component and measurement cell~$(\mathcal P,\mathcal C)$, the statistical parameters are set to the corresponding posterior parameters. The measurement update does not change the trajectory length, so that $n_{\pi_k}^{(j)}=n_{\omega_k}^{(j)}$. These terminal-extraction relations also apply to the variational posterior parameters at every iteration.

Based on the structured variational factorization and the terminal trajectory statistics defined above, the trajectory-level variational posterior of the $j$-th predicted trajectory component under measurement cell~$(\mathcal P,\mathcal C)$ can be recursively obtained using CAVI. Given the available results from the $\ell$-th iteration, the posterior kinematic state sequence, the two posterior SSAL state sequences, the posterior orientation state sequence, and the measurement-source variables are sequentially updated at the $(\ell+1)$-th iteration. The historical trajectory states are corrected while the terminal trajectory state is updated. The variational updates of the posterior kinematic state sequence, SSAL state sequences, orientation state sequence, and measurement-source variables are given in Propositions~\ref{prop:C6_r_update}--\ref{prop:C6_y_update}, respectively.

\begin{proposition}[Variational Update of the Kinematic State Sequence]
	\label{prop:C6_r_update}
	
	Given the $j$-th predicted trajectory component and measurement cell~$(\mathcal P,\mathcal C)$, suppose that the variational posteriors of the SSAL state sequences, orientation state sequence, and measurement-source variables at the $\ell$-th iteration are available. The posterior kinematic state sequence obtained at the $(\ell+1)$-th iteration remains Gaussian, i.e.,
	\begin{equation}
	\begin{aligned}
	&q_k^{\bm r,(j,\mathcal P,\mathcal C),[\ell+1]}
	\!\left(
	\bm r^{1:n_{\pi_k}^{(j)}}
	\right)\\
	&=
	\mathcal N\!\left(
	\bm r^{1:n_{\pi_k}^{(j)}};
	\hat{\bm\varGamma}_{\pi_k}^{(j,\mathcal P,\mathcal C),[\ell+1]},
	\bm\varXi_{\pi_k}^{\bm\varGamma,(j,\mathcal P,\mathcal C),[\ell+1]}
	\right)
	\end{aligned}
	\label{eq:C6_r_traj_posterior}
	\end{equation}
	
	The posterior mean vector and covariance matrix are respectively given by

	\begin{equation}
	\begin{aligned}
	&\hat{\bm\varGamma}_{\pi_k}^{(j,\mathcal P,\mathcal C),[\ell+1]}=
	\hat{\bm\varGamma}_{\omega_k}^{(j)}
	+
	\bm K_{\pi_k}^{(j,\mathcal P,\mathcal C),[\ell+1]}
	\left[
	\bar{\bm y}_{\pi_k}^{(j,\mathcal P,\mathcal C),[\ell]}
	-
	\dot{\bm H}_k^{(j)}
	\hat{\bm\varGamma}_{\omega_k}^{(j)}
	\right]
	\end{aligned}
	\label{eq:C6_r_traj_mean_update}
	\end{equation}
	\begin{equation}
	\begin{aligned}
	&\bm\varXi_{\pi_k}^{\bm\varGamma,(j,\mathcal P,\mathcal C),[\ell+1]}=
	\bm\varXi_{\omega_k}^{\bm\varGamma,(j)}
	-
	\bm K_{\pi_k}^{(j,\mathcal P,\mathcal C),[\ell+1]}
	\dot{\bm H}_k^{(j)}
	\bm\varXi_{\omega_k}^{\bm\varGamma,(j)}
	\end{aligned}
	\label{eq:C6_r_traj_cov_update}
	\end{equation}
	\noindent where
	\begin{equation}
	\dot{\bm H}_k^{(j)}
	=
	\bm H_k
	\left(
	\bm e_{n_{\omega_k}^{(j)}}^{\mathrm T}
	\otimes
	\bm I_{d_{\bm r}}
	\right)
	\label{eq:C6_terminal_measurement_matrix}
	\end{equation}
	\begin{equation}
	\begin{aligned}
	&\bm K_{\pi_k}^{(j,\mathcal P,\mathcal C),[\ell+1]}=
	\bm\varXi_{\omega_k}^{\bm\varGamma,(j)}
	(\dot{\bm H}_k^{(j)})^{\mathrm T}
	\left(
	\bm\varXi_{\pi_k}^{\bar{\bm y},(j,\mathcal P,\mathcal C),[\ell]}
	\right)^{-1}
	\end{aligned}
	\label{eq:C6_r_traj_gain}
	\end{equation}
	\begin{equation}
	\begin{aligned}
	&\bm\varXi_{\pi_k}^{\bar{\bm y},(j,\mathcal P,\mathcal C),[\ell]}=
	\dot{\bm H}_k^{(j)}
	\bm\varXi_{\omega_k}^{\bm\varGamma,(j)}
	(\dot{\bm H}_k^{(j)})^{\mathrm T}
	+
	\bm\varXi_{\pi_k}^{\bm\psi,(j,\mathcal P,\mathcal C),[\ell]}
	\end{aligned}
	\label{eq:C6_r_innovation_cov}
	\end{equation}
	\begin{equation}
	\begin{aligned}
	&\bar{\bm y}_{\pi_k}^{(j,\mathcal P,\mathcal C),[\ell]}
	=\frac{1}{M_k^{\mathcal P,\mathcal C}}
	\sum_{m=1}^{M_k^{\mathcal P,\mathcal C}}
	\hat{\bm y}_{\pi_k}^{(j,\mathcal P,\mathcal C),m,[\ell]}
	\end{aligned}
	\label{eq:C6_r_equivalent_measurement_mean}
	\end{equation}
	
	The equivalent measurement noise~$\bm\psi$ has covariance matrix
	\begin{equation}
	\bm\varXi_{\pi_k}^{\bm\psi,(j,\mathcal P,\mathcal C),[\ell]}
	=
	\left[
	M_k^{\mathcal P,\mathcal C}
	\bar{\bm G}_{\pi_k}^{(j,\mathcal P,\mathcal C),[\ell]}
	\right]^{-1}
	\label{eq:C6_r_equivalent_measurement_cov}
	\end{equation}
	\begin{equation}
	\begin{aligned}
	&\bar{\bm G}_{\pi_k}^{(j,\mathcal P,\mathcal C),[\ell]}\\
	&\quad=
	\mathbb E_{q_k^{\theta,(j,\mathcal P,\mathcal C),[\ell]}}
	\!\Bigg[
	\mathcal R\!\left(\theta^{n_{\omega_k}^{(j)}}\right)
	\operatorname{diag}\!\left(
	\begin{gathered}
	\iota_{1,\pi_k}^{(j,\mathcal P,\mathcal C),
		n_{\pi_k}^{(j)},[\ell]},\\
	\iota_{2,\pi_k}^{(j,\mathcal P,\mathcal C),
		n_{\pi_k}^{(j)},[\ell]}
	\end{gathered}
	\right)
	\mathcal R\!\left(-\theta^{n_{\omega_k}^{(j)}}\right)
	\Bigg]
	\end{aligned}
	\label{eq:C6_r_expected_precision}
	\end{equation}
	\begin{equation}
	\begin{aligned}
	\iota_{u,\pi_k}^{(j,\mathcal P,\mathcal C),n_{\pi_k}^{(j)},[\ell]}
	&=
	\mathbb E_{q_k^{\ddot l_u,(j,\mathcal P,\mathcal C),[\ell]}}
	\!\left[
	(\mathtt s\ddot l_u^{n_{\pi_k}^{(j)}})^{-1}
	\right]\\
	&\quad=
	\frac{
		a_{u,\pi_k}^{(j,\mathcal P,\mathcal C),n_{\pi_k}^{(j)}}
	}{
	\mathtt s
	b_{u,\pi_k}^{(j,\mathcal P,\mathcal C),n_{\pi_k}^{(j)},[\ell]}
},
\qquad u\in\{1,2\}
\end{aligned}
\label{eq:C6_axis_inv_expectation_l}
\end{equation}

Eq.~\eqref{eq:C6_r_expected_precision} is evaluated using the closed-form
result associated with Eq.~\eqref{eq:Gbar_def}. In this result,
$\iota_{u,k}^{[\ell]}$ is instantiated as
$\iota_{u,\pi_k}^{(j,\mathcal P,\mathcal C),n_{\pi_k}^{(j)},[\ell]}$, while the
mean and variance of the orientation factor are replaced by
$\hat\theta_{\pi_k}^{(j,\mathcal P,\mathcal C), n_{\pi_k}^{(j)},[\ell]}$ and $\varXi_{\pi_k}^{\theta,(j,\mathcal P,\mathcal C), n_{\pi_k}^{(j)},[\ell]}$, respectively.
The latter two quantities are extracted using
Eqs.~\eqref{eq:C6_terminal_theta_mean} and
\eqref{eq:C6_terminal_theta_statistics}.

The proof of \textbf{Proposition \ref{prop:C6_r_update}} is provided in Appendix~\ref{app:C6_r_update_proof}.

\end{proposition}

~\par

\begin{proposition}[Variational Update of the SSAL State Sequences]
	\label{prop:C6_axis_update}
	
	Given the $j$-th predicted trajectory component and measurement cell~$(\mathcal P,\mathcal C)$, suppose that the variational posterior of the kinematic state sequence at the $(\ell+1)$-th iteration and those of the orientation state sequence and measurement-source variables at the $\ell$-th iteration are available. The smoothed marginals of the two posterior SSAL state sequences at all historical times remain inverse-Gamma distributions, i.e.,
	\begin{equation}
	\begin{aligned}
	&q_k^{\ddot l_u,(j,\mathcal P,\mathcal C),[\ell+1]}
	\!\left(\ddot l_u^h\right)=
	\mathcal{IG}\!\left(
	\ddot l_u^h;
	a_{u,\pi_k}^{(j,\mathcal P,\mathcal C),h},
	b_{u,\pi_k}^{(j,\mathcal P,\mathcal C),h,[\ell+1]}
	\right)\\
	&\hspace{35mm}
	h=1,\ldots,n_{\pi_k}^{(j)},
	\qquad u\in\{1,2\}
	\end{aligned}
	\label{eq:C6_axis_marginal_posterior}
	\end{equation}
	
	The posterior shape-parameter vector and reciprocal scale-parameter vector at the $(\ell+1)$-th iteration are respectively given by
	\begin{equation}
	\begin{aligned}
	&\bm A_{u,\pi_k}^{(j,\mathcal P,\mathcal C)}=
	\bm T_{n_{\pi_k}^{(j)}}^{-1}(\mathtt g)
	\left[
	\bm W_{n_{\pi_k}^{(j)}}(\mathtt g)
	\bm U_{n_{\pi_k}^{(j)}}(\mathtt g)
	\bm A_{u,\omega_k}^{(j)}
	+\frac{M_k^{\mathcal P,\mathcal C}}{2}
	\bm e_{n_{\pi_k}^{(j)}}
	\right]\\
	&\hspace{75mm} u\in\{1,2\}
	\end{aligned}
	\label{eq:C6_axis_matrix_a}
	\end{equation}
	\begin{equation}
	\begin{aligned}
	\breve{\bm B}_{u,\pi_k}^{(j,\mathcal P,\mathcal C),[\ell+1]}&=
	\bm T_{n_{\pi_k}^{(j)}}^{-1}(\mathtt g)
	\left\{
	(1-\mathtt g)
	\bm\varPi_{n_{\pi_k}^{(j)}}
	\bm U_{n_{\pi_k}^{(j)}}(\mathtt g)
	\breve{\bm B}_{u,\omega_k}^{(j)}
	\right.\\
	&\left.
	\quad+
	\left[
	b_{u,\omega_k}^{(j),n_{\omega_k}^{(j)}}
	+
	\left[
	\widetilde{\bm C}_k^{(j,\mathcal P,\mathcal C),[\ell+1]}
	\right]_{uu}
	\right]^{-1}
	\bm e_{n_{\pi_k}^{(j)}}
	\right\}\\
	&\hspace{55mm} u\in\{1,2\}
	\end{aligned}
	\label{eq:C6_axis_matrix_eta}
	\end{equation}
	\noindent where
	\begin{equation}
	\begin{aligned}
	\breve{\bm B}_{u,\omega_k}^{(j)}
	=
	\begin{bmatrix}
	(b_{u,\omega_k}^{(j),1})^{-1}&
	\cdots&
	(b_{u,\omega_k}^{(j),n_{\omega_k}^{(j)}})^{-1}
	\end{bmatrix}^{\mathrm T},
	\quad u\in\{1,2\}
	\end{aligned}
	\label{eq:C6_axis_pred_invscale}
	\end{equation}
	\begin{equation}
		\begin{aligned}
			\widetilde{\bm C}_k^{(j,\mathcal P,\mathcal C),[\ell+1]}
			=
			&\mathbb E_{q_k^{\theta,(j,\mathcal P,\mathcal C),[\ell]}}\Bigl[
			\mathcal R\!\left(-\theta^{n_{\omega_k}^{(j)}}\right)\\[-1mm]
			&\qquad{}\times
			\bar{\bm C}_k^{(j,\mathcal P,\mathcal C),[\ell+1]}
			\mathcal R\!\left(\theta^{n_{\omega_k}^{(j)}}\right)
			\Bigr]
		\end{aligned}
		\label{eq:C6_axis_Bbar_update}
	\end{equation}
	\begin{equation}
	\begin{aligned}
	\bar{\bm C}_k^{(j,\mathcal P,\mathcal C),[\ell+1]}
	=
	\frac{1}{2\mathtt s}
	\sum_{m=1}^{M_k^{\mathcal P,\mathcal C}}
	\bm C_k^{(j,\mathcal P,\mathcal C),m,[\ell+1]}
	\end{aligned}
	\label{eq:C6_axis_Cbar_update}
	\end{equation}
	\begin{equation}
	\begin{aligned}
	&\bm C_k^{(j,\mathcal P,\mathcal C),m,[\ell+1]}=
	\left(
	\hat{\bm y}_{\pi_k}^{(j,\mathcal P,\mathcal C),m,[\ell]}
	-
	\dot{\bm H}_k^{(j)}
	\hat{\bm\varGamma}_{\pi_k}^{(j,\mathcal P,\mathcal C),[\ell+1]}
	\right)\\
	&\quad\times
	\left(
	\hat{\bm y}_{\pi_k}^{(j,\mathcal P,\mathcal C),m,[\ell]}
	-
	\dot{\bm H}_k^{(j)}
	\hat{\bm\varGamma}_{\pi_k}^{(j,\mathcal P,\mathcal C),[\ell+1]}
	\right)^{\mathrm T}\\
	&\quad+
	\dot{\bm H}_k^{(j)}
	\bm\varXi_{\pi_k}^{\bm\varGamma,(j,\mathcal P,\mathcal C),[\ell+1]}
	(\dot{\bm H}_k^{(j)})^{\mathrm T}
	+
	\bm\varXi_{\pi_k}^{\bm y,(j,\mathcal P,\mathcal C),m,[\ell]}
	\end{aligned}
	\label{eq:C6_axis_Cm_update}
	\end{equation}
	
	The posterior scale-parameter vector~$\bm B_{u,\pi_k}^{(j,\mathcal P,\mathcal C),[\ell+1]}$ is obtained by taking the elementwise reciprocal of~$\breve{\bm B}_{u,\pi_k}^{(j,\mathcal P,\mathcal C),[\ell+1]}$.
	
	In addition, the matrices~$\bm T_n(\mathtt g)$, $\bm W_n(\mathtt g)$,
	$\bm\varPi_n$, $\bm J_n$, and~$\bm U_n(\mathtt g)$ are auxiliary
	matrices used in Eqs.~\eqref{eq:C6_axis_matrix_a} and
	\eqref{eq:C6_axis_matrix_eta}. Their required instances are obtained
	by setting~$n=n_{\pi_k}^{(j)}$ in the following definitions:
	
	\begin{equation}
	\bm T_n(\mathtt g)
	=
	\bm I_n-\mathtt g\bm J_n
	\label{eq:C6_axis_transport_operator}
	\end{equation}
	\begin{equation}
	\bm W_n(\mathtt g)
	=
	(1-\mathtt g)\bm\varPi_n
	+\bm e_n\bm e_n^{\mathrm T}
	\label{eq:C6_axis_weight_operator}
	\end{equation}
	\begin{equation}
	\bm\varPi_n
	=
	\bm I_n-\bm e_n\bm e_n^{\mathrm T}
	\label{eq:C6_axis_projection_operator}
	\end{equation}
	\begin{equation}
	\bm J_n
	=
	\begin{cases}
	\bm 0_{1\times1}, & n=1\\
	\begin{bmatrix}
	\bm 0_{(n-1)\times1} & \bm I_{n-1}\\
	\bm 0_{1\times1} & \bm 0_{1\times(n-1)}
	\end{bmatrix}, & n\geq2
	\end{cases}
	\label{eq:C6_axis_shift_operator}
	\end{equation}
	\begin{equation}
	\bm U_n(\mathtt g)
	=
	\begin{cases}
	\bm I_1, & n=1\\
	\operatorname{blkdiag}\!\left(
	\bm W_{n-1}^{-1}(\mathtt g)
	\bm T_{n-1}(\mathtt g),1
	\right), & n\geq2
	\end{cases}
	\label{eq:C6_axis_recovery_operator}
	\end{equation}
	
	The proof of \textbf{Proposition \ref{prop:C6_axis_update}} is provided in Appendix~\ref{app:C6_axis_update_proof}.
	
\end{proposition}
~\par

\begin{proposition}[Variational Update of the Orientation State Sequence]
	\label{prop:C6_theta_update}
	
	Given the $j$-th predicted trajectory component and measurement
	cell~$(\mathcal P,\mathcal C)$, suppose that the variational posteriors of
	the kinematic state sequence and the two SSAL state sequences at the
	$(\ell+1)$-th iteration and those of the measurement-source variables at
	the $\ell$-th iteration are available.
	The following update uses a first-order expansion of the
	rotation matrix about the current orientation estimate. If
	$\zeta_{\pi_k}^{(j,\mathcal P,\mathcal C),[\ell+1]}>0$, the variational
	posterior of the orientation state sequence obtained at the
	$(\ell+1)$-th iteration is approximated as Gaussian, i.e.,
	\begin{equation}
	\begin{aligned}
	&q_k^{\theta,(j,\mathcal P,\mathcal C),[\ell+1]}
	\!\left(\bm\theta^{1:n_{\pi_k}^{(j)}}\right)\\
	&=
	\mathcal N\!\left(
	\bm\theta^{1:n_{\pi_k}^{(j)}};
	\hat{\bm\varTheta}_{\pi_k}^{(j,\mathcal P,\mathcal C),[\ell+1]},
	\bm\varXi_{\pi_k}^{\bm\varTheta,(j,\mathcal P,\mathcal C),[\ell+1]}
	\right)
	\end{aligned}
	\label{eq:C6_theta_traj_posterior}
	\end{equation}
	
	The posterior mean vector and covariance matrix are respectively given by
	\begin{equation}
	\begin{aligned}
	&\hat{\bm\varTheta}_{\pi_k}^{(j,\mathcal P,\mathcal C),[\ell+1]}\\
	&=
	\hat{\bm\varTheta}_{\omega_k}^{(j)}
	+\bm\varPsi_{\pi_k}^{(j,\mathcal P,\mathcal C),[\ell+1]}
	\left[
	\widetilde\theta_{\pi_k}^{(j,\mathcal P,\mathcal C),[\ell+1]}
	-\bm e_{n_{\omega_k}^{(j)}}^{\mathrm T}
	\hat{\bm\varTheta}_{\omega_k}^{(j)}
	\right]
	\end{aligned}
	\label{eq:C6_theta_traj_mean_update}
	\end{equation}
	\begin{equation}
	\begin{aligned}
	\bm\varXi_{\pi_k}^{\bm\varTheta,(j,\mathcal P,\mathcal C),[\ell+1]}=
	\bm\varXi_{\omega_k}^{\bm\varTheta,(j)}
	-\bm\varPsi_{\pi_k}^{(j,\mathcal P,\mathcal C),[\ell+1]}
	\bm e_{n_{\omega_k}^{(j)}}^{\mathrm T}
	\bm\varXi_{\omega_k}^{\bm\varTheta,(j)}
	\end{aligned}
	\label{eq:C6_theta_traj_cov_update}
	\end{equation}
	\noindent where
	\begin{equation}
	\begin{aligned}
	\bm\varPsi_{\pi_k}^{(j,\mathcal P,\mathcal C),[\ell+1]}=
	\bm\varXi_{\omega_k}^{\bm\varTheta,(j)}
	\bm e_{n_{\omega_k}^{(j)}}
	\left(
	\varXi_{\pi_k}^{\widetilde\theta,
		(j,\mathcal P,\mathcal C),[\ell+1]}
	\right)^{-1}
	\end{aligned}
	\label{eq:C6_theta_traj_gain}
	\end{equation}
	\begin{equation}
	\begin{aligned}
	\varXi_{\pi_k}^{\widetilde\theta,
		(j,\mathcal P,\mathcal C),[\ell+1]}=
	\bm e_{n_{\omega_k}^{(j)}}^{\mathrm T}
	\bm\varXi_{\omega_k}^{\bm\varTheta,(j)}
	\bm e_{n_{\omega_k}^{(j)}}
	+\left(
	\zeta_{\pi_k}^{(j,\mathcal P,\mathcal C),[\ell+1]}
	\right)^{-1}
	\end{aligned}
	\label{eq:C6_theta_innovation_variance}
	\end{equation}
	\begin{equation}
	\widetilde\theta_{\pi_k}^{(j,\mathcal P,\mathcal C),[\ell+1]}
	=
	\frac{
		\eta_{\pi_k}^{(j,\mathcal P,\mathcal C),[\ell+1]}
	}{
	\zeta_{\pi_k}^{(j,\mathcal P,\mathcal C),[\ell+1]}
}
\label{eq:C6_theta_equivalent_measurement}
\end{equation}
\begin{equation}
\resizebox{0.98\columnwidth}{!}{$
	\begin{aligned}
	\zeta_{\pi_k}^{(j,\mathcal P,\mathcal C),[\ell+1]}
	&=
	\sum_{m=1}^{M_k^{\mathcal P,\mathcal C}}
	\operatorname{tr}\!\Bigg[
	\bm\varPhi_{\pi_k}^{(j,\mathcal P,\mathcal C),[\ell+1]}
	\left[
	-\dot{\mathcal R}\!\left(
	-\hat\theta_{\pi_k}^{(j,\mathcal P,\mathcal C),
		n_{\pi_k}^{(j)},[\ell]}
	\right)
	\right]\\
	&\quad\times
	\bm C_k^{(j,\mathcal P,\mathcal C),m,[\ell+1]}
	\dot{\mathcal R}\!\left(
	\hat\theta_{\pi_k}^{(j,\mathcal P,\mathcal C),
		n_{\pi_k}^{(j)},[\ell]}
	\right)
	\Bigg]
	\end{aligned}
	$}
\label{eq:C6_theta_precision}
\end{equation}
\begin{equation}
\begin{aligned}
\eta_{\pi_k}^{(j,\mathcal P,\mathcal C),[\ell+1]}
&=
\hat\theta_{\pi_k}^{(j,\mathcal P,\mathcal C),
	n_{\pi_k}^{(j)},[\ell]}
\zeta_{\pi_k}^{(j,\mathcal P,\mathcal C),[\ell+1]}\\
&\quad-
\sum_{m=1}^{M_k^{\mathcal P,\mathcal C}}
\operatorname{tr}\!\Bigg[
\bm\varPhi_{\pi_k}^{(j,\mathcal P,\mathcal C),[\ell+1]}
\mathcal R\!\left(
-\hat\theta_{\pi_k}^{(j,\mathcal P,\mathcal C),
	n_{\pi_k}^{(j)},[\ell]}
\right)\\
&\quad\times
\bm C_k^{(j,\mathcal P,\mathcal C),m,[\ell+1]}
\dot{\mathcal R}\!\left(
\hat\theta_{\pi_k}^{(j,\mathcal P,\mathcal C),
	n_{\pi_k}^{(j)},[\ell]}
\right)
\Bigg]
\end{aligned}
\label{eq:C6_theta_linear_parameter}
\end{equation}
\begin{equation}
\begin{aligned}
&\bm\varPhi_{\pi_k}^{(j,\mathcal P,\mathcal C),[\ell+1]}=\operatorname{diag}\!\left(
\iota_{1,\pi_k}^{(j,\mathcal P,\mathcal C),
	n_{\pi_k}^{(j)},[\ell+1]},
\iota_{2,\pi_k}^{(j,\mathcal P,\mathcal C),
	n_{\pi_k}^{(j)},[\ell+1]}
\right)
\end{aligned}
\label{eq:C6_theta_axis_precision}
\end{equation}
\begin{equation}
\dot{\mathcal R}(\theta)
=
\frac{\partial\mathcal R(\theta)}{\partial\theta}
=
\begin{bmatrix}
-\sin\theta&-\cos\theta\\
\cos\theta&-\sin\theta
\end{bmatrix}
\label{eq:C6_rotation_derivative}
\end{equation}

If~$\zeta_{\pi_k}^{(j,\mathcal P,\mathcal C),[\ell+1]}\leq0$,
the variational posterior parameters of the orientation state sequence
from the preceding iteration remain unchanged.

The proof of \textbf{Proposition \ref{prop:C6_theta_update}} is provided in Appendix~\ref{app:C6_theta_update_proof}.

\end{proposition}
~\par

\begin{proposition}[Variational Update of the Measurement-Source Variables]
	\label{prop:C6_y_update}
	
	Given the $j$-th predicted trajectory component and measurement
	cell~$(\mathcal P,\mathcal C)$, suppose that the variational posteriors of
	the kinematic, orientation, and two SSAL state sequences at the
	$(\ell+1)$-th iteration are available. The variational posteriors of the
	measurement-source variables in the cell are mutually independent and
	remain Gaussian, i.e.,
	
	\begin{equation}
	\begin{aligned}
	&q_k^{\bm y^m,(j,\mathcal P,\mathcal C),[\ell+1]}
	\!\left(\bm y_k^{\mathcal P,\mathcal C,m}\right)\\
	&=
	\mathcal N\!\left(
	\bm y_k^{\mathcal P,\mathcal C,m};
	\hat{\bm y}_{\pi_k}^{(j,\mathcal P,\mathcal C),m,[\ell+1]},
	\bm\varXi_{\pi_k}^{\bm y,(j,\mathcal P,\mathcal C),m,[\ell+1]}
	\right)\\
	&\hspace{45mm}
	m=1,\ldots,M_k^{\mathcal P,\mathcal C}
	\end{aligned}
	\label{eq:C6_y_posterior}
	\end{equation}
	
	Their posterior mean vectors and covariance matrices are respectively
	given by
	\begin{equation}
	\begin{aligned}
	&\hat{\bm y}_{\pi_k}^{(j,\mathcal P,\mathcal C),m,[\ell+1]}\\
	&=
	\bm\varXi_{\pi_k}^{\bm y,(j,\mathcal P,\mathcal C),m,[\ell+1]}
	\left[
	\bm R_k^{-1}\bm z_k^{\mathcal P,\mathcal C,m}
	\right.\\
	&\hspace{15mm}\left.
	+\bar{\bm G}_{\pi_k}^{(j,\mathcal P,\mathcal C),[\ell+1]}
	\dot{\bm H}_k^{(j)}
	\hat{\bm\varGamma}_{\pi_k}^{(j,\mathcal P,\mathcal C),[\ell+1]}
	\right]
	\end{aligned}
	\label{eq:C6_y_mean_update}
	\end{equation}
	\begin{equation}
	\begin{aligned}
	\bm\varXi_{\pi_k}^{\bm y,
		(j,\mathcal P,\mathcal C),m,[\ell+1]}=
	\left[
	\bm R_k^{-1}
	+\bar{\bm G}_{\pi_k}^{(j,\mathcal P,\mathcal C),[\ell+1]}
	\right]^{-1}
	\end{aligned}
	\label{eq:C6_y_cov_update}
	\end{equation}
	
	The proof of \textbf{Proposition \ref{prop:C6_y_update}} is provided in Appendix~\ref{app:C6_y_update_proof}.
	
\end{proposition}

~\par

\textbf{Propositions \ref{prop:C6_r_update}--\ref{prop:C6_y_update}} jointly
constitute one complete CAVI iteration for the $j$-th predicted
trajectory component under measurement cell~$(\mathcal P,\mathcal C)$.
To determine whether the variational iteration has converged, the
maximum relative change in the variational posterior statistical
parameters between two consecutive iterations is adopted as the
variational Bayes (VB) convergence criterion. Let~$\{\bm\chi_s^{[\ell]}\}_{s=1}^{N_q}$ denote
the parameter blocks formed by vectorizing the variational posterior
statistical parameters at the $\ell$-th iteration, and define
\begin{equation}
E_{\mathsf{VB}}^{[\ell]}
=
\max_{s=1,\ldots,N_q}
\frac{
	\left\|
	\bm\chi_s^{[\ell]}-\bm\chi_s^{[\ell-1]}
	\right\|_2
}{
\max\!\left(
\left\|\bm\chi_s^{[\ell-1]}\right\|_2,
\mathtt c_{\slashed{0}}
\right)
}
\label{eq:C6_VB_convergence}
\end{equation}
\noindent where~$N_q$ is the number of statistical-parameter blocks, and
$\mathtt c_{\slashed{0}}$ is a positive constant introduced to prevent
a zero denominator. The CAVI recursion terminates when the number of
iterations reaches the prescribed maximum~$\mathtt N_{\max}^{\mathsf{VB}}$
or when~$E_{\mathsf{VB}}^{[\ell]}\leq\mathtt T_{\mathsf{VB}}$, where
$\mathtt T_{\mathsf{VB}}$ is the convergence threshold. Substituting the
statistical parameters obtained at termination into
Eq.~\eqref{eq:C6_DpiD_expand} determines the DGDIG state-density
parameters of the posterior trajectory component associated with
measurement cell~$(\mathcal P,\mathcal C)$. The CAVI recursion for a
predicted-trajectory-component--measurement-cell pair is summarized in
Algorithm~\ref{alg:C6_DOAM_CAVI}.

\begin{algorithm}[!t]
	\caption{CAVI Recursion for a Predicted Trajectory
		Component--Measurement-Cell Pair}
	\label{alg:C6_DOAM_CAVI}
	\renewcommand{\algorithmicrequire}{\textbf{Input:}}
	\renewcommand{\algorithmicensure}{\textbf{Output:}}
	\begin{algorithmic}[1]
		\REQUIRE DGDIG state-density parameters of the $j$-th predicted
		trajectory component; measurement cell~$(\mathcal P,\mathcal C)$;
		maximum number of variational iterations~$\mathtt N_{\max}^{\mathsf{VB}}$;
		convergence threshold~$\mathtt T_{\mathsf{VB}}$
		\ENSURE DGDIG state-density parameters of the corresponding posterior
		trajectory component
		\STATE Set~$\ell\gets0$
		\STATE $\hat{\bm\varGamma}_{\pi_k}^{(j,\mathcal P,\mathcal C),[0]}
		\gets\hat{\bm\varGamma}_{\omega_k}^{(j)}$
		\STATE $\bm\varXi_{\pi_k}^{\bm\varGamma,
			(j,\mathcal P,\mathcal C),[0]}
		\gets\bm\varXi_{\omega_k}^{\bm\varGamma,(j)}$
		\STATE $\hat{\bm\varTheta}_{\pi_k}^{(j,\mathcal P,\mathcal C),[0]}
		\gets\hat{\bm\varTheta}_{\omega_k}^{(j)}$
		\STATE $\bm\varXi_{\pi_k}^{\bm\varTheta,
			(j,\mathcal P,\mathcal C),[0]}
		\gets\bm\varXi_{\omega_k}^{\bm\varTheta,(j)}$
		\STATE $\bm B_{u,\pi_k}^{(j,\mathcal P,\mathcal C),[0]}
		\gets\bm B_{u,\omega_k}^{(j)}$, $u\in\{1,2\}$
		\FOR{$m=1,\ldots,M_k^{\mathcal P,\mathcal C}$}
		\STATE $\hat{\bm y}_{\pi_k}^{(j,\mathcal P,\mathcal C),m,[0]}
		\gets\bm z_k^{\mathcal P,\mathcal C,m}$
		\STATE $\bm\varXi_{\pi_k}^{\bm y,
			(j,\mathcal P,\mathcal C),m,[0]}
		\gets\bm R_k$
		\ENDFOR
		\STATE Compute~$\bm A_{u,\pi_k}^{(j,\mathcal P,\mathcal C)}$ using
		Proposition~\ref{prop:C6_axis_update} and
		Eq.~\eqref{eq:C6_axis_matrix_a}, $u\in\{1,2\}$
		\REPEAT
		\STATE Update the kinematic state sequence using
		Proposition~\ref{prop:C6_r_update} and
		Eqs.~\eqref{eq:C6_r_traj_mean_update} and
		\eqref{eq:C6_r_traj_cov_update}
		\STATE Update~$\bm B_{u,\pi_k}^{(j,\mathcal P,\mathcal C),[\ell+1]}$
		using Proposition~\ref{prop:C6_axis_update} and
		Eq.~\eqref{eq:C6_axis_matrix_eta}, $u\in\{1,2\}$
		\STATE Update the orientation state sequence using
		Proposition~\ref{prop:C6_theta_update} and
		Eqs.~\eqref{eq:C6_theta_traj_mean_update} and
		\eqref{eq:C6_theta_traj_cov_update}
		\FOR{$m=1,\ldots,M_k^{\mathcal P,\mathcal C}$}
		\STATE Update the $m$-th measurement-source variable using
		Proposition~\ref{prop:C6_y_update} and
		Eqs.~\eqref{eq:C6_y_mean_update} and~\eqref{eq:C6_y_cov_update}
		\ENDFOR
		\STATE Compute~$E_{\mathsf{VB}}^{[\ell+1]}$ using
		Eq.~\eqref{eq:C6_VB_convergence}
		\STATE Set~$\ell\gets\ell+1$
		\UNTIL{$\ell=\mathtt N_{\max}^{\mathsf{VB}}$ or
			$E_{\mathsf{VB}}^{[\ell]}\leq\mathtt T_{\mathsf{VB}}$}
	\end{algorithmic}
\end{algorithm}

Furthermore, the weight~$w_{\pi_k}^{(j,\mathcal P,\mathcal C)}$ of each
posterior trajectory component is given by
Proposition~\ref{prop:C6_component_weight}.

\begin{proposition}[Trajectory-Component Weight Update]
	\label{prop:C6_component_weight}
	
	Given the $j$-th predicted trajectory component and measurement
	cell~$(\mathcal P,\mathcal C)$, the predicted DGDIG state density at the
	terminal trajectory time is
	\begin{equation}
	\begin{aligned}
	&p_{\omega_k}^{(j),n_{\omega_k}^{(j)}}
	\!\left(
	\bm r_k,\theta_k,\ddot l_{1,k},\ddot l_{2,k}
	\right)\\
	&=
	\mathcal N\!\left(
	\bm r_k;
	\hat{\bm r}_{\omega_k}^{(j),n_{\omega_k}^{(j)}},
	\bm\varXi_{\omega_k}^{\bm r,(j),n_{\omega_k}^{(j)}}
	\right)
	\mathcal N\!\left(
	\theta_k;
	\hat\theta_{\omega_k}^{(j),n_{\omega_k}^{(j)}},
	\varXi_{\omega_k}^{\theta,(j),n_{\omega_k}^{(j)}}
	\right)\\
	&\quad\times
	\prod_{u=1}^{2}
	\mathcal{IG}\!\left(
	\ddot l_{u,k};
	a_{u,\omega_k}^{(j),n_{\omega_k}^{(j)}},
	b_{u,\omega_k}^{(j),n_{\omega_k}^{(j)}}
	\right)
	\end{aligned}
	\label{eq:C6_terminal_pred_density}
	\end{equation}
	
	The terminal kinematic and
	orientation statistics are obtained using
	Eqs.~\eqref{eq:C6_terminal_r_mean}--\eqref{eq:C6_terminal_theta_statistics}.
	The quantities~$a_{u,\omega_k}^{(j),n_{\omega_k}^{(j)}}$ and
	$b_{u,\omega_k}^{(j),n_{\omega_k}^{(j)}}$ are respectively the terminal
	shape and scale parameters of the $u$-th predicted SSAL state sequence.
	
	The measurement likelihood of this predicted trajectory component with respect to measurement cell~$(\mathcal P,\mathcal C)$ is approximated by
	\begin{equation}
	\begin{aligned}
	&\mathcal L_{\omega_k}^{(j,\mathcal P,\mathcal C)}\approx
	\prod_{m=1}^{M_k^{\mathcal P,\mathcal C}}
	\mathcal N\!\left(
	\bm z_k^{\mathcal P,\mathcal C,m};
	\bm H_k\hat{\bm r}_{\omega_k}^{(j),n_{\omega_k}^{(j)}},
	\bm\varXi_{\omega_k}^{\bm z,(j),n_{\omega_k}^{(j)}}
	\right)
	\end{aligned}
	\label{eq:C6_DOAM_component_likelihood}
	\end{equation}
	\noindent where
	\begin{equation}
	\begin{aligned}
	\bm\varXi_{\omega_k}^{\bm z,(j),n_{\omega_k}^{(j)}}
	=\bm H_k
	\bm\varXi_{\omega_k}^{\bm r,(j),n_{\omega_k}^{(j)}}
	\bm H_k^{\mathrm T}+
	\bar{\bm S}_{\omega_k}^{(j),n_{\omega_k}^{(j)}}
	+\bm R_k
	\end{aligned}
	\label{eq:C6_DOAM_measurement_cov}
	\end{equation}
	\begin{equation}
	\begin{aligned}
	\overline{\ddot l}_{u,\omega_k}^{(j),n_{\omega_k}^{(j)}}
	&=
	\mathbb E_{p_{\omega_k}^{(j),n_{\omega_k}^{(j)}}}
	\!\left[\ddot l_{u,k}\right]=
	\frac{
		b_{u,\omega_k}^{(j),n_{\omega_k}^{(j)}}
	}{
	a_{u,\omega_k}^{(j),n_{\omega_k}^{(j)}}-1
},
\qquad u\in\{1,2\}
\end{aligned}
\label{eq:C6_DOAM_axis_prediction_mean}
\end{equation}
\begin{equation}
	\begin{aligned}
		\bar{\bm S}_{\omega_k}^{(j),n_{\omega_k}^{(j)}}
		&=
		\mathtt s\,
		\mathbb E_{p_{\omega_k}^{(j),n_{\omega_k}^{(j)}}}
		\!\left[
		\mathcal R\!\left(
		\theta^{n_{\omega_k}^{(j)}}
		\right)
		\mathcal W\!\left(
		\ddot l_1^{n_{\omega_k}^{(j)}},
		\ddot l_2^{n_{\omega_k}^{(j)}}
		\right)
		\mathcal R\!\left(
		-\theta^{n_{\omega_k}^{(j)}}
		\right)
		\right]
	\end{aligned}
	\label{eq:C6_DOAM_expected_extent_def}
\end{equation}

For compactness, define
\begin{equation}
\overline{\ddot l}_{\Sigma,\omega_k}^{(j),n_{\omega_k}^{(j)}}
=
\overline{\ddot l}_{1,\omega_k}^{(j),n_{\omega_k}^{(j)}}
+
\overline{\ddot l}_{2,\omega_k}^{(j),n_{\omega_k}^{(j)}}
\label{eq:C6_DOAM_axis_sum}
\end{equation}
\begin{equation}
\overline{\ddot l}_{\Delta,\omega_k}^{(j),n_{\omega_k}^{(j)}}
=
\overline{\ddot l}_{1,\omega_k}^{(j),n_{\omega_k}^{(j)}}
-
\overline{\ddot l}_{2,\omega_k}^{(j),n_{\omega_k}^{(j)}}
\label{eq:C6_DOAM_axis_difference}
\end{equation}
\begin{equation}
\varrho_{c,\omega_k}^{(j),n_{\omega_k}^{(j)}}
=
\cos\!\left(
2\hat\theta_{\omega_k}^{(j),n_{\omega_k}^{(j)}}
\right)
\mathrm e^{-2\varXi_{\omega_k}^{\theta,(j),n_{\omega_k}^{(j)}}}
\label{eq:C6_DOAM_orientation_cosine_moment}
\end{equation}
\begin{equation}
\varrho_{s,\omega_k}^{(j),n_{\omega_k}^{(j)}}
=
\sin\!\left(
2\hat\theta_{\omega_k}^{(j),n_{\omega_k}^{(j)}}
\right)
\mathrm e^{-2\varXi_{\omega_k}^{\theta,(j),n_{\omega_k}^{(j)}}}
\label{eq:C6_DOAM_orientation_sine_moment}
\end{equation}

Then the expected extent-spread covariance matrix is
\begin{equation}
\begin{aligned}
&\bar{\bm S}_{\omega_k}^{(j),n_{\omega_k}^{(j)}}\\
&=
\frac{\mathtt s}{2}
\begin{bmatrix}
\overline{\ddot l}_{\Sigma,\omega_k}^{(j),n_{\omega_k}^{(j)}}
+\overline{\ddot l}_{\Delta,\omega_k}^{(j),n_{\omega_k}^{(j)}}
\varrho_{c,\omega_k}^{(j),n_{\omega_k}^{(j)}}
&
\overline{\ddot l}_{\Delta,\omega_k}^{(j),n_{\omega_k}^{(j)}}
\varrho_{s,\omega_k}^{(j),n_{\omega_k}^{(j)}}\\
\overline{\ddot l}_{\Delta,\omega_k}^{(j),n_{\omega_k}^{(j)}}
\varrho_{s,\omega_k}^{(j),n_{\omega_k}^{(j)}}
&
\overline{\ddot l}_{\Sigma,\omega_k}^{(j),n_{\omega_k}^{(j)}}
-\overline{\ddot l}_{\Delta,\omega_k}^{(j),n_{\omega_k}^{(j)}}
\varrho_{c,\omega_k}^{(j),n_{\omega_k}^{(j)}}
\end{bmatrix}
\end{aligned}
\label{eq:C6_DOAM_expected_extent}
\end{equation}

Define the normalized weight of the predicted trajectory component as
\begin{equation}
\bar w_{\omega_k}^{(j)}
=
\frac{
	w_{\omega_k}^{(j)}
}{
\displaystyle
\sum_{s=1}^{J_{\omega_k}}w_{\omega_k}^{(s)}
}
\label{eq:C6_normalized_pred_weight}
\end{equation}

Then the posterior trajectory-component weight in
Eq.~\eqref{eq:C6_DpiD_expand} is 
\begin{equation}
\begin{aligned}
&w_{\pi_k}^{(j,\mathcal P,\mathcal C)}=
\frac{
	p^{\mathsf D}
	\bar w_{\omega_k}^{(j)}
	\nu_{\mathcal P,\mathcal C}
	\mathcal L_{\omega_k}^{(j,\mathcal P,\mathcal C)}
	\rho^{-M_k^{\mathcal P,\mathcal C}}
}{
\displaystyle
\sum_{\mathcal P'\angle\mathcal Z_k}
\sum_{\mathcal C'\in\mathcal P'}
\vartheta_{\mathcal P',\mathcal C'}
\epsilon_{\mathcal P',\mathcal C'}
}
\end{aligned}
\label{eq:C6_post_component_weight}
\end{equation}
\noindent where Eq.~\eqref{eq:C6_DOAM_axis_prediction_mean} requires
$a_{u,\omega_k}^{(j),n_{\omega_k}^{(j)}}>1$, $\rho$ is the clutter
spatial density, and~$\nu_{\mathcal P,\mathcal C}$,
$\vartheta_{\mathcal P,\mathcal C}$, and
$\epsilon_{\mathcal P,\mathcal C}$ are given by
\cite[Eqs.~(46), (42), and~(44)]{Cheng2026ExplicitExtent}, respectively.

The likelihood approximation used above is established in
Appendix~\ref{app:C6_component_weight_proof}.

\end{proposition}

\subsection{TCPHD-\DOAM Cardinality Distribution Recursion}
\label{C6_7}

\begingroup
TCPHD-\DOAM modifies only the single-trajectory state density and its
measurement likelihood; it does not change the TCPHD cardinality
recursion. The predicted cardinality distribution therefore follows
\cite[Eq.~(103)]{Cheng2026ExplicitExtent}, and the posterior cardinality
distributions for nonempty and empty measurement sets follow
\cite[Eqs.~(47) and~(48)]{Cheng2026ExplicitExtent}, respectively. These
standard recursions are not repeated here.

All component, partition, and probability-generating-function terms are
calculated as in \cite[Eqs.~(40), (42), (44)--(46), and
(110)--(114)]{Cheng2026ExplicitExtent}, except that the single-trajectory
likelihood in the posterior recursion is evaluated using the proposed
\DOAM likelihood $\mathcal L_{\omega_k}^{(j,\mathcal P,\mathcal C)}$ in
Eq.~\eqref{eq:C6_DOAM_component_likelihood}.
\endgroup

\subsection{Pruning, Absorption, Capping, and \texorpdfstring{$L$}{L}-Scan Implementation}
\label{C6_8}

After the measurement update, the number of posterior trajectory
components increases rapidly with the numbers of candidate measurement
partitions and cells. The pruning, anchor-based absorption, capping, and
state-extraction procedures in \cite[Secs.~4.5 and~4.6]
{Cheng2026ExplicitExtent} are adopted. Components with weights below
\(\mathtt T_{\mathsf P}\) are discarded, and at most the
\(\mathtt J_{\max}\) components with the largest weights are retained.
The parameters~$\mathtt T_{\mathsf P}$ and~$\mathtt J_{\max}$ denote
the pruning threshold and the maximum number of retained components,
respectively.
Unlike the model in \cite{Cheng2026ExplicitExtent}, the \DOAM separately
represents the orientation and two SSAL states. Therefore, absorption
jointly evaluates the terminal kinematic, orientation, and two SSAL
states, whose statistics are extracted using
Eqs.~\eqref{eq:C6_terminal_r_mean}--\eqref{eq:C6_terminal_axis_statistics}.

Let \(i\) and \(\ddot j\) denote a candidate component and the
largest-weight anchor component, respectively. The terminal kinematic
distance \(d_{\bm r,\pi_k}^{(i,\ddot j)}\) is the Mahalanobis distance
defined in \cite[Sec.~4.5]{Cheng2026ExplicitExtent}. Since elliptical
orientations are equivalent modulo~\(\pi\), the periodic orientation
difference and the corresponding distance are
\begin{equation}
\begin{aligned}
\Delta_\pi(\delta)&=\operatorname{mod}(\delta+\pi/2,\pi)-\pi/2\\
d_{\theta,\pi_k}^{(i,\ddot j)}
&=\Delta_\pi^2(\hat\theta_{\pi_k}^{(i)}-\hat\theta_{\pi_k}^{(\ddot j)})
 /\varXi_{\pi_k}^{\theta,(\ddot j)}
\end{aligned}
\label{eq:C6_absorption_theta_distance}
\end{equation}

Let \(q_{u,\pi_k}^{(i)}\) denote the inverse-Gamma density of the
terminal \(u\)-th SSAL state of component~\(i\). Its symmetric KL
distance from the corresponding state of the anchor component is
\begin{equation}
\begin{aligned}
d_{\ddot l_u,\pi_k}^{(i,\ddot j)}
&=\mathrm{KL}(q_{u,\pi_k}^{(i)}\|q_{u,\pi_k}^{(\ddot j)})\\
&\quad+\mathrm{KL}(q_{u,\pi_k}^{(\ddot j)}\|q_{u,\pi_k}^{(i)})
\quad u\in\{1,2\}
\end{aligned}
\label{eq:C6_absorption_axis_distance}
\end{equation}

Component~\(i\) is absorbed into component~\(\ddot j\) when
\begin{equation}
\begin{aligned}
d_{\bm r,\pi_k}^{(i,\ddot j)}&\leq\mathtt T_{\mathsf{Mr}},\quad
d_{\theta,\pi_k}^{(i,\ddot j)}\leq\mathtt T_{\mathsf{M}\theta}\\
d_{\ddot l_u,\pi_k}^{(i,\ddot j)}&\leq\mathtt T_{\mathsf M\ddot l_u}
\quad u\in\{1,2\}
\end{aligned}
\label{eq:C6_absorption_criterion}
\end{equation}
\noindent where~$\mathtt T_{\mathsf{Mr}}$, $\mathtt T_{\mathsf M\theta}$,
and~$\mathtt T_{\mathsf M\ddot l_u}$ are the absorption thresholds for
the terminal kinematic state, orientation state, and $u$-th SSAL state,
respectively.

The absorbed weights are summed, while the start time and DGDIG
parameters are inherited from the anchor component. The number of
trajectories and the kinematic and orientation state sequences are extracted
as in \cite[Sec.~4.6 and Eq.~(116)]{Cheng2026ExplicitExtent}. For a
selected component, the SSAL estimate is
\(\widehat{\ddot l}_{u,\pi_k}^{(i),h}
=b_{u,\pi_k}^{(i),h}/(a_{u,\pi_k}^{(i),h}-1)\), assuming
\(a_{u,\pi_k}^{(i),h}>1\).

To prevent the computational and storage costs from growing with the
trajectory length, the fixed-lag \(L\)-scan implementation in
\cite{Cheng2026ExplicitExtent} is adopted. Only the most recent \(L\)
kinematic, orientation, and SSAL states are corrected using the current
measurements, while earlier states remain fixed. The cases \(L=1\) and
\(L\geq n\) correspond to terminal-state-only and full-trajectory
updates, respectively; further details are given in
\cite{GarciaFernandez2019TPHDTCPHD}.
\subsection{Computational Complexity Analysis}
\label{C6_9}

\begingroup
The random-matrix CPHD implementation in
\cite[Sec.~III-B]{Lundquist2013GGIWCPHD} requires no subpartitioning
and has complexity comparable to the extended-target PHD filter.
At $L=1$, trajectory filtering retains the terminal-state computations
of its nontrajectory counterpart~\cite[Sec.~VI-E]{GarciaFernandez2019TPHDTCPHD}.
The \DOAM preserves the partition structure. For fixed state dimensions,
the CAVI update for one predicted component and measurement cell has cost
\begin{equation}
T_k^{(j,\mathcal P,\mathcal C)}
=\mathcal O\!\left(
\mathtt N_{\max}^{\mathsf{VB}}M_k^{\mathcal P,\mathcal C}
\right)
\label{eq:C6_complexity_cell}
\end{equation}

Let~$N_k^{\mathcal P}$ denote the number of candidate partitions.
Since~$\sum_{\mathcal C\in\mathcal P}M_k^{\mathcal P,\mathcal C}=M_k$,
summing over all predicted components and candidate partitions gives
\begin{equation}
T_k^{\mathsf{CAVI}}
=\mathcal O\!\left(
\mathtt N_{\max}^{\mathsf{VB}}J_{\omega_k}N_k^{\mathcal P}M_k
\right)
\label{eq:C6_complexity_cavi}
\end{equation}

With a prescribed iteration limit, CAVI contributes a constant iteration
factor, without additional partition combinations or update matrices
whose dimensions grow with trajectory length at~$L=1$.
\endgroup

\section{Experimental Validation}
\label{C6_10}

\begingroup
This section evaluates the proposed algorithm in two road-vehicle
tracking scenarios. The first uses simulated spatial
measurements from vehicle trajectories at a signalized intersection, whereas the
second uses real onboard LiDAR measurements.
The performance statistics for Experiment~1 are computed
from 150 Monte Carlo (MC) runs, whereas Experiment~2 reports results
from a single run in the real-world scenario.
\endgroup

\subsection{Compared Methods}
\label{C6_10_1}

\begingroup
\renewcommand{\DOAM}{DOAM\xspace}

\begingroup
Five methods are compared in Experiment~1.

\textit{1) TCPHD-RMM:} Combines TCPHD filtering with an RMM extent
representation~\cite{Sjudin2021ExtendedTPHD}, using particle-based kinematic and extent
updates.

\textit{2) TCPHD-MEM:} Combines TCPHD filtering with an MEM that explicitly
estimates orientation and semi-axis lengths through sequential
measurement updates~\cite{Cheng2026ExplicitExtent}.

\textit{3) BP-GEM:} Combines online particle-based belief propagation
(BP) with a geometric extent model (GEM), jointly estimating object
existence, kinematic states, and extents under unknown measurement
associations~\cite{Meyer2021ScalableEOT}.

\textit{4) TPMB-BP-RMM:} Combines trajectory Poisson multi-Bernoulli
(TPMB) filtering with particle-based BP and a
random-matrix extent model~\cite{Xia2023TPMBBP}.

\textit{5) TCPHD-\DOAM:} The current implementation combines TCPHD filtering
with the \DOAM, separately estimating orientation and SSAL states
through CAVI as derived in Section~III.
\endgroup

\subsection{State-Space Model Setup}
\label{C6_10_2}

Consider a two-dimensional Cartesian coordinate system.
All five methods use the kinematic state in
Eq.~\eqref{eq:rk_def} and the following state-space model.
\begin{equation}
\bm F_k^{\bm r}
=
\begin{bmatrix}
1&\mathtt T\\
0&1
\end{bmatrix}
\otimes\bm I_2
\label{eq:C6_simulation_kinematic_transition}
\end{equation}
\begin{equation}
\bm Q_k^{\bm r}
=
\mathtt q^{\bm r}
\left(
\begin{bmatrix}
\mathtt T^4/4&\mathtt T^3/2\\
\mathtt T^3/2&\mathtt T^2
\end{bmatrix}
\otimes\bm I_2
\right)
\label{eq:C6_simulation_kinematic_noise}
\end{equation}
\begin{equation}
\bm H_k
=
\begin{bmatrix}
1&0
\end{bmatrix}
\otimes\bm I_2
\label{eq:C6_simulation_measurement_matrix}
\end{equation}
\begin{equation}
\bm R_k
=
\left(\mathtt q^{\bm e}\right)^2\bm I_2
\label{eq:C6_simulation_measurement_noise}
\end{equation}
\noindent where~$\mathtt q^{\bm r}$ is the kinematic-state process-noise
intensity parameter and~$\mathtt q^{\bm e}$ is the additive
measurement-noise standard deviation.

For TCPHD-\DOAM, the orientation state follows a random-walk model. Its
state-transition coefficient and process-noise variance are respectively
set as
\begin{equation}
F_k^\theta=1,
\qquad
Q_k^\theta=(\mathtt q^\theta)^2
\label{eq:C6_simulation_orientation_prediction}
\end{equation}
\noindent where~$\mathtt q^\theta$ is the process-noise standard deviation of the
orientation state. The two SSAL states evolve according to the scalar
Beta--Bartlett model, with their evolution controlled by the forgetting
factor~$\mathtt g$, as given in
Eqs.~\eqref{eq:l_shape_predict} and~\eqref{eq:l_scale_predict}.

\subsection{Performance Metrics}
\label{C6_10_3}

The following experiments evaluate the methods using orientation
root-mean-square error (RMSE), Gaussian Wasserstein distance (GWD),
root-mean-square generalized optimal sub-pattern assignment (RMS-GOSPA),
and cycle time (CYC). Orientation RMSE measures the angular error of
matched estimates and GWD measures their combined centroid and extent
error. RMS-GOSPA additionally accounts for missed targets and redundant
estimates. CYC characterizes computational efficiency.

\begingroup
\setlength{\baselineskip}{11pt}
\setlength{\abovedisplayskip}{4pt}
\setlength{\belowdisplayskip}{4pt}
\setlength{\abovedisplayshortskip}{3pt}
\setlength{\belowdisplayshortskip}{4pt}
Orientation RMSE and GWD use centroid-based one-to-one
matching~\cite[Sec.~5.3]{Cheng2026ExplicitExtent} with a common distance gate.
Let $N_{\mathsf{MC}}$ denote the number of MC runs, $\mathcal A_k^{[c]}$ the
matched-pair set at frame~$k$ in run~$c$, and
$N_k^{\mathcal A}=\sum_{c=1}^{N_{\mathsf{MC}}}|\mathcal A_k^{[c]}|$
the pooled pair count. The orientation RMSE in degrees is
\begin{equation}
\mathsf{RMSE}_{\theta,k}
=
\frac{180}{\pi}
\sqrt{
\frac{1}{N_k^{\mathcal A}}
\sum_{c=1}^{N_{\mathsf{MC}}}
\sum_{(i,j)\in\mathcal A_k^{[c]}}
\left(e_{\theta,k,ij}^{[c]}\right)^2
}
\label{eq:C6_metric_angle_RMSE}
\end{equation}
\noindent where
$e_{\theta,k,ij}^{[c]}=
\operatorname{mod}(\theta_{k,j}^{\mathsf{Est},[c]}
-\theta_{k,i}^{\mathsf{Gt},[c]}+\pi/2,\pi)-\pi/2$
is the wrapped angular error in radians. Both angles are
obtained from the major axes of the physical extent matrices. Their
definitions are independent of the native semiaxis labels and extent
representation.

We use~$\bm p$ for the centroid and~$\bm S$ for the physical
extent matrix. The superscripts~$\mathsf{Gt}$ and~$\mathsf{Est}$ identify
ground truth and estimates. The GWD is defined as~\cite{Yang2016GWDMetric}
\begingroup
\begin{align}
&d_{\mathsf{GWD},k,ij}^{[c]}=\Bigg[
\underbrace{
\left\|
\bm p_{k,i}^{\mathsf{Gt},[c]}
-\bm p_{k,j}^{\mathsf{Est},[c]}
\right\|_2^2
}_{\left(d_{\bm p,k,ij}^{[c]}\right)^2}
\label{eq:C6_metric_GWD}\\
&\,+
\underbrace{
\operatorname{tr}\!\Bigg[
\bm S_{k,i}^{\mathsf{Gt},[c]}
+\bm S_{k,j}^{\mathsf{Est},[c]}
-2\sqrt{
\sqrt{\bm S_{k,i}^{\mathsf{Gt},[c]}}\,
\bm S_{k,j}^{\mathsf{Est},[c]}
\sqrt{\bm S_{k,i}^{\mathsf{Gt},[c]}}
}
\Bigg]
}_{\left(d_{\bm S,k,ij}^{[c]}\right)^2}
\Bigg]^{1/2}
\notag
\end{align}
\endgroup

All matrix square roots are principal. The framewise mean GWD is
\begin{equation}
\bar d_{\mathsf{GWD},k}
=
\frac{1}{N_k^{\mathcal A}}
\sum_{c=1}^{N_{\mathsf{MC}}}
\sum_{(i,j)\in\mathcal A_k^{[c]}}
d_{\mathsf{GWD},k,ij}^{[c]}
\label{eq:C6_metric_mean_GWD}
\end{equation}

We also report centroid RMSE and mean centered-extent GWD
(CEGWD) over the same matched pairs. Centroid RMSE is the square root
of the pooled mean of~$\left(d_{\bm p,k,ij}^{[c]}\right)^2$.
CEGWD is the pooled mean of~$d_{\bm S,k,ij}^{[c]}$.

The GOSPA metric~\cite{Rahmathullah2017GOSPA} evaluates
the complete target and estimate sets and penalizes missed targets
and redundant estimates. Its assignment is independent of the
preceding centroid matching. Let~$d_{\mathsf{GOSPA},k}^{[c]}$ denote
its value with GWD as the base distance. We set the order to
$\mathsf{p}=2$, the cutoff to~$\mathsf{c}=2~\mathrm m$, and
$\mathsf{\alpha}=2$. The RMS-GOSPA over MC runs is
\begin{equation}
\mathsf{RMS}_{\mathsf{GOSPA},k}
=
\sqrt{
\frac{1}{N_{\mathsf{MC}}}
\sum_{c=1}^{N_{\mathsf{MC}}}
\left(d_{\mathsf{GOSPA},k}^{[c]}\right)^2
}
\label{eq:C6_metric_RMS_GOSPA}
\end{equation}

CYC is the mean online processing time per frame.
The simulations are implemented in MATLAB R2025a on a 64-bit Windows
computer with an Intel Core Ultra~7~265K processor and 64~GB of memory.
\endgroup

\endgroup

\subsection{Experiment 1: Signalized-Intersection Simulation}
\label{C6_10_4}

\begingroup
\renewcommand{\DOAM}{DOAM\xspace}
The scenario is adapted from the vehicle trajectories and extent
information in the Signalized Intersection Dataset
(SIND)~\cite{Xu2022SIND}, as used in
\cite[Sec.~5.5]{Cheng2026ExplicitExtent}. The surveillance region is
$[-40,40]\times[-30,60]~(\mathrm{m}^2)$. The sampling interval is
$\mathtt T=1~(\mathrm{s})$, and tracking lasts for $\mathtt K=90$
frames. Targets~1--5 enter at frames~5, 48, 21, 1, and~48,
respectively, and remain until frame~90.
Fig.~\ref{fig:C6_exp1_ground_truth} shows the resulting trajectories,
elliptical extents, and entrance markers.

\setlength{\baselineskip}{11pt}
For all methods, $p^{\mathsf D}=0.99$, $p^{\mathsf S}=0.99$,
$\gamma^{\mathsf O}=10$, $\mathtt q^{\bm r}=0.2~\mathrm{m}^2/\mathrm{s}^4$,
and $\mathtt q^{\bm e}=0.1~\mathrm m$.
To ensure a fair comparison in smoothing lag, the three
TCPHD-based filters use $L=1$ for the single-MC illustration and the
five-method curves. All methods share the same clutter rate
$\lambda^{\mathsf C}=5$.

\setlength{\intextsep}{3pt}
\begin{figure}[H]
	\centering
	\setlength{\abovecaptionskip}{3pt}
	\includegraphics[width=\columnwidth]{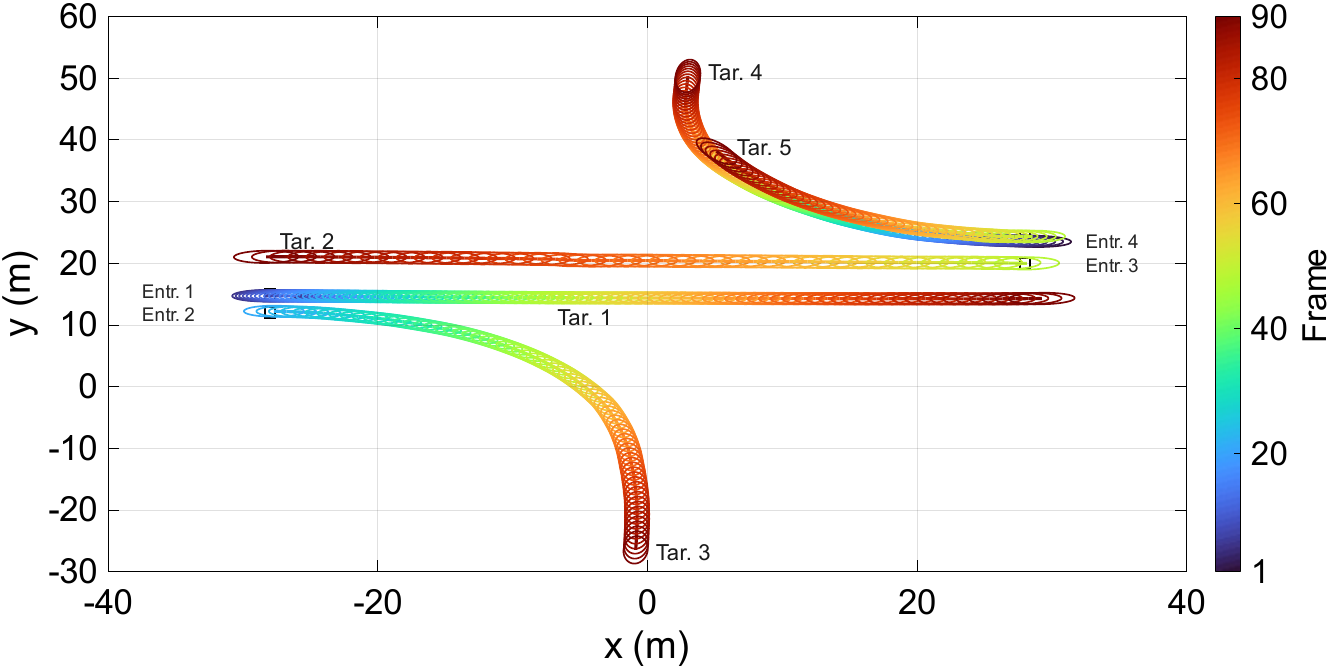}
	\caption{Ground-truth in Experiment~1, where color denotes the frame
		index and squares indicate the entrance positions}
	\label{fig:C6_exp1_ground_truth}
\end{figure}

\begingroup
Let $l_{u,\beta_k}$ denote birth semiaxis $u\in\{1,2\}$.
The quantitative comparisons in Experiment~1 use common baseline
birth-prior settings. All methods start with no targets and use zero
birth-velocity means and kinematic covariance
$\operatorname{diag}(0.5~\mathrm{m}^2,0.5~\mathrm{m}^2,
0.1~\mathrm{m}^2/\mathrm{s}^2,0.1~\mathrm{m}^2/\mathrm{s}^2)$.
For each entrance, a position offset is drawn once from
$\mathcal N(\bm 0,\bm I_2~\mathrm{m}^2)$ and fixed across methods and MC runs.
The shared orientation prior has $\hat\theta_{\beta_k}=\pi/6~\mathrm{rad}$
and $\varXi_{\beta_k}^{\theta}=0.01~\mathrm{rad}^2$.
The semiaxis means are $\mathbb E[l_{1,\beta_k}]=1.6~\mathrm m$ and
$\mathbb E[l_{2,\beta_k}]=1.4~\mathrm m$, with variances approximately
$0.09~\mathrm{m}^2$ and $0.07~\mathrm{m}^2$, respectively.
Each of the three TCPHD filters uses four entrance birth components
of weight~0.1; the two BP-based methods use the same entrances,
a total birth intensity of~0.4, and matched physical extent moments. TCPHD-\DOAM additionally uses
$\mathtt q^\theta=\sqrt{0.05}~\mathrm{rad}$, $\mathtt g=0.99$,
$\mathtt N_{\max}^{\mathsf{VB}}=15$, and $\mathtt T_{\mathsf{VB}}=10^{-4}$.
\endgroup

\Needspace{22\baselineskip}
\setlength{\baselineskip}{11pt}
\setlength{\intextsep}{1pt}
Fig.~\ref{fig:C6_exp1_singleMC_five} shows the matched trajectories and
current extents at frames~49, 55, 70, and~82. The insets focus on
Target~3 in the first two panels and Target~5 in the last two.

\begin{figure}[H]
	\centering
	\setlength{\abovecaptionskip}{3pt}
	\includegraphics[width=\columnwidth]{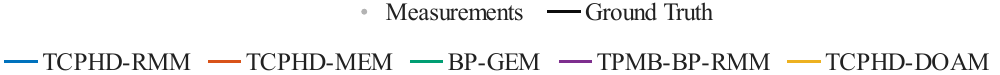}
	\par
	\includegraphics[width=0.93\columnwidth]{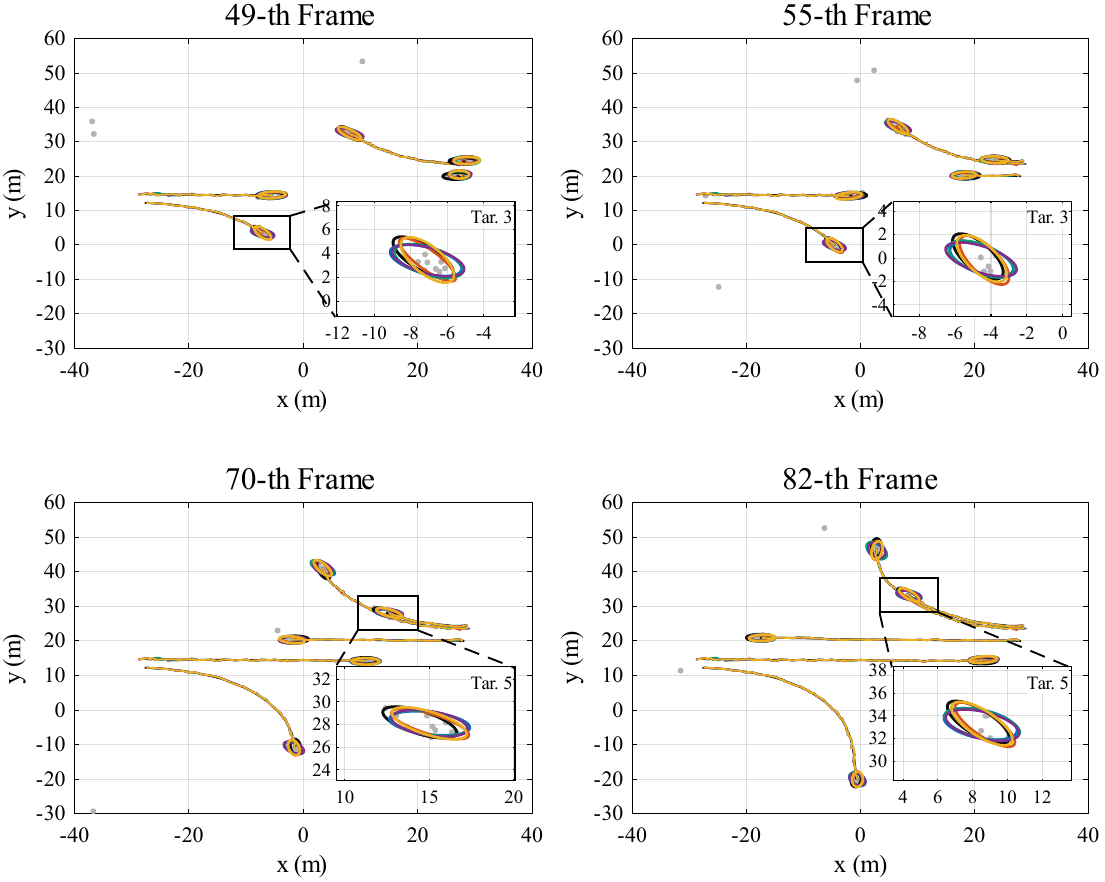}
	\caption{Single-MC matched estimates in Experiment~1}
	\label{fig:C6_exp1_singleMC_five}
\end{figure}

\newpage
\setlength{\baselineskip}{12.5pt}
It can be seen from Fig.~\ref{fig:C6_exp1_singleMC_five}
that all five methods follow the target trajectories in the displayed frames.
The enlarged views reveal orientation deviations for TCPHD-RMM,
BP-GEM, and TPMB-BP-RMM for Target~3, while TCPHD-MEM gives a noticeably
narrow extent for this target at frame~49. The proposed TCPHD-\DOAM
closely matches the ground-truth orientations and extent contours in
these snapshots.

\begin{figure*}[!b]
    \centering
    \includegraphics[width=\textwidth]{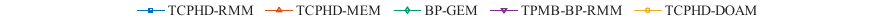}
    \par\vspace{2pt}
    \subfloat[Orientation RMSE\label{fig:C6_exp1_MC150_orientation}]{%
        \includegraphics[width=0.325\textwidth]
        {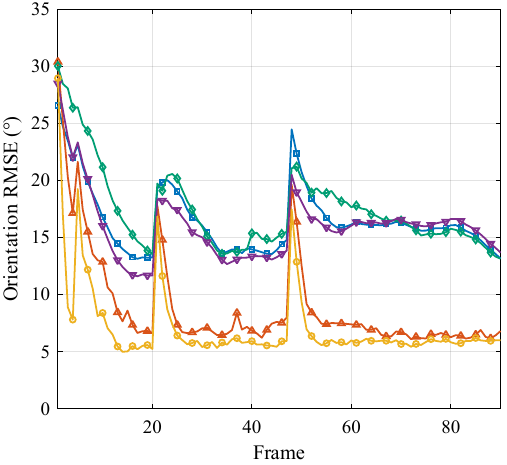}%
    }\hfill
    \subfloat[GWD\label{fig:C6_exp1_MC150_gwd}]{%
        \includegraphics[width=0.325\textwidth]
        {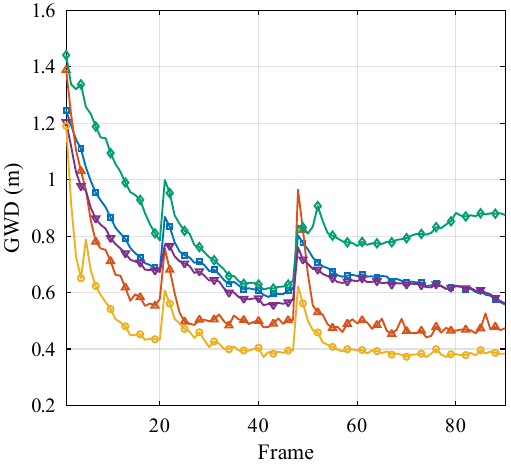}%
    }\hfill
    \subfloat[RMS-GOSPA\label{fig:C6_exp1_MC150_gospa}]{%
        \includegraphics[width=0.325\textwidth]
        {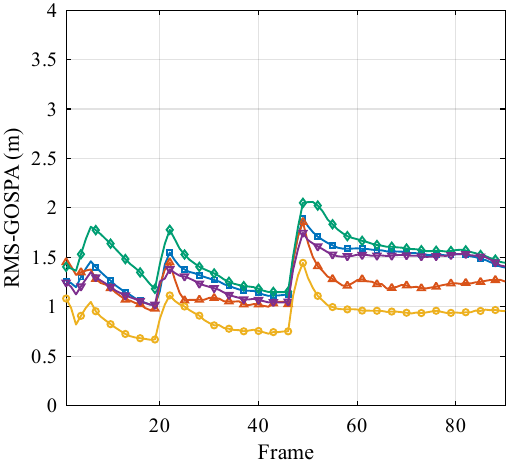}%
    }
    \caption{Performance comparison in Experiment~1}
    \label{fig:C6_exp1_MC150_metrics}
\end{figure*}

\begin{table*}[!b]
\centering
\caption{Time-averaged estimation errors and CYC in Experiment~1}
\label{tab:C6_exp1_MC150_performance}
\footnotesize
\setlength{\tabcolsep}{2.4pt}
\papertablespacing
\begin{tabular*}{\textwidth}{@{\extracolsep{\fill}}ccccccccccccc@{}}
\toprule
\multicolumn{2}{c}{\multirow{2}{*}{\textbf{Metric}}} & \multicolumn{3}{c}{\textbf{TCPHD-RMM}} & \multicolumn{3}{c}{\textbf{TCPHD-MEM}} & \multirow{2}{*}{\textbf{BP-GEM}} & \multirow{2}{*}{\makecell{\textbf{TPMB-BP-}\\\textbf{RMM}}} & \multicolumn{3}{c}{\textbf{TCPHD-\DOAM}} \\
\cmidrule(lr){3-5}\cmidrule(lr){6-8}\cmidrule(lr){11-13}
 & & $L=1$ & $L=3$ & $L=5$ & $L=1$ & $L=3$ & $L=5$ & & & $L=1$ & $L=3$ & $L=5$ \\
\midrule
\multicolumn{2}{c}{Orientation RMSE ($^\circ$) $\downarrow$} & 16.42 & 14.25 & 13.32 & 7.78 & 7.02 & 6.71 & 17.05 & 16.03 & 7.39 & 6.01 & \textbf{5.97} \\
\cmidrule(lr){1-13}
\multicolumn{1}{c|}{} & Overall (m) $\downarrow$ & 0.69 & 0.60 & 0.59 & 0.47 & 0.44 & 0.43 & 0.75 & 0.63 & 0.46 & 0.34 & \textbf{0.33} \\
\multicolumn{1}{c|}{GWD} & Centroid RMSE (m) $\downarrow$ & 0.36 & 0.27 & 0.28 & 0.37 & 0.47 & 0.52 & 0.46 & 0.34 & 0.36 & 0.26 & \textbf{0.25} \\
\multicolumn{1}{c|}{} & CEGWD (m) $\downarrow$ & 0.57 & 0.53 & 0.51 & 0.31 & 0.31 & 0.30 & 0.61 & 0.52 & 0.28 & 0.24 & \textbf{0.23} \\
\cmidrule(lr){1-13}
\multicolumn{2}{c}{RMS-GOSPA (m) $\downarrow$} & 1.39 & 1.19 & 1.17 & 0.95 & 0.98 & 0.96 & 1.47 & 1.33 & 0.94 & 0.70 & \textbf{0.68} \\
\multicolumn{2}{c}{CYC (s) $\downarrow$} & \textbf{0.04} & 0.06 & 0.07 & 0.22 & 0.27 & 0.28 & 0.06 & 0.08 & 0.25 & 0.43 & 0.59 \\
\bottomrule
\end{tabular*}
\end{table*}
Fig.~\ref{fig:C6_exp1_MC150_metrics} summarizes the three metrics.

In Fig.~\ref{fig:C6_exp1_MC150_orientation}, birth-related
transients occur around frames~5, 21, and~48, with different
magnitudes across methods. TCPHD-\DOAM recovers rapidly and
achieves the lowest overall orientation RMSE, followed by
TCPHD-MEM. After frame~60, TCPHD-MEM approaches the error
level of TCPHD-\DOAM, whereas TCPHD-RMM, BP-GEM, and
TPMB-BP-RMM exhibit similar, higher errors.

In Fig.~\ref{fig:C6_exp1_MC150_gwd}, TCPHD-\DOAM maintains the lowest
GWD throughout the sequence, indicating the best combined centroid
and extent accuracy. The TCPHD-MEM peak at frame~48
coincides with the entry of Targets~2 and~5. The reason is that transient
extent errors for the newborn targets and large centroid errors in a small
number of matched pairs increase both GWD components.
BP-GEM's GWD increases after frame~65, mainly because centroid errors grow
in a small subset of MC runs, while its centered-extent error remains
comparatively stable.

\pagebreak[4]
Fig.~\ref{fig:C6_exp1_MC150_gospa} shows the lowest overall RMS-GOSPA
for TCPHD-\DOAM, followed by TCPHD-MEM. TCPHD-\DOAM quickly recovers
from birth-related transients near frames~21 and~48, whereas
TCPHD-RMM, BP-GEM, and TPMB-BP-RMM
retain higher errors at later frames. By accounting
for both matched-state discrepancies and unmatched targets or estimates,
GOSPA complements the pairwise orientation RMSE and GWD. 

Across the 150 MC runs, TCPHD-\DOAM yields the lowest
reported mean orientation RMSE, GWD, and RMS-GOSPA among the compared
methods in this scenario.
\endgroup

\begingroup
\renewcommand{\DOAM}{DOAM\xspace}

\setlength{\baselineskip}{12.5pt}
Table~\ref{tab:C6_exp1_MC150_performance} compares orientation RMSE,
GWD and its components, RMS-GOSPA, and CYC.
TCPHD-RMM, TCPHD-MEM, and TCPHD-\DOAM are evaluated at $L\in\{1,3,5\}$.
At each common $L$, TCPHD-\DOAM matches or improves on TCPHD-RMM
and TCPHD-MEM across all error measures in the table. At $L=1$, it gives
the best overall accuracy among all five methods. Increasing $L$
generally improves accuracy, but the gains diminish from $L=3$ to $L=5$.
Over this step, centroid RMSE increases for TCPHD-RMM and TCPHD-MEM.
The increase for TCPHD-RMM reflects particle-path degeneracy over
longer histories, while that for TCPHD-MEM is driven mainly by large
centroid errors in a few MC runs.
Although TCPHD-\DOAM has a higher CYC, its mean processing times
remain below the sampling interval of
$1~\mathrm s$, indicating an acceptable computational cost in this setting.
\endgroup

\Needspace{30\baselineskip}
\begingroup
Furthermore, Fig.~\ref{fig:C6_exp1_Lscan_comparison} compares
TCPHD-\DOAM at $L\in\{1,2,3,5,10,20\}$.
Changing $L$ has little effect on target birth and death estimation.
We therefore omit RMS-GOSPA and focus on orientation RMSE and GWD
to assess state-estimation accuracy for existing targets.
It can be seen that increasing $L$ reduces both errors relative to $L=1$,
as additional measurements refine the historical-state estimates.
The orientation insets at frames~20--22 and~60--62 show improvements
both near target entry and during subsequent tracking.
The curves become closely spaced at larger $L$, indicating diminishing
accuracy gains. To more clearly illustrate the effect of $L$ on algorithm
performance,
Table~\ref{tab:C6_exp1_Lscan_performance} reports the time-averaged
results. It can be seen from the table that increasing $L$ initially
improves estimation accuracy, but the improvement saturates.
For $L>20$, orientation RMSE changes very little, whereas GWD and
RMS-GOSPA increase. This deterioration is mainly associated with larger
centered-extent errors, while centroid accuracy remains nearly unchanged.
Moreover, CYC already exceeds the sampling interval of $1~\mathrm s$
at $L=10$, making the computational cost unacceptable for processing
at the measurement rate in this setup. For industrial applications with
similar processing constraints, $L=3$ or~$5$ is therefore recommended
as a practical accuracy--cost compromise.

\setlength{\intextsep}{3pt}
\begin{table}[H]
\centering
\renewcommand{\DOAM}{DOAM\xspace}
\caption{Statistics of TCPHD-\DOAM for different L-scan lengths}
\label{tab:C6_exp1_Lscan_performance}
\footnotesize
\setlength{\tabcolsep}{0.6pt}
\papertablespacing
\begin{tabular*}{\columnwidth}{@{\extracolsep{\fill}}ccccccccc@{}}
\toprule
\multirow{2}{*}{\textbf{Metric}} & \multicolumn{8}{c}{\textbf{L-scan length} $L$} \\
\cmidrule(lr){2-9}
 & 1 & 2 & 3 & 5 & 10 & 20 & 40 & 80 \\
\midrule
Orientation RMSE ($^\circ$) $\downarrow$ & 7.39 & 6.20 & 6.01 & 5.97 & 5.96 & 5.96 & 5.96 &   \textbf{5.96} \\
GWD (m) $\downarrow$ & 0.46 & 0.36 & 0.34 & 0.33 & 0.32 & \textbf{0.32} & 0.33 &  0.34 \\
RMS-GOSPA (m) $\downarrow$ & 0.94 & 0.71 & 0.70 & 0.68 & 0.67 & \textbf{0.66} & 0.68 & 0.69 \\
CYC (s) $\downarrow$ & \textbf{0.25} & 0.36 & 0.43 & 0.59 & 1.18 & 1.80 & 4.66 & 11.01 \\
\bottomrule
\end{tabular*}
\end{table}

\begin{table*}[!b]
\centering
\caption{Robustness comparison under shared parameter perturbations}
\label{tab:C6_exp1_robustness}
\begingroup
\papertablespacing
\begin{minipage}{\textwidth}
\setlength{\tabcolsep}{0.55pt}
\begin{tabular*}{\textwidth}{@{\extracolsep{\fill}}>{\raggedright\arraybackslash}p{22mm}ccc@{\hspace{2.5pt}}ccc@{\hspace{2.5pt}}ccc@{\hspace{2.5pt}}ccc@{\hspace{2.5pt}}ccc@{}}
\toprule
\multirow{2}{22mm}{\textbf{Parameter condition}} & \multicolumn{3}{c}{\textbf{TCPHD-RMM}} & \multicolumn{3}{c}{\textbf{TCPHD-MEM}} & \multicolumn{3}{c}{\textbf{BP-GEM}} & \multicolumn{3}{c}{\textbf{TPMB-BP-RMM}} & \multicolumn{3}{c}{\textbf{TCPHD-DOAM}} \\
\cmidrule(lr){2-4}\cmidrule(lr){5-7}\cmidrule(lr){8-10}\cmidrule(lr){11-13}\cmidrule(lr){14-16}
 & {\fontsize{7}{8.4}\selectfont \shortstack{Orientation\\RMSE ($^\circ$)}} & {\fontsize{7}{8.4}\selectfont \shortstack{GWD\\(m)}} & {\fontsize{7}{8.4}\selectfont \shortstack{RMS-\\GOSPA (m)}} & {\fontsize{7}{8.4}\selectfont \shortstack{Orientation\\RMSE ($^\circ$)}} & {\fontsize{7}{8.4}\selectfont \shortstack{GWD\\(m)}} & {\fontsize{7}{8.4}\selectfont \shortstack{RMS-\\GOSPA (m)}} & {\fontsize{7}{8.4}\selectfont \shortstack{Orientation\\RMSE ($^\circ$)}} & {\fontsize{7}{8.4}\selectfont \shortstack{GWD\\(m)}} & {\fontsize{7}{8.4}\selectfont \shortstack{RMS-\\GOSPA (m)}} & {\fontsize{7}{8.4}\selectfont \shortstack{Orientation\\RMSE ($^\circ$)}} & {\fontsize{7}{8.4}\selectfont \shortstack{GWD\\(m)}} & {\fontsize{7}{8.4}\selectfont \shortstack{RMS-\\GOSPA (m)}} & {\fontsize{7}{8.4}\selectfont \shortstack{Orientation\\RMSE ($^\circ$)}} & {\fontsize{7}{8.4}\selectfont \shortstack{GWD\\(m)}} & {\fontsize{7}{8.4}\selectfont \shortstack{RMS-\\GOSPA (m)}} \\
\midrule
Baseline & 16.42 & 0.69 & 1.39 & 7.78 & 0.47 & 0.95 & 17.05 & 0.75 & 1.47 & 16.03 & 0.63 & 1.33 & \textbf{7.39} & \textbf{0.46} & \textbf{0.94} \\
\addlinespace[0.5pt]
$p^{\mathsf D}=0.95$ & 17.37 & 0.74 & 2.17 & 8.09 & 0.48 & 1.14 & 17.81 & 0.80 & 1.57 & 16.50 & 0.71 & 1.49 & \textbf{7.76} & \textbf{0.47} & \textbf{1.13} \\
\addlinespace[0.5pt]
$p^{\mathsf D}=0.80$ & 20.29 & 0.94 & 4.52 & 12.05 & 0.85 & 1.62 & 20.26 & 0.93 & 1.80 & 18.54 & 0.84 & 1.85 & \textbf{10.00} & \textbf{0.81} & \textbf{1.59} \\
\addlinespace[0.5pt]
$p^{\mathsf S}=0.70$ & 21.68 & 0.97 & 3.86 & 7.81 & 0.47 & 0.95 & 18.06 & 0.79 & 1.56 & 16.22 & 0.70 & 1.37 & \textbf{7.42} & \textbf{0.46} & \textbf{0.94} \\
\addlinespace[0.5pt]
$\mathtt q^{\bm e}=2\mathtt q^{\bm e,\star}$ & 16.77 & 0.70 & 1.39 & 8.27 & 0.49 & 0.99 & 17.17 & 0.77 & 1.49 & 16.62 & 0.71 & 1.40 & \textbf{7.74} & \textbf{0.48} & \textbf{0.98} \\
\addlinespace[0.5pt]
$l_{u,\beta_k}=0.5l_{u,\beta_k}^{\star}$ & 19.64 & 0.98 & 2.25 & 7.90 & 0.51 & 1.01 & 19.54 & 1.07 & 2.02 & 19.61 & 1.00 & 2.01 & \textbf{7.44} & \textbf{0.47} & \textbf{0.96} \\
\addlinespace[0.5pt]
$l_{u,\beta_k}=2l_{u,\beta_k}^{\star}$ & 17.51 & 0.87 & 1.68 & 8.24 & \textbf{0.51} & \textbf{1.03} & 24.06 & 1.33 & 2.23 & 17.11 & 0.84 & 1.63 & \textbf{7.82} & 0.62 & 1.20 \\
\addlinespace[0.5pt]
$\lambda^{\mathsf C}=10$ & 16.30 & 0.68 & 1.39 & 8.01 & 0.48 & \textbf{0.99} & 16.96 & 0.75 & 1.47 & 16.18 & 0.70 & 1.40 & \textbf{7.63} & \textbf{0.48} & 1.01 \\
\addlinespace[0.5pt]
$\mathtt q^{\bm r}=0.25\mathtt q^{\bm r,\star}$ & 15.92 & 0.66 & 1.34 & 7.52 & 0.44 & 0.89 & 16.46 & 0.72 & 1.41 & 15.56 & 0.67 & 1.31 & \textbf{7.24} & \textbf{0.43} & \textbf{0.87} \\
\addlinespace[0.5pt]
$\mathtt q^{\bm r}=4\mathtt q^{\bm r,\star}$ & 17.15 & 0.72 & 1.46 & 8.55 & 0.51 & 1.03 & 19.10 & 0.84 & 1.63 & 17.21 & 0.74 & 1.46 & \textbf{7.93} & \textbf{0.49} & \textbf{1.02} \\
\addlinespace[0.5pt]
$\hat\theta_{\beta_k}=45^\circ$ & 15.56 & 0.72 & 1.46 & 11.15 & \textbf{0.53} & \textbf{1.03} & 17.37 & 0.79 & 1.51 & 15.66 & 0.71 & 1.38 & \textbf{10.16} & 0.53 & 1.04 \\
\bottomrule
\end{tabular*}
\par\vspace{2pt}\noindent
\textit{Note:} $\star$ denotes the baseline: $L=1$, $p^{\mathsf D,\star}=0.99$, $p^{\mathsf S,\star}=0.99$, $\gamma^{\mathsf O,\star}=10$, $\lambda^{\mathsf C,\star}=5$, $\mathtt q^{\bm r,\star}=0.2~\mathrm{m}^2/\mathrm{s}^4$, $\mathtt q^{\bm e,\star}=0.1~\mathrm m$ and $\hat\theta_{\beta_k}^{\star}=30^\circ$.
\end{minipage}
\endgroup

\end{table*}

\endgroup

\begingroup
\renewcommand{\DOAM}{DOAM\xspace}
To compare the methods' robustness to parameter perturbations,
Table~\ref{tab:C6_exp1_robustness} reports time-averaged errors under
one-factor changes to detection and survival probabilities, noise levels,
and birth-extent priors. TCPHD-\DOAM achieves the lowest orientation RMSE in all 11 settings
and the lowest GWD and RMS-GOSPA in 9 and 8 settings, respectively.
Most nonbest entries are close to TCPHD-MEM, with gaps below
$0.02~\mathrm m$ except for the doubled birth semiaxes.
At $\hat\theta_{\beta_k}=45^\circ$, a slightly larger centroid error
offsets the centered-extent gain. At $\lambda^{\mathsf C}=10$, increased unmatched-state penalties offset
part of the pairwise accuracy gain.
The clearest exception is $l_{u,\beta_k}=2l_{u,\beta_k}^{\star}$:
the enlarged prior produces more pronounced minor-axis overestimation
during birth transients. This mainly raises the centered-extent error,
giving GWD and RMS-GOSPA gaps of $0.11~\mathrm m$ and $0.17~\mathrm m$,
while late-sequence GWD becomes comparable to TCPHD-MEM.
TCPHD-\DOAM therefore provides the strongest overall robustness in
estimation accuracy among the five methods over the tested parameter range.

\setlength{\intextsep}{3pt}
\begin{figure}[H]
    \centering
    \setlength{\abovecaptionskip}{3pt}
    \subfloat[Orientation RMSE\label{fig:C6_exp1_Lscan_orientation}]{%
        \includegraphics[width=\columnwidth]{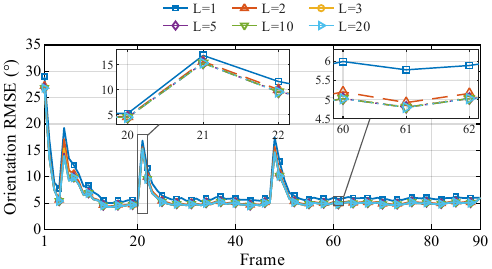}%
}
\par\smallskip
\subfloat[GWD\label{fig:C6_exp1_Lscan_gwd}]{%
        \includegraphics[width=\columnwidth]{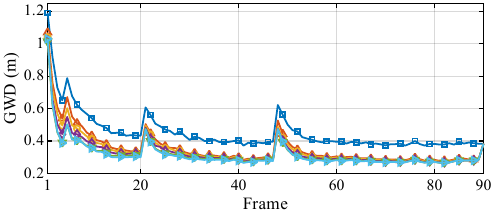}%
}
\caption{Effect of the L-scan length on TCPHD-\DOAM in Experiment~1}
\label{fig:C6_exp1_Lscan_comparison}
\end{figure}
\endgroup

\newpage
\addtocounter{figure}{2}
\begin{figure*}[!b]
    \centering
    \setlength{\abovecaptionskip}{3pt}
    \includegraphics[width=\textwidth]{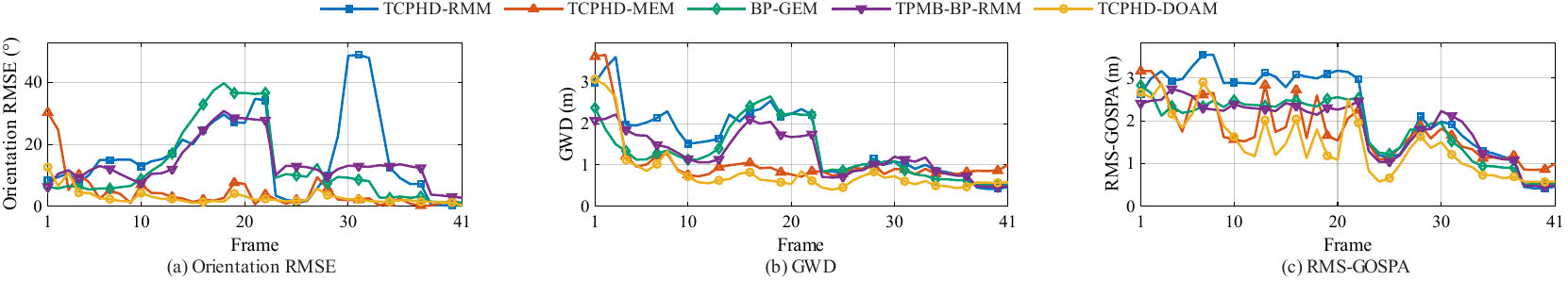}
    \caption{Performance comparison in Experiment~2}
    \label{fig:C6_exp2_metrics}
\end{figure*}
\addtocounter{figure}{-3}
\subsection{Experiment 2: Onboard-LiDAR Road Scenario}
\label{C6_10_5}
\begingroup
\setlength{\baselineskip}{10.7pt}
\setlength{\intextsep}{2pt}

\begingroup
Experiment~2 evaluates the method using the ``scene-0757'' onboard LiDAR
sequence from the public nuScenes dataset~\cite{Caesar2020nuScenes}.
The measurements are expressed in a two-dimensional Cartesian frame.
\endgroup

\looseness=-1
To obtain the tracking measurements in
Fig.~\ref{fig:C6_exp2_scene0757_b}, invalid, duplicate, and ground returns
are removed from the point cloud in Fig.~\ref{fig:C6_exp2_scene0757_a}
before coordinate transformation. Annotation-box gating and clustering
select the vehicle returns, which are registered across frames,
azimuth-balanced, and sparsified. Spatial thresholding and clustering
are common point-cloud preprocessing operations. The constructed
measurement sets retain clutter. Annotations are used only for offline
preparation; all trackers receive the same unlabeled measurement sets.

\begingroup
\begin{figure}[!h]
	\centering
	\captionsetup[subfloat]{farskip=2pt,captionskip=2pt,nearskip=0pt}
	\setlength{\abovecaptionskip}{3pt}
	\subfloat[Total 2-D LiDAR point cloud%
	\label{fig:C6_exp2_scene0757_a}]{%
		\includegraphics[width=0.47\columnwidth]
		{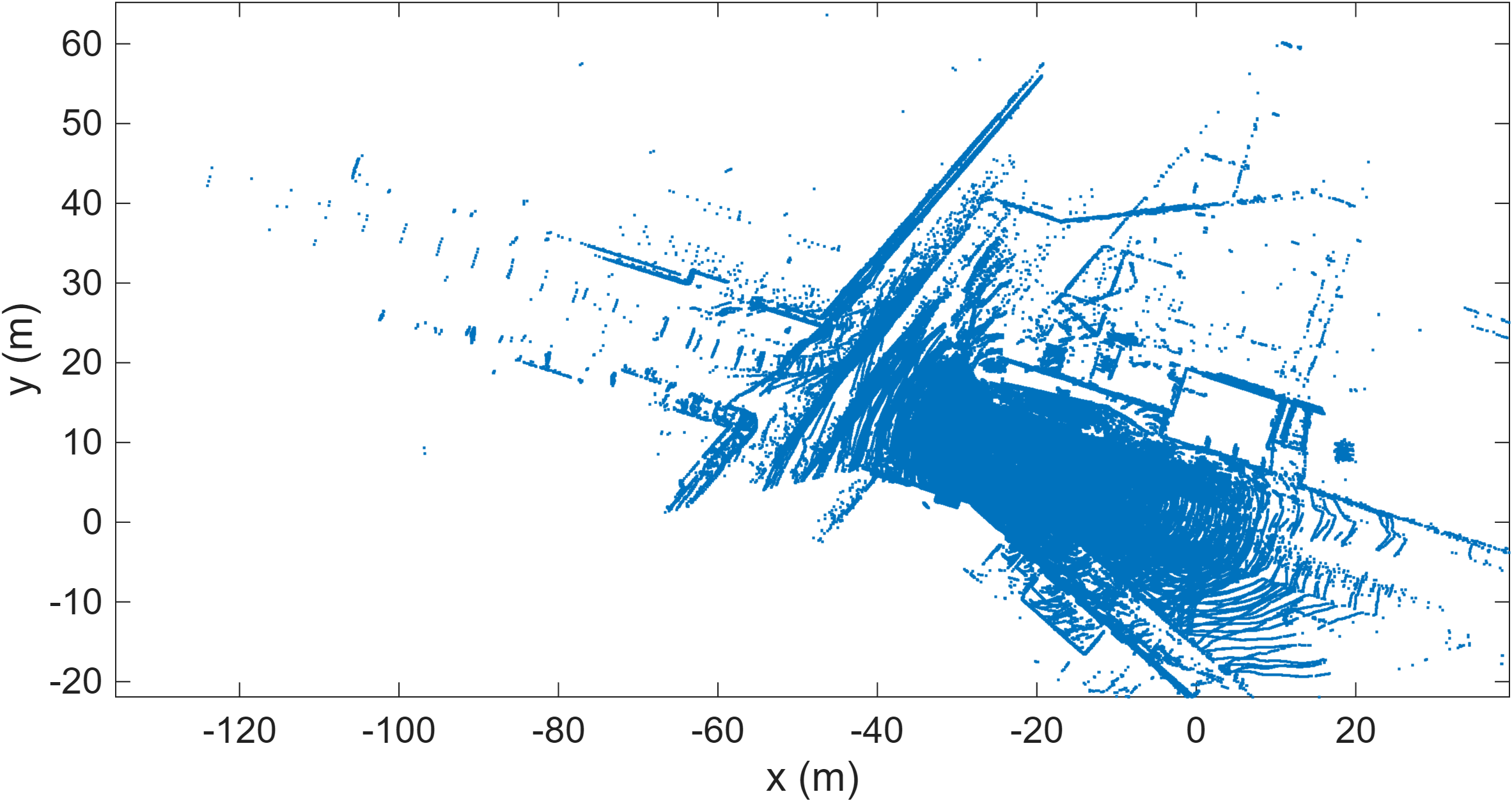}%
	}\hfill
	\subfloat[Tracking measurements with clutter%
	\label{fig:C6_exp2_scene0757_b}]{%
		\includegraphics[width=0.47\columnwidth]
		{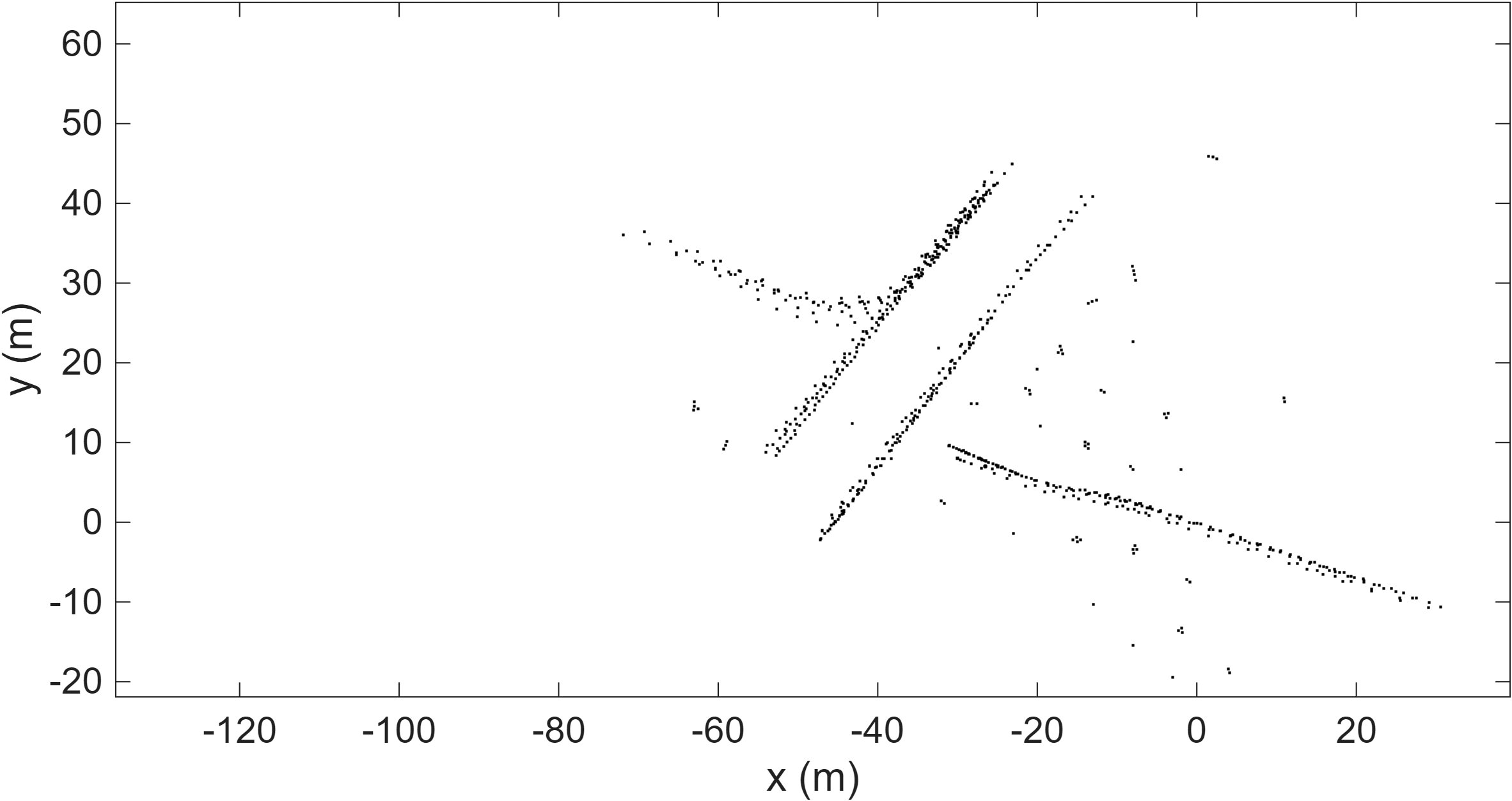}%
	}
	\par
	\subfloat[Ground truth by target identity%
	\label{fig:C6_exp2_gt_target}]{%
			\includegraphics[width=0.47\columnwidth]
			{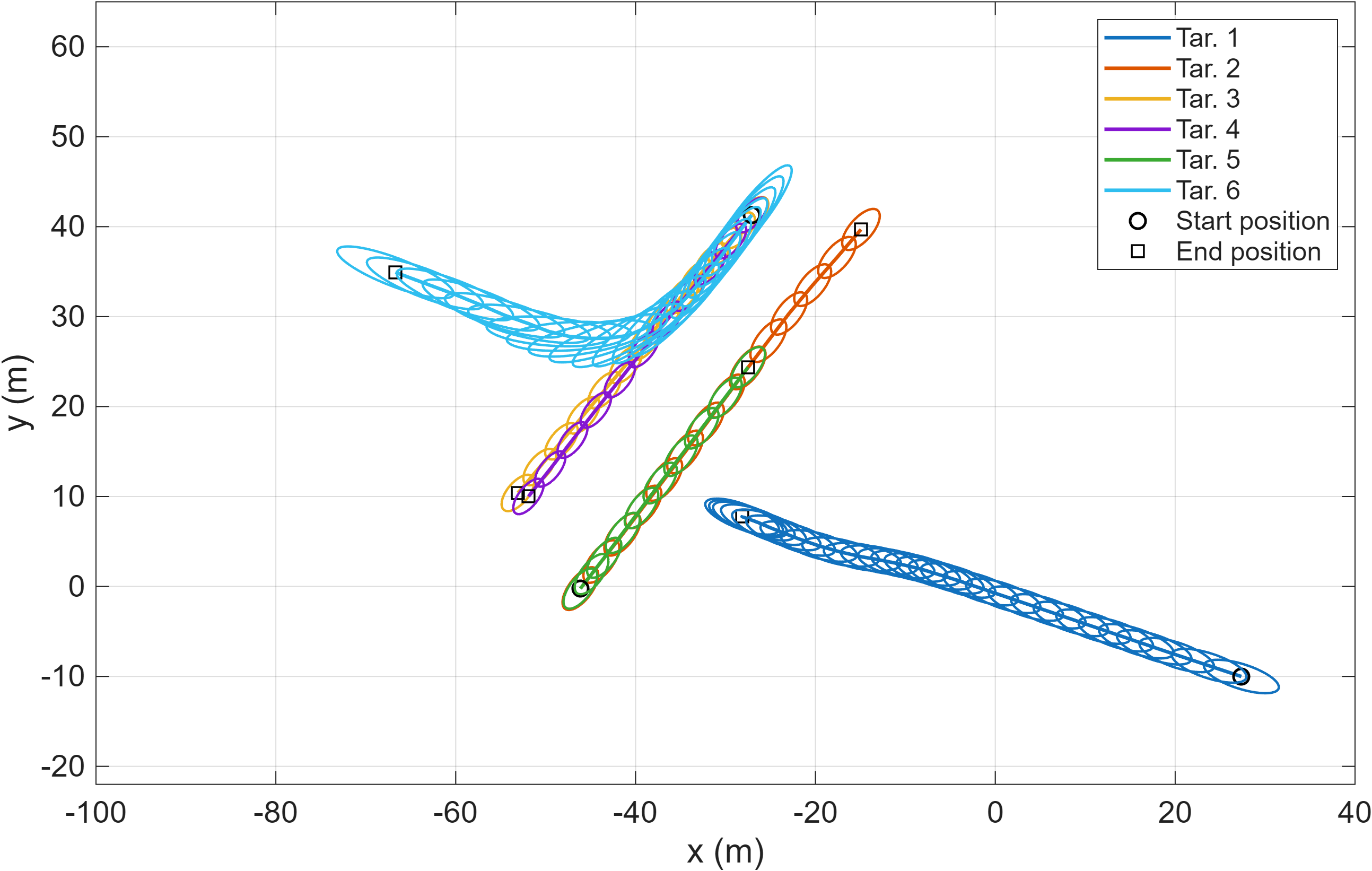}%
		}
		\hfill
	\subfloat[Ground truth by frame index%
	\label{fig:C6_exp2_gt_frame}]{%
			\includegraphics[width=0.47\columnwidth]
			{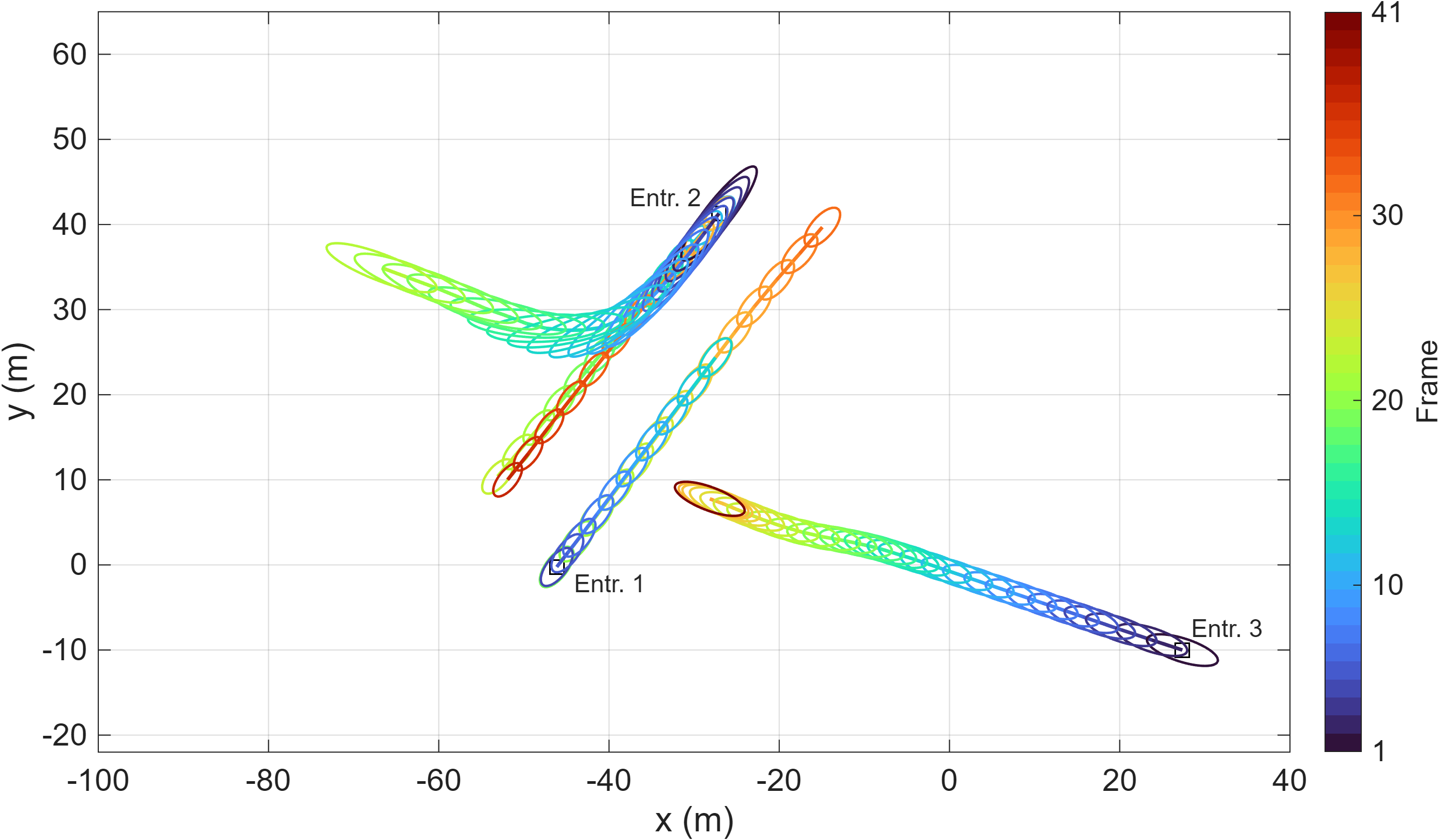}%
		}
	\caption{LiDAR measurements and ground truth in Experiment~2}
	\label{fig:C6_exp2_scene0757_measurement_process}
	\label{fig:C6_exp2_ground_truth}
\end{figure}
\endgroup

\looseness=-1
The ground-truth trajectories in the same
coordinate frame are shown in the lower row of
Fig.~\ref{fig:C6_exp2_ground_truth}, colored by target
identity in panel~(c) and frame index in panel~(d).
The resulting scenario contains six vehicle targets and 41 frames with
$\mathtt T=0.5~(\mathrm{s})$, corresponding to $20~(\mathrm{s})$ of
tracking, and at most four simultaneously present targets.
The surveillance region is
$[-100,40]\times[-22,65]~(\mathrm{m}^2)$. The entrances are Entr.~1 at
$[-46.138,-0.240]^{\mathrm T}$, Entr.~2 at
$[-27.078,41.298]^{\mathrm T}$, and Entr.~3 at
$[27.364,-10.008]^{\mathrm T}$, all in meters. Target~1 enters
through Entr.~3 at frame~1, moves leftward with a slight turn, and
remains until frame~41. Target~2 moves northeastward from Entr.~1 during
frames~19--32; Targets~3 and~4 move southwestward from Entr.~2 during
frames~10--23 and~27--37, respectively; and Target~5 moves northeastward
from Entr.~1 during frames~4--13. Target~6 enters through Entr.~2 at
frame~1, moves southwestward, turns northwest, and last
appears at frame~22.

\begingroup
All five methods process the same fixed 41-frame sequence containing
815 measurements, including 67 clutter points. The common
settings are $p^{\mathsf D}=1$, $p^{\mathsf S}=0.99$,
$\bm R=0.01\bm I_2~\mathrm{m}^2$, $\gamma^{\mathsf O}=6$,
$\lambda^{\mathsf C}=3$, and
$\mathtt q^{\bm r}=0.8~\mathrm{m}^2/\mathrm{s}^4$;
the TCPHD implementations use $L=1$. All methods start with no targets
and share a three-entrance birth prior with total intensity~0.3
and zero mean velocity; each TCPHD entrance component has weight~0.1.
One fixed $\mathcal N(\bm 0,\bm I_2~\mathrm{m}^2)$ position offset per entrance
is shared by all methods. The birth position and velocity covariances are
$4\bm I_2~\mathrm{m}^2$ and
$25\bm I_2~\mathrm{m}^2/\mathrm{s}^2$, respectively, with zero cross-covariance.

The shared extent initialization has semiaxis means
$(2.5,2.3)~\mathrm m$ and variances approximately
$(0.541,0.458)~\mathrm{m}^2$, with
$\hat\theta_{\beta_k}=\pi/6~\mathrm{rad}$ and
$\varXi^\theta_{\beta_k}=(\pi/6)^2~\mathrm{rad}^2$.
TCPHD-DOAM uses inverse-Gamma shapes $a_{u,\beta_k}=4$ and scales
$(b_{1,\beta_k},b_{2,\beta_k})=(20.372,17.243)~\mathrm{m}^2$.
The other methods use moment-matched extent priors to align the nominal
semiaxis lengths.
The remaining method-specific hyperparameters are unchanged from Experiment~1.
\endgroup

Fig.~\ref{fig:C6_exp2_single_run_five} shows estimates at
frames~5, 10, 15, and~25. The insets focus on Target~6 in the first three
panels and Target~1 in the last panel.

\begin{figure}[H]
	\centering
	\setlength{\abovecaptionskip}{4pt}
	\includegraphics[width=\columnwidth]
	{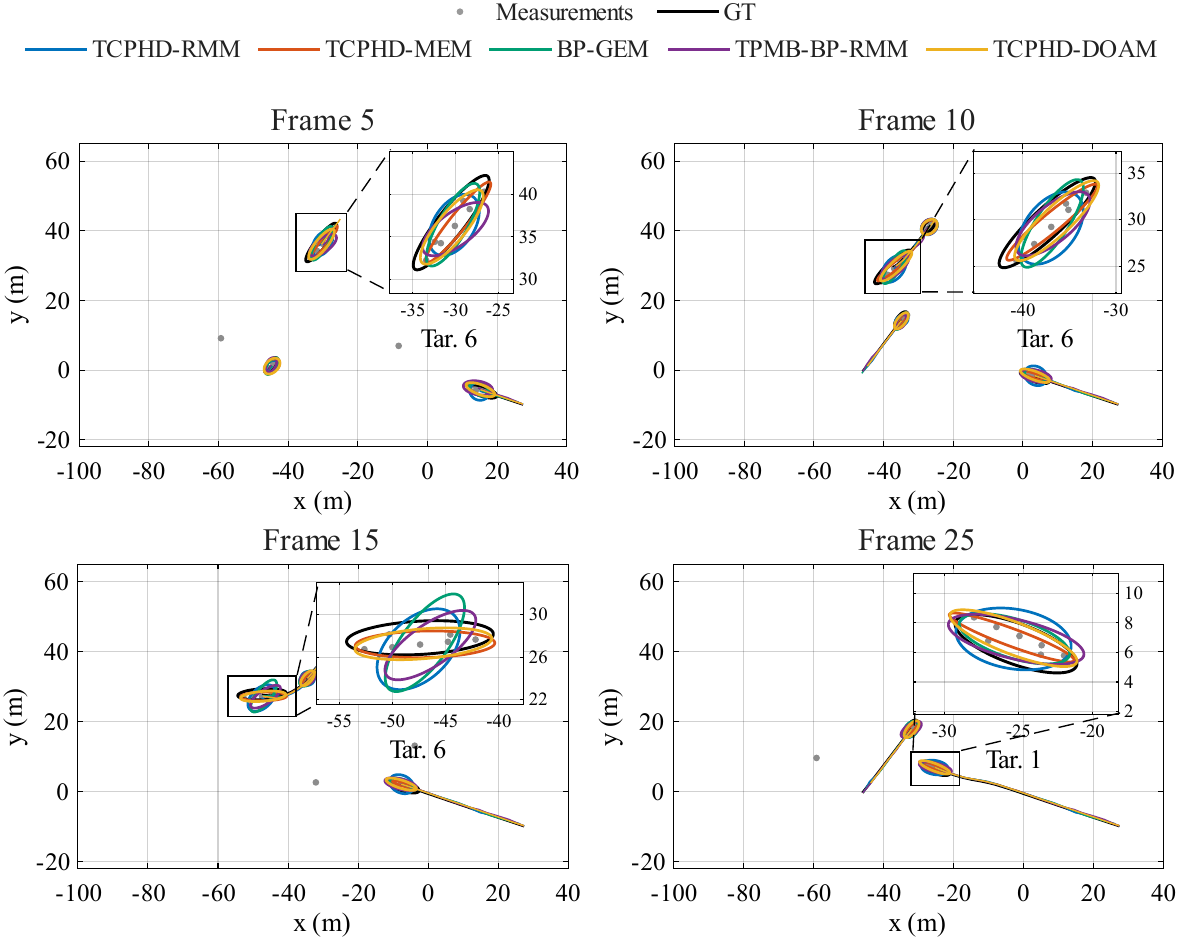}
	\caption{Representative tracking results in Experiment~2}
	\label{fig:C6_exp2_single_run_five}
\end{figure}

It can be seen that all five methods broadly recover the
displayed target positions. At frame~15, TCPHD-RMM, BP-GEM, and TPMB-BP-RMM show clear
orientation deviations for Target~6, whereas TCPHD-MEM and TCPHD-DOAM
follow its near-horizontal extent. At frame~25, TCPHD-MEM underestimates
Target~1's width more noticeably than TCPHD-DOAM.

\stepcounter{figure}

Fig.~\ref{fig:C6_exp2_metrics} presents the three
performance metrics for Experiment~2.
The results in Experiment~2 further corroborate the
findings of Experiment~1: TCPHD-DOAM attains the lowest time-averaged
orientation RMSE, GWD, and RMS-GOSPA on this sequence.
The increases around frames~10--20 in panels~(a) and~(b) are dominated
by Target~6's turn, where TCPHD-RMM, BP-GEM, and TPMB-BP-RMM exhibit
growing orientation lag and extent mismatch; centroid errors also
contribute to GWD. The drop at frame~23 coincides with this target's
exit. The TCPHD-RMM peak around frames~30--35 occurs when a target's
estimated extent becomes nearly circular. Its principal-axis direction
becomes unstable, producing a large angular error without a comparable
GWD peak.

\Needspace{18\baselineskip}
Further comparison of the time-averaged results in
Table~\ref{tab:C6_exp2_performance} shows that TCPHD-DOAM achieves the
lowest errors in all three metrics, followed by TCPHD-MEM. Although
TCPHD-DOAM has the highest CYC, its measured mean of $0.30~\mathrm{s}$
remains below the $0.5~\mathrm{s}$ sampling interval in this setup.
Therefore, the proposed method meets the required processing rate
on average.

\begingroup
\setlength{\tabcolsep}{0.6pt}
\setlength{\heavyrulewidth}{1.0pt}
\setlength{\lightrulewidth}{0.5pt}
\setlength{\arrayrulewidth}{0.5pt}
\setlength{\aboverulesep}{1.5pt}
\setlength{\belowrulesep}{1.5pt}
\begin{table}[!htbp]
\centering
\caption{Performance statistics in Experiment~2}
\label{tab:C6_exp2_performance}
\footnotesize
\papertablespacing
\begin{tabular*}{\columnwidth}{@{\extracolsep{\fill}}cccccc@{}}
\toprule
\textbf{Metric} & \makecell[c]{TCPHD-\\RMM} & \makecell[c]{TCPHD-\\MEM} & BP-GEM & \makecell[c]{TPMB-BP-\\RMM} & \makecell[c]{TCPHD-\\DOAM} \\
\midrule
Orientation RMSE ($^\circ$) $\downarrow$ & 16.39 & 4.38 & 12.89 & 13.82 & \textbf{2.84} \\
GWD (m) $\downarrow$ & 1.57 & 1.07 & 1.30 & 1.27 & \textbf{0.84} \\
RMS-GOSPA (m) $\downarrow$ & 2.22 & 1.77 & 1.84 & 1.86 & \textbf{1.44} \\
CYC (s) $\downarrow$ &\textbf{0.04} & 0.10 & 0.04 & 0.05 & 0.30 \\
\bottomrule
\end{tabular*}
\end{table}

\endgroup

Overall, under annotation-assisted measurement construction,
this experiment based on real LiDAR point clouds shows that TCPHD-DOAM
improves vehicle orientation and extent estimation during turns in the
presence of clutter, further supporting the findings of Experiment~1.
\par
\endgroup

\section{Conclusion}

\begingroup
This paper developed TCPHD-\DOAM for estimating the orientation, centroid,
and footprint of multiple road vehicles during vehicle turns. The \DOAM decouples
orientation from the axial scales, and structured variational updates were
derived for the corresponding trajectory-state sequences, likelihood, and
component weights. Fixed-lag smoothing further refines
historical vehicle states. Experiments with simulated intersection traffic
and real LiDAR point clouds show improved orientation and extent accuracy
during turns, while parameter studies demonstrate robustness across the
tested settings. These results support its use for vehicle tracking at
signalized intersections and for onboard LiDAR perception in road traffic.

Future work will accelerate the CAVI fixed-point iteration
with an acceptance criterion to reduce computational
cost~\cite{Saad2025FixedPointAcceleration}. We will also introduce a
Gamma-distributed measurement-rate state for joint variational estimation
with the trajectory kinematic, orientation, and SSAL states under unknown
and time-varying measurement
rates~\citepair{Wang2025GGIWPMBMSmoother}{Yin2026HeavyTailedGP}.
Furthermore, the TCPHD framework provides a compact trajectory
representation, but its independent and identically distributed cluster
approximation does not preserve dependencies induced by competing
association hypotheses~\cite{GarciaFernandez2019TPHDTCPHD}. To address this
limitation, we will investigate combining the DOAM with trajectory Poisson
multi-Bernoulli mixture
(TPMBM)~\cite{Granstrom2025PMBMTrajectories} or
TPMB~\cite{Xia2023TPMBBP} filtering for dense traffic with ambiguous
associations and temporary occlusions.
\par
\endgroup

\ifCLASSOPTIONcaptionsoff
  \newpage
\fi

\bibliographystyle{IEEEtran}
\bibliography{bibtex/references}

@IEEEtranBSTCTL{IEEEtran:BSTcontrol,
  CTLdash_repeated_names = "no"
}

@article{Cao2024RectangularObjectTITS,
  author  = {Cao, Xiaomeng and Lan, Jian and Liu, Yushuang and Tan, Boyi},
  title   = {Tracking of Rectangular Object Using Key Points with Regionally Concentrated Measurements},
  journal = {IEEE Transactions on Intelligent Transportation Systems},
  year    = {2024},
  volume  = {25},
  number  = {6},
  pages   = {5312--5327},
  doi     = {10.1109/TITS.2023.3332606}
}

@article{Kellner2016HighResolutionDoppler,
  author  = {Kellner, Dominik and Barjenbruch, Michael and Klappstein, Jens and Dickmann, J{\"u}rgen and Dietmayer, Klaus},
  title   = {Tracking of Extended Objects with High-Resolution {Doppler} Radar},
  journal = {IEEE Transactions on Intelligent Transportation Systems},
  year    = {2016},
  volume  = {17},
  number  = {5},
  pages   = {1341--1353},
  doi     = {10.1109/TITS.2015.2501759}
}

@article{2019YangShishanMEM,
  author  = {Yang, Shishan and Baum, Marcus},
  title   = {Tracking the Orientation and Axes Lengths of an Elliptical Extended Object},
  journal = {IEEE Transactions on Signal Processing},
  year    = {2019},
  volume  = {67},
  number  = {18},
  pages   = {4720--4729},
  doi     = {10.1109/TSP.2019.2929462}
}

@article{2025WeiShaoxiuTAES,
  author  = {Wei, Shaoxiu and Garc{\'i}a-Fern{\'a}ndez, {\'A}ngel F. and Yi, Wei},
  title   = {The Trajectory {PHD} Filter for Coexisting Point and Extended Target Tracking},
  journal = {IEEE Transactions on Aerospace and Electronic Systems},
  year    = {2025},
  volume  = {61},
  number  = {3},
  pages   = {5669--5685},
  doi     = {10.1109/TAES.2024.3521921}
}

@inproceedings{Caesar2020nuScenes,
  author    = {Caesar, Holger and Bankiti, Varun and Lang, Alex H. and Vora, Sourabh and Liong, Venice Erin and Xu, Qiang and Krishnan, Anush and Pan, Yu and Baldan, Giancarlo and Beijbom, Oscar},
  title     = {{nuScenes}: A Multimodal Dataset for Autonomous Driving},
  booktitle = {2020 IEEE/CVF Conference on Computer Vision and Pattern Recognition (CVPR)},
  year      = {2020},
  pages     = {11618--11628},
  doi       = {10.1109/CVPR42600.2020.01164}
}

@article{Cao2021EDA,
  author  = {Cao, Xiaomeng and Lan, Jian and Li, X. Rong},
  title   = {{Extension-Deformation Approach} to Extended Object Tracking},
  journal = {IEEE Transactions on Aerospace and Electronic Systems},
  year    = {2021},
  volume  = {57},
  number  = {2},
  pages   = {866--881},
  doi     = {10.1109/TAES.2020.3034028}
}

@article{Cheng2025VBMultipleExtended,
  author  = {Cheng, Yuanhao and Cao, Yunhe and Yeo, Tat-Soon and Zhang, Yulin and Fu, Jie},
  title   = {{Variational Bayesian Inference} for Multiple Extended Targets or Unresolved Group Targets Tracking},
  journal = {IET Radar, Sonar \& Navigation},
  year    = {2025},
  volume  = {19},
  number  = {1},
  pages   = {e70098},
  doi     = {10.1049/rsn2.70098}
}

@article{Cheng2025VGMCPHD,
  author  = {Cheng, Yuanhao and Cao, Yunhe and Yeo, Tat-Soon and Zhang, Yulin and Fu, Jie},
  title   = {{Variational Gaussian Mixture Model} for Tracking Multiple Extended Targets or Unresolvable Group Targets in Closely Spaced Scenarios},
  journal = {Remote Sensing},
  year    = {2025},
  volume  = {17},
  number  = {22},
  pages   = {3696},
  doi     = {10.3390/rs17223696}
}

@article{Cheng2026ExplicitExtent,
  author  = {Cheng, Yuanhao and Cao, Yunhe and Yeo, Tat-Soon and Fu, Jie and Zhang, Wei and Han, Mengmeng},
  title   = {Trajectory {PHD} and {CPHD} Filters for Tracking Multiple Extended Targets with Explicit Extent Estimation},
  journal = {Signal Processing},
  year    = {2026},
  volume  = {244},
  pages   = {110550},
  doi     = {10.1016/j.sigpro.2026.110550}
}

@article{GarciaFernandez2019TPHDTCPHD,
  author  = {Garc{\'i}a-Fern{\'a}ndez, {\'A}ngel F. and Svensson, Lennart},
  title   = {Trajectory {PHD} and {CPHD} Filters},
  journal = {IEEE Transactions on Signal Processing},
  year    = {2019},
  volume  = {67},
  number  = {22},
  pages   = {5702--5714},
  doi     = {10.1109/TSP.2019.2943234}
}

@article{Lundquist2013GGIWCPHD,
  author  = {Lundquist, Christian and Granstr{\"o}m, Karl and Orguner, Umut},
  title   = {An Extended Target {CPHD} Filter and a {Gamma Gaussian Inverse Wishart} Implementation},
  journal = {IEEE Journal of Selected Topics in Signal Processing},
  year    = {2013},
  volume  = {7},
  number  = {3},
  pages   = {472--483},
  doi     = {10.1109/JSTSP.2013.2245632}
}

@article{GarciaFernandez2022TreeTrajectories,
  author  = {Garc{\'i}a-Fern{\'a}ndez, {\'A}ngel F. and Svensson, Lennart},
  title   = {Tracking Multiple Spawning Targets Using {Poisson Multi-Bernoulli Mixtures} on Sets of Tree Trajectories},
  journal = {IEEE Transactions on Signal Processing},
  year    = {2022},
  volume  = {70},
  pages   = {1987--1999},
  doi     = {10.1109/TSP.2022.3165947}
}

@article{Granstrom2014NewPrediction,
  author  = {Granstr{\"o}m, Karl and Orguner, Umut},
  title   = {New Prediction for Extended Targets with {Random Matrices}},
  journal = {IEEE Transactions on Aerospace and Electronic Systems},
  year    = {2014},
  volume  = {50},
  number  = {2},
  pages   = {1577--1589},
  doi     = {10.1109/TAES.2014.120211}
}

@article{Granstrom2017Overview,
  author  = {Granstr{\"o}m, Karl and Baum, Marcus},
  title   = {Extended Object Tracking: Introduction, Overview, and Applications},
  journal = {Journal of Advances in Information Fusion},
  year    = {2017},
  volume  = {12},
  number  = {2},
  pages   = {139--174}
}

@article{Granstrom2025PMBMTrajectories,
  author  = {Granstr{\"o}m, Karl and Svensson, Lennart and Xia, Yuxuan and Williams, Jason L. and Garc{\'i}a-Fern{\'a}ndez, {\'A}ngel F.},
  title   = {{Poisson Multi-Bernoulli Mixtures} for Sets of Trajectories},
  journal = {IEEE Transactions on Aerospace and Electronic Systems},
  year    = {2025},
  volume  = {61},
  number  = {2},
  pages   = {5178--5194},
  doi     = {10.1109/TAES.2024.3517576}
}

@article{IMM-SOEKF,
  author  = {Wang, Shenghua and Men, Chenkai and Li, Renxian and Yeo, Tat-Soon},
  title   = {A Maneuvering Extended Target Tracking {IMM} Algorithm Based on Second-Order {EKF}},
  journal = {IEEE Transactions on Instrumentation and Measurement},
  year    = {2024},
  volume  = {73},
  pages   = {1--11},
  doi     = {10.1109/TIM.2024.3418076}
}

@article{Jiao2024DistributedRM,
  author  = {Jiao, Qinqin and Yang, Xiaojun},
  title   = {Distributed Variational Measurement Update for Extended Target Tracking With {Random Matrix}},
  journal = {IEEE Transactions on Aerospace and Electronic Systems},
  year    = {2024},
  volume  = {60},
  number  = {4},
  pages   = {3792--3806},
  doi     = {10.1109/TAES.2024.3368405}
}

@mastersthesis{Kartal-2022-VS-ETT,
  author = {Kartal, Sava{\c{s}} Erdem},
  title  = {Variational Smoothing for Extended Target Tracking with {Random Matrices}},
  school = {Middle East Technical University},
  year   = {2022},
  url    = {https://hdl.handle.net/11511/96808}
}

@article{Koch2008RM,
  author  = {Koch, Johann Wolfgang},
  title   = {{Bayesian} Approach to Extended Object and Cluster Tracking Using {Random Matrices}},
  journal = {IEEE Transactions on Aerospace and Electronic Systems},
  year    = {2008},
  volume  = {44},
  number  = {3},
  pages   = {1042--1059},
  doi     = {10.1109/TAES.2008.4655362}
}

@article{Lan-model1-2016,
  author  = {Lan, Jian and Li, X. Rong},
  title   = {Tracking of Extended Object or Target Group Using {Random Matrix}: New Model and Approach},
  journal = {IEEE Transactions on Aerospace and Electronic Systems},
  year    = {2016},
  volume  = {52},
  number  = {6},
  pages   = {2973--2989},
  doi     = {10.1109/TAES.2016.130346}
}

@article{Li2022VBEMCPHD,
  author  = {Li, Yawen and Wang, Bo},
  title   = {Multi-Extended Target Tracking Algorithm Based on {VBEM-CPHD}},
  journal = {International Journal of Pattern Recognition and Artificial Intelligence},
  year    = {2022},
  volume  = {36},
  number  = {6},
  pages   = {2250026},
  doi     = {10.1142/S0218001422500264}
}

@article{Li2023JDTCPMBMGGIW,
  author  = {Li, Yuansheng and Wei, Ping and You, Mingyi and Wei, Yifan and Zhang, Huaguo},
  title   = {Joint Detection, Tracking, and Classification of Multiple Extended Objects Based on the {JDTC-PMBM-GGIW} Filter},
  journal = {Remote Sensing},
  year    = {2023},
  volume  = {15},
  number  = {4},
  pages   = {887},
  doi     = {10.3390/rs15040887}
}

@article{Li-model2023,
  author  = {Li, Mingkai and Lan, Jian and Li, X. Rong},
  title   = {Tracking of Elliptical Object With Unknown but Fixed Lengths of Axes},
  journal = {IEEE Transactions on Aerospace and Electronic Systems},
  year    = {2023},
  volume  = {59},
  number  = {5},
  pages   = {6518--6533},
  doi     = {10.1109/TAES.2023.3276951}
}

@article{Liu2021EMExtendedObject,
  author  = {Liu, Shun and Liang, Yan and Xu, Linfeng and Li, Tiancheng and Hao, Xiaohui},
  title   = {{EM}-Based Extended Object Tracking Without a Priori Extension Evolution Model},
  journal = {Signal Processing},
  year    = {2021},
  volume  = {188},
  pages   = {108181},
  doi     = {10.1016/j.sigpro.2021.108181}
}

@article{Liu2022ConstrainedEM,
  author  = {Liu, Shun and Liang, Yan and Xu, Linfeng},
  title   = {Maneuvering Extended Object Tracking Based on Constrained {Expectation Maximization}},
  journal = {Signal Processing},
  year    = {2022},
  volume  = {201},
  pages   = {108729},
  doi     = {10.1016/j.sigpro.2022.108729}
}

@article{Liu2025GMMVB,
  author  = {Liu, Bao and Wu, Ziwei and Liu, Qiang},
  title   = {{Gaussian Mixture Model}-Based {Variational Bayesian} Approach for Extended Target Tracking},
  journal = {IEEE Transactions on Instrumentation and Measurement},
  year    = {2025},
  volume  = {74},
  pages   = {1--18},
  doi     = {10.1109/TIM.2025.3565347}
}

@article{Mihaylova2014SMCReviewGroupExtended,
  author  = {Mihaylova, Lyudmila and Carmi, Avishy Y. and Septier, Fran{\c{c}}ois and Gning, Amadou and Pang, Sze Kim and Godsill, Simon},
  title   = {Overview of {Bayesian} Sequential {Monte Carlo} Methods for Group and Extended Object Tracking},
  journal = {Digital Signal Processing},
  year    = {2014},
  volume  = {25},
  pages   = {1--16},
  doi     = {10.1016/j.dsp.2013.11.006}
}

@article{Orguner2012VB,
  author  = {Orguner, Umut},
  title   = {A Variational Measurement Update for Extended Target Tracking With {Random Matrices}},
  journal = {IEEE Transactions on Signal Processing},
  year    = {2012},
  volume  = {60},
  number  = {7},
  pages   = {3827--3834},
  doi     = {10.1109/TSP.2012.2192927}
}

@article{Pei2025ConstrainedRM,
  author  = {Pei, Shiqi and Wang, Zhen and Li, Ruiyuan and Liu, Jun and Li, Pin and Chen, Chang and Chen, Weidong},
  title   = {Constrained Extended Target Tracking With {Random Matrix} and Approximate Projection},
  journal = {IEEE Sensors Journal},
  year    = {2025},
  volume  = {25},
  number  = {20},
  pages   = {38594--38612},
  doi     = {10.1109/JSEN.2025.3609528}
}

@article{Sahin2024TrajectoryAligned,
  author  = {{\c{S}}ahin, Kurtulu{\c{s}} Kerem and Balc{\i}, Ali Emre and {\"O}zkan, Emre},
  title   = {{Random Matrix} Extended Target Tracking for Trajectory-Aligned and Drifting Targets},
  journal = {IET Radar, Sonar \& Navigation},
  year    = {2024},
  volume  = {18},
  number  = {11},
  pages   = {2247--2263},
  doi     = {10.1049/rsn2.12628}
}

@article{Sarkar2026IMM,
  author  = {Sarkar, Sanglap and Roy, Anirban},
  title   = {Extended Object Tracking Algorithms Within {Interacting Multiple Model} Framework for Automotive Radar Applications},
  journal = {IEEE Transactions on Radar Systems},
  year    = {2026},
  volume  = {4},
  pages   = {1415--1425},
  doi     = {10.1109/TRS.2026.3718282}
}

@inproceedings{Sjudin2021ExtendedTPHD,
  author    = {Sjudin, Jakob and Marcusson, Martin and Svensson, Lennart and Hammarstrand, Lars},
  title     = {Extended Object Tracking Using Sets of Trajectories with a {PHD} Filter},
  booktitle = {2021 IEEE 24th International Conference on Information Fusion (FUSION)},
  year      = {2021},
  pages     = {1--8},
  doi       = {10.23919/FUSION49465.2021.9626994}
}

@article{Steuernagel2025RBPF,
  author  = {Steuernagel, Simon and Baum, Marcus},
  title   = {Extended Object Tracking by {Rao--Blackwellized Particle Filtering} for Orientation Estimation},
  journal = {IEEE Transactions on Signal Processing},
  year    = {2025},
  volume  = {73},
  pages   = {2590--2602},
  doi     = {10.1109/TSP.2025.3574689}
}

@article{Tuncer-model-2021,
  author  = {Tuncer, Bark{\i}n and {\"O}zkan, Emre},
  title   = {{Random Matrix} Based Extended Target Tracking With Orientation: A New Model and Inference},
  journal = {IEEE Transactions on Signal Processing},
  year    = {2021},
  volume  = {69},
  pages   = {1910--1923},
  doi     = {10.1109/TSP.2021.3065136}
}

@article{Wang2022RobustTrajectoryRFS,
  author  = {Wang, Zhiwei and Lu, Zhejun and Liu, Yongxiang and Zhang, Chi},
  title   = {Robust Distributed Fusion with Trajectory Random Finite Sets},
  journal = {Signal Processing},
  year    = {2022},
  volume  = {200},
  pages   = {108675},
  doi     = {10.1016/j.sigpro.2022.108675}
}

@article{Wang2026IMMPF,
  author  = {Wang, Shenghua and Luo, Jiaxuan and Zhang, Shaohua and Cao, Yunhe and Yeo, Tat-Soon},
  title   = {Extended Target Tracking Algorithm Based on {Interactive Multiple Model Particle Filter}},
  journal = {IEEE Transactions on Instrumentation and Measurement},
  year    = {2026},
  volume  = {75},
  pages   = {8511414},
  doi     = {10.1109/TIM.2026.3711351}
}

@article{Wen2024VelocityOrientation,
  author  = {Wen, Zheng and Lan, Jian and Zheng, Le and Zeng, Tao},
  title   = {Velocity-Dependent Orientation Estimation Using Variance Adaptation for Extended Object Tracking},
  journal = {IEEE Signal Processing Letters},
  year    = {2024},
  volume  = {31},
  pages   = {3109--3113},
  doi     = {10.1109/LSP.2024.3492718}
}

@article{Xie2023MMPMBM,
  author  = {Xie, Xingxiang and Wang, Yang and Guo, Junqi and Zhou, Rundong},
  title   = {The {Multiple Model Poisson Multi-Bernoulli Mixture Filter} for Extended Target Tracking},
  journal = {IEEE Sensors Journal},
  year    = {2023},
  volume  = {23},
  number  = {13},
  pages   = {14304--14314},
  doi     = {10.1109/JSEN.2023.3270272}
}

@inproceedings{Xu2022SIND,
  author    = {Xu, Yanchao and Shao, Wenbo and Li, Jun and Yang, Kai and Wang, Weida and Huang, Hua and Lv, Chen and Wang, Hong},
  title     = {{SIND}: A Drone Dataset at Signalized Intersection in {China}},
  booktitle = {2022 IEEE 25th International Conference on Intelligent Transportation Systems (ITSC)},
  year      = {2022},
  pages     = {2471--2478},
  doi       = {10.1109/ITSC55140.2022.9921959}
}

@inproceedings{Yang2016GWDMetric,
  author    = {Yang, Shishan and Baum, Marcus and Granstr{\"o}m, Karl},
  title     = {Metrics for Performance Evaluation of Elliptic Extended Object Tracking Methods},
  booktitle = {2016 IEEE International Conference on Multisensor Fusion and Integration for Intelligent Systems (MFI)},
  year      = {2016},
  pages     = {523--528},
  doi       = {10.1109/MFI.2016.7849541}
}

@article{Yang2019NetworkFlow,
  author  = {Yang, Shishan and Teich, Florian and Baum, Marcus},
  title   = {Network Flow Labeling for Extended Target Tracking {PHD} Filters},
  journal = {IEEE Transactions on Industrial Informatics},
  year    = {2019},
  volume  = {15},
  number  = {7},
  pages   = {4164--4171},
  doi     = {10.1109/TII.2019.2898992}
}

@article{Yang2023AdaptiveVB,
  author  = {Yang, Xiaojun and Jiao, Qinqin},
  title   = {Variational Approximation for Adaptive Extended Target Tracking in Clutter With {Random Matrix}},
  journal = {IEEE Transactions on Vehicular Technology},
  year    = {2023},
  volume  = {72},
  number  = {10},
  pages   = {12639--12652},
  doi     = {10.1109/TVT.2023.3275633}
}

@article{Zhang-model-2025,
  author  = {Zhang, Le and Lan, Jian},
  title   = {Extended Object Tracking Using Aspect Ratio},
  journal = {IEEE Transactions on Signal Processing},
  year    = {2025},
  volume  = {73},
  pages   = {4193--4207},
  doi     = {10.1109/TSP.2024.3400870}
}

@article{Zheng2025HeavyTailed,
  author  = {Zheng, Xiangfei and Zhang, Yujie and Wu, Sunyong and Li, Hongwei},
  title   = {Extended Object Tracking With Inaccurate Heavy-Tailed Noises},
  journal = {IEEE Transactions on Instrumentation and Measurement},
  year    = {2025},
  volume  = {74},
  pages   = {1--10},
  doi     = {10.1109/TIM.2025.3614909}
}

@article{Zheng2025OrientationVector,
  author  = {Wen, Zheng and Zheng, Le and Zeng, Tao},
  title   = {Extended Object Tracking Using an Orientation Vector Based on Constrained Filtering},
  journal = {Remote Sensing},
  year    = {2025},
  volume  = {17},
  number  = {8},
  pages   = {1419},
  doi     = {10.3390/rs17081419}
}

@article{Tuncer2022MultiEllipsoidalVB,
  author  = {Tuncer, Bark{\i}n and Orguner, Umut and {\"O}zkan, Emre},
  title   = {Multi-Ellipsoidal Extended Target Tracking With {Variational Bayes} Inference},
  journal = {IEEE Transactions on Signal Processing},
  year    = {2022},
  volume  = {70},
  pages   = {3921--3934},
  doi     = {10.1109/TSP.2022.3192617}
}

@article{Saad2025FixedPointAcceleration,
  author  = {Saad, Yousef},
  title   = {Acceleration Methods for Fixed-Point Iterations},
  journal = {Acta Numerica},
  year    = {2025},
  volume  = {34},
  pages   = {805--890},
  doi     = {10.1017/S0962492924000096}
}

@article{Wang2025GGIWPMBMSmoother,
  author  = {Wang, Wenhui and Xu, Ye and Zhang, Kejie and Sun, Youpeng and Li, Peng},
  title   = {A {GGIW PMBM} Smoother for Multiple Extended Object Tracking},
  journal = {Electronics Letters},
  year    = {2025},
  volume  = {61},
  number  = {1},
  pages   = {e70188},
  doi     = {10.1049/ELL2.70188}
}

@article{Yin2026HeavyTailedGP,
  author  = {Yin, Xiaoqian and Yang, Xiaojun and Ning, Hang},
  title   = {Variational Extended Target Tracking With Heavy-Tailed Measurement Noise Using {Gaussian Processes}},
  journal = {IEEE Transactions on Automation Science and Engineering},
  year    = {2026},
  volume  = {23},
  pages   = {13874--13888},
  doi     = {10.1109/TASE.2026.3717112}
}

@article{Meyer2021ScalableEOT,
  author  = {Meyer, Florian and Williams, Jason L.},
  title   = {Scalable Detection and Tracking of Geometric Extended Objects},
  journal = {IEEE Transactions on Signal Processing},
  year    = {2021},
  volume  = {69},
  pages   = {6283--6298},
  doi     = {10.1109/TSP.2021.3121631}
}

@article{Xia2023TPMBBP,
  author  = {Xia, Yuxuan and Garc{\'i}a-Fern{\'a}ndez, {\'A}ngel F. and Meyer, Florian and Williams, Jason L. and Granstr{\"o}m, Karl and Svensson, Lennart},
  title   = {Trajectory {PMB} Filters for Extended Object Tracking Using {Belief Propagation}},
  journal = {IEEE Transactions on Aerospace and Electronic Systems},
  year    = {2023},
  volume  = {59},
  number  = {6},
  pages   = {9312--9331},
  doi     = {10.1109/TAES.2023.3317233}
}

@article{Liang2024NeuralBP,
  author  = {Liang, Mingchao and Meyer, Florian},
  title   = {{Neural Enhanced Belief Propagation} for Multiobject Tracking},
  journal = {IEEE Transactions on Signal Processing},
  year    = {2024},
  volume  = {72},
  pages   = {15--30},
  doi     = {10.1109/TSP.2023.3314275}
}

@article{Lyu2026PLMBBP,
  author  = {Lyu, Runyan and Hao, Liang and Zheng, Litao and Cai, Yunze},
  title   = {A Fast {Poisson Labeled Multi-Bernoulli} Filter for Extended Object Tracking Using {Belief Propagation}},
  journal = {Signal Processing},
  year    = {2026},
  volume  = {238},
  pages   = {110156},
  doi     = {10.1016/j.sigpro.2025.110156}
}

@inproceedings{Rahmathullah2017GOSPA,
  author    = {Rahmathullah, Abu Sajana and Garc{\'i}a-Fern{\'a}ndez, {\'A}ngel F. and Svensson, Lennart},
  title     = {Generalized Optimal Sub-Pattern Assignment Metric},
  booktitle = {2017 20th International Conference on Information Fusion (FUSION)},
  year      = {2017},
  pages     = {1--8},
  doi       = {10.23919/ICIF.2017.8009645}
}

\onecolumn
\raggedbottom
\begin{bibunit}[IEEEtran]
\bstctlcite{IEEEtran:BSTcontrol}
\makeatletter
\def\@extra@b@citeb{.supp}
\def\bibcite#1#2{%
	\global\@namedef{b@#1}{%
		\hyper@@link[cite]{}{cite.#1.supp}{#2}}%
}
\makeatother
\section*{Supplementary Material}
\setcounter{section}{0}
\renewcommand{\thesection}{S\arabic{section}}
\renewcommand{\theHsection}{supplementary.\arabic{section}}
\setcounter{equation}{0}
\renewcommand{\theequation}{S\arabic{equation}}
\renewcommand{\theHequation}{supplementary.\arabic{equation}}

\section{Proof of Proposition~\ref{prop:C6_r_update}}
\label{app:C6_r_update_proof}

According to the structured variational factorization in
Eq.~\eqref{eq:C6_traj_vb_factor} and the CAVI coordinate-update rule,
the variational factor of the posterior trajectory kinematic state
sequence satisfies
\begin{equation}
\begin{aligned}
\ln q_k^{\bm r,(j,\mathcal P,\mathcal C),[\ell+1]}
\!\left(\bm r^{1:n_{\pi_k}^{(j)}}\right)
&=
\mathbb E_{\substack{
		q_k^{\theta,(j,\mathcal P,\mathcal C),[\ell]}
		\prod_{u=1}^{2}q_k^{\ddot l_u,(j,\mathcal P,\mathcal C),[\ell]}\\
		{}\times\prod_{m=1}^{M_k^{\mathcal P,\mathcal C}}
		q_k^{\bm y^m,(j,\mathcal P,\mathcal C),[\ell]}}}
\!\Bigg[
\ln p_k\!\left(
\begin{gathered}
\bm r^{1:n_{\omega_k}^{(j)}},
\bm\theta^{1:n_{\omega_k}^{(j)}},
\ddot l_1^{1:n_{\omega_k}^{(j)}},
\ddot l_2^{1:n_{\omega_k}^{(j)}},\\
\mathcal Y_k^{(\mathcal P,\mathcal C)},
\mathcal Z_k^{(\mathcal P,\mathcal C)}
\end{gathered}
\right)
\Bigg]
+\mathtt c_{\bm r}
\end{aligned}
\label{eq:C6_proof_r_CAVI}
\end{equation}
\noindent where~$p_k(\cdot)$ denotes the joint density of all state
variables, measurement-source variables, and actual measurements for the
$j$-th predicted trajectory component and measurement
cell~$(\mathcal P,\mathcal C)$, and~$\mathtt c_{\bm r}$ denotes a term
independent of~$\bm r^{1:n_{\pi_k}^{(j)}}$.

The trajectory kinematic-state density of the $j$-th predicted
trajectory component is
\begin{equation}
p_{\omega_k}^{\bm r,(j)}
\!\left(\bm r^{1:n_{\omega_k}^{(j)}}\right)
=
\mathcal N\!\left(
\bm r^{1:n_{\omega_k}^{(j)}};
\hat{\bm\varGamma}_{\omega_k}^{(j)},
\bm\varXi_{\omega_k}^{\bm\varGamma,(j)}
\right)
\label{eq:C6_proof_r_pred_density}
\end{equation}

For the $m$-th measurement-source variable in the measurement cell, its
conditional density given the terminal trajectory state is
\begin{equation}
\begin{aligned}
p\!\left(
\bm y_k^{\mathcal P,\mathcal C,m}\mid
\bm r^{1:n_{\omega_k}^{(j)}},\theta^{n_{\omega_k}^{(j)}},
\ddot l_1^{n_{\omega_k}^{(j)}},
\ddot l_2^{n_{\omega_k}^{(j)}}
\right) &= \mathcal N\!\left(
\bm y_k^{\mathcal P,\mathcal C,m};
\dot{\bm H}_k^{(j)}\bm r^{1:n_{\omega_k}^{(j)}},
\mathtt s \mathcal{S}\left(\theta^{n_{\omega_k}^{(j)}},\ddot l_1^{n_{\omega_k}^{(j)}},
\ddot l_2^{n_{\omega_k}^{(j)}}\right)
\right)\\
&\quad=
\mathcal N\!\left(
\bm y_k^{\mathcal P,\mathcal C,m};
\dot{\bm H}_k^{(j)}\bm r^{1:n_{\omega_k}^{(j)}},
\mathtt s\mathcal R\left(\theta^{n_{\omega_k}^{(j)}}\right)\mathcal W\left(\ddot l_1^{n_{\omega_k}^{(j)}},
\ddot l_2^{n_{\omega_k}^{(j)}}\right)\mathcal R\left(-\theta^{n_{\omega_k}^{(j)}}\right)
\right)
\end{aligned}
\label{eq:C6_proof_r_source_density2}
\end{equation}
\begin{equation}
\dot{\bm H}_k^{(j)}
=
\bm H_k
\left(
\bm e_{n_{\omega_k}^{(j)}}^{\mathrm T}
\otimes\bm I_{d_{\bm r}}
\right)
\label{eq:C6_proof_terminal_measurement_matrix}
\end{equation}

Given~$\bm y_k^{\mathcal P,\mathcal C,m}$, the actual
measurement conditional density
$p(\bm z_k^{\mathcal P,\mathcal C,m}\mid
\bm y_k^{\mathcal P,\mathcal C,m})$ does not contain the trajectory
kinematic state. Hence, the terms in the joint density that depend on
the trajectory kinematic state satisfy
\begin{equation}
\begin{aligned}
&\ln q_k^{\bm r,(j,\mathcal P,\mathcal C),[\ell+1]}
\!\left(\bm r^{1:n_{\pi_k}^{(j)}}\right)
=\underbrace{
	-\frac{1}{2}
	\left(\bm r^{1:n_{\omega_k}^{(j)}}-
	\hat{\bm\varGamma}_{\omega_k}^{(j)}\right)^{\mathrm T}
	\left(\bm\varXi_{\omega_k}^{\bm\varGamma,(j)}\right)^{-1}
	\left(\bm r^{1:n_{\omega_k}^{(j)}}-
	\hat{\bm\varGamma}_{\omega_k}^{(j)}\right)
}_{\mathlarger{\blacktriangle}}\\
&\quad-\frac{1}{2}
\sum_{m=1}^{M_k^{\mathcal P,\mathcal C}}
\mathbb E_{\substack{
		q_k^{\theta,(j,\mathcal P,\mathcal C),[\ell]}
		\prod_{u=1}^{2}q_k^{\ddot l_u,(j,\mathcal P,\mathcal C),[\ell]}\\
		{}\times q_k^{\bm y^m,(j,\mathcal P,\mathcal C),[\ell]}}}
\!\Bigg[
\left(\underbrace{\bm y_k^{\mathcal P,\mathcal C,m}-
	\dot{\bm H}_k^{(j)}\bm r^{1:n_{\omega_k}^{(j)}}}_{\blacksquare}\right)^{\mathrm T}
\left(\mathtt s\mathcal{S}\left(\theta^{n_{\omega_k}^{(j)}},\ddot l_1^{n_{\omega_k}^{(j)}},
\ddot l_2^{n_{\omega_k}^{(j)}}\right)\right)^{-1}\left(\bm y_k^{\mathcal P,\mathcal C,m}-
\dot{\bm H}_k^{(j)}\bm r^{1:n_{\omega_k}^{(j)}}\right)
\Bigg]+\mathtt c_{\bm r}\\
&\quad=\mathlarger{\blacktriangle} -
\frac{1}{2}
\sum_{m=1}^{M_k^{\mathcal P,\mathcal C}}
\mathbb E_{q_k^{\bm y^m,(j,\mathcal P,\mathcal C),[\ell]}}
\!\Bigg[
\left(\blacksquare\right)^{\mathrm T}
\mathbb E_{\substack{
		q_k^{\theta,(j,\mathcal P,\mathcal C),[\ell]}
		\prod_{u=1}^{2}
		q_k^{\ddot l_u,(j,\mathcal P,\mathcal C),[\ell]}}}
\!\left[
\left(\mathtt s\mathcal R\left(\theta^{n_{\omega_k}^{(j)}}\right)\mathcal W\left(\ddot l_1^{n_{\omega_k}^{(j)}},
\ddot l_2^{n_{\omega_k}^{(j)}}\right)\mathcal R\left(-\theta^{n_{\omega_k}^{(j)}}\right)\right)^{-1}
\right]
\left(\blacksquare\right)
\Bigg]
+\mathtt c_{\bm r}
\end{aligned}
\label{eq:C6_proof_r_dependent_terms}
\end{equation}

Since~$\mathcal R\!\left(-\theta^{n_{\omega_k}^{(j)}}\right)
\mathcal R\!\left(\theta^{n_{\omega_k}^{(j)}}\right)=\bm I_2$, the inverse
terminal extent scatter covariance is
\begin{equation}
\begin{aligned}
\left(
\mathtt s\mathcal R\left(\theta^{n_{\omega_k}^{(j)}}\right)\mathcal W\left(\ddot l_1^{n_{\omega_k}^{(j)}},
\ddot l_2^{n_{\omega_k}^{(j)}}\right)\mathcal R\left(-\theta^{n_{\omega_k}^{(j)}}\right)
\right)^{-1}
&=
\mathcal R\left(\theta^{n_{\omega_k}^{(j)}}\right)
\operatorname{diag}\!\left(
(\mathtt s\ddot l_1^{n_{\omega_k}^{(j)}})^{-1},
(\mathtt s\ddot l_2^{n_{\omega_k}^{(j)}})^{-1}
\right)
\mathcal R\left(-\theta^{n_{\omega_k}^{(j)}}\right)
\end{aligned}
\label{eq:C6_proof_r_extent_precision}
\end{equation}

The reciprocal moment of the inverse-Gamma distribution gives
\begin{equation}
\begin{aligned}
\iota_{u,\pi_k}^{(j,\mathcal P,\mathcal C),
	n_{\pi_k}^{(j)},[\ell]}
=
\mathbb E_{q_k^{\ddot l_u,(j,\mathcal P,\mathcal C),[\ell]}}
\!\left[
(\mathtt s\ddot l_u^{n_{\omega_k}^{(j)}})^{-1}
\right]=\mathbb E_{q_k^{\ddot l_u,(j,\mathcal P,\mathcal C),[\ell]}}
\!\left[
(\mathtt s\ddot l_u^{n_{\pi_k}^{(j)}})^{-1}
\right]
=
\frac{
	a_{u,\pi_k}^{(j,\mathcal P,\mathcal C),n_{\pi_k}^{(j)}}
}{
\mathtt s b_{u,\pi_k}^{(j,\mathcal P,\mathcal C),
	n_{\pi_k}^{(j)},[\ell]}
},
\qquad u\in\{1,2\}
\end{aligned}
\label{eq:C6_proof_axis_inv_expectation}
\end{equation}

\begingroup

Substituting Eq.~\eqref{eq:C6_proof_r_extent_precision} into the
expectation in Eq.~\eqref{eq:C6_proof_r_dependent_terms} and then
applying the structured variational factorization in
Eq.~\eqref{eq:C6_traj_vb_factor} together with
Eq.~\eqref{eq:C6_proof_axis_inv_expectation} gives
\endgroup
\begin{equation}
\begin{aligned}
&\mathbb E_{\substack{
		q_k^{\theta,(j,\mathcal P,\mathcal C),[\ell]}
		\prod_{u=1}^{2}
		q_k^{\ddot l_u,(j,\mathcal P,\mathcal C),[\ell]}}}
\!\left[
\left(
\mathtt s \mathcal R\left(\theta^{n_{\omega_k}^{(j)}}\right)\mathcal W\left(\ddot l_1^{n_{\omega_k}^{(j)}},
\ddot l_2^{n_{\omega_k}^{(j)}}\right)\mathcal R\left(-\theta^{n_{\omega_k}^{(j)}}\right)
\right)^{-1}
\right]\\
&\quad=
\mathbb E_{\substack{
		q_k^{\theta,(j,\mathcal P,\mathcal C),[\ell]}
		\prod_{u=1}^{2}
		q_k^{\ddot l_u,(j,\mathcal P,\mathcal C),[\ell]}}}
\!\left[
\mathcal R\left(\theta^{n_{\omega_k}^{(j)}}\right)
\operatorname{diag}\!\left(
(\mathtt s\ddot l_1^{n_{\omega_k}^{(j)}})^{-1},
(\mathtt s\ddot l_2^{n_{\omega_k}^{(j)}})^{-1}
\right)
\mathcal R\left(-\theta^{n_{\omega_k}^{(j)}}\right)
\right]\\
&\quad=
\mathbb E_{q_k^{\theta,(j,\mathcal P,\mathcal C),[\ell]}}
\!\Bigg[
\mathcal R\left(\theta^{n_{\omega_k}^{(j)}}\right)
\mathbb E_{\prod_{u=1}^{2}
	q_k^{\ddot l_u,(j,\mathcal P,\mathcal C),[\ell]}}
\!\Bigg[
\operatorname{diag}\!\left(
(\mathtt s\ddot l_1^{n_{\omega_k}^{(j)}})^{-1},
(\mathtt s\ddot l_2^{n_{\omega_k}^{(j)}})^{-1}
\right)
\Bigg]
\mathcal R\left(-\theta^{n_{\omega_k}^{(j)}}\right)
\Bigg]\\
&\quad=
\mathbb E_{q_k^{\theta,(j,\mathcal P,\mathcal C),[\ell]}}
\!\left[
\mathcal R\left(\theta^{n_{\omega_k}^{(j)}}\right)
\operatorname{diag}\!\left(
\iota_{1,\pi_k}^{(j,\mathcal P,\mathcal C),
	n_{\pi_k}^{(j)},[\ell]},
\iota_{2,\pi_k}^{(j,\mathcal P,\mathcal C),
	n_{\pi_k}^{(j)},[\ell]}
\right)
\mathcal R\left(-\theta^{n_{\omega_k}^{(j)}}\right)
\right]\\
&\quad=
\bar{\bm G}_{\pi_k}^{(j,\mathcal P,\mathcal C),[\ell]}
\end{aligned}
\label{eq:C6_proof_r_expected_precision}
\end{equation}

Thus, the expectation with respect to the $m$-th measurement-source
variable satisfies
\begin{equation}
\begin{aligned}
&\mathbb E_{q_k^{\bm y^m,(j,\mathcal P,\mathcal C),[\ell]}}
\!\Bigg[
\left(
\bm y_k^{\mathcal P,\mathcal C,m}-
\dot{\bm H}_k^{(j)}
\bm r^{1:n_{\omega_k}^{(j)}}
\right)^{\mathrm T}
\bar{\bm G}_{\pi_k}^{(j,\mathcal P,\mathcal C),[\ell]}
\left(
\bm y_k^{\mathcal P,\mathcal C,m}-
\dot{\bm H}_k^{(j)}
\bm r^{1:n_{\omega_k}^{(j)}}
\right)
\Bigg]\\
&\quad=
\left(
\hat{\bm y}_{\pi_k}^{(j,\mathcal P,\mathcal C),m,[\ell]}-
\dot{\bm H}_k^{(j)}
\bm r^{1:n_{\omega_k}^{(j)}}
\right)^{\mathrm T}
\bar{\bm G}_{\pi_k}^{(j,\mathcal P,\mathcal C),[\ell]}
\left(
\hat{\bm y}_{\pi_k}^{(j,\mathcal P,\mathcal C),m,[\ell]}-
\dot{\bm H}_k^{(j)}
\bm r^{1:n_{\omega_k}^{(j)}}
\right)+
\operatorname{tr}\!\left[
\bar{\bm G}_{\pi_k}^{(j,\mathcal P,\mathcal C),[\ell]}
\bm\varXi_{\pi_k}^{\bm y,(j,\mathcal P,\mathcal C),m,[\ell]}
\right]
\end{aligned}
\label{eq:C6_proof_r_source_expectation}
\end{equation}

Define the arithmetic mean of the posterior measurement-source means as
\begin{equation}
\bar{\bm y}_{\pi_k}^{(j,\mathcal P,\mathcal C),[\ell]}
=
\frac{1}{M_k^{\mathcal P,\mathcal C}}
\sum_{m=1}^{M_k^{\mathcal P,\mathcal C}}
\hat{\bm y}_{\pi_k}^{(j,\mathcal P,\mathcal C),m,[\ell]}
\label{eq:C6_proof_r_equivalent_mean}
\end{equation}

Equation~\eqref{eq:C6_proof_r_equivalent_mean} implies
\begin{equation}
\begin{aligned}
\sum_{m=1}^{M_k^{\mathcal P,\mathcal C}}
\left(
\hat{\bm y}_{\pi_k}^{(j,\mathcal P,\mathcal C),m,[\ell]}
-
\bar{\bm y}_{\pi_k}^{(j,\mathcal P,\mathcal C),[\ell]}
\right)
&=
\sum_{m=1}^{M_k^{\mathcal P,\mathcal C}}
\hat{\bm y}_{\pi_k}^{(j,\mathcal P,\mathcal C),m,[\ell]}-
M_k^{\mathcal P,\mathcal C}
\bar{\bm y}_{\pi_k}^{(j,\mathcal P,\mathcal C),[\ell]}
=
\bm 0
\end{aligned}
\label{eq:C6_proof_r_centered_mean_sum}
\end{equation}

By adding and subtracting
$\bar{\bm y}_{\pi_k}^{(j,\mathcal P,\mathcal C),[\ell]}$ in each
residual and expanding the quadratic form, it follows that
\begin{equation}
\begin{aligned}
&\sum_{m=1}^{M_k^{\mathcal P,\mathcal C}}
\left(
\hat{\bm y}_{\pi_k}^{(j,\mathcal P,\mathcal C),m,[\ell]}
-
\dot{\bm H}_k^{(j)}
\bm r^{1:n_{\omega_k}^{(j)}}
\right)^{\mathrm T}
\bar{\bm G}_{\pi_k}^{(j,\mathcal P,\mathcal C),[\ell]}
\left(
\hat{\bm y}_{\pi_k}^{(j,\mathcal P,\mathcal C),m,[\ell]}
-
\dot{\bm H}_k^{(j)}
\bm r^{1:n_{\omega_k}^{(j)}}
\right)\\
&\quad=
\sum_{m=1}^{M_k^{\mathcal P,\mathcal C}}
\Bigg[
\left(
\hat{\bm y}_{\pi_k}^{(j,\mathcal P,\mathcal C),m,[\ell]}
-
\bar{\bm y}_{\pi_k}^{(j,\mathcal P,\mathcal C),[\ell]}
+
\bar{\bm y}_{\pi_k}^{(j,\mathcal P,\mathcal C),[\ell]}
-
\dot{\bm H}_k^{(j)}
\bm r^{1:n_{\omega_k}^{(j)}}
\right)^{\mathrm T}\\
&\hspace{14mm}\times
\bar{\bm G}_{\pi_k}^{(j,\mathcal P,\mathcal C),[\ell]}
\left(
\hat{\bm y}_{\pi_k}^{(j,\mathcal P,\mathcal C),m,[\ell]}
-
\bar{\bm y}_{\pi_k}^{(j,\mathcal P,\mathcal C),[\ell]}
+
\bar{\bm y}_{\pi_k}^{(j,\mathcal P,\mathcal C),[\ell]}
-
\dot{\bm H}_k^{(j)}
\bm r^{1:n_{\omega_k}^{(j)}}
\right)
\Bigg]\\
&\quad=
\sum_{m=1}^{M_k^{\mathcal P,\mathcal C}}
\left(
\hat{\bm y}_{\pi_k}^{(j,\mathcal P,\mathcal C),m,[\ell]}
-
\bar{\bm y}_{\pi_k}^{(j,\mathcal P,\mathcal C),[\ell]}
\right)^{\mathrm T}
\bar{\bm G}_{\pi_k}^{(j,\mathcal P,\mathcal C),[\ell]}
\left(
\hat{\bm y}_{\pi_k}^{(j,\mathcal P,\mathcal C),m,[\ell]}
-
\bar{\bm y}_{\pi_k}^{(j,\mathcal P,\mathcal C),[\ell]}
\right)\\
&\qquad+
\underbrace{
	\sum_{m=1}^{M_k^{\mathcal P,\mathcal C}}
	\left(
	\hat{\bm y}_{\pi_k}^{(j,\mathcal P,\mathcal C),m,[\ell]}
	-
	\bar{\bm y}_{\pi_k}^{(j,\mathcal P,\mathcal C),[\ell]}
	\right)^{\mathrm T}
	\bar{\bm G}_{\pi_k}^{(j,\mathcal P,\mathcal C),[\ell]}
	\left(
	\bar{\bm y}_{\pi_k}^{(j,\mathcal P,\mathcal C),[\ell]}
	-
	\dot{\bm H}_k^{(j)}
	\bm r^{1:n_{\omega_k}^{(j)}}
	\right)
}_{0}\\
&\qquad+
\underbrace{
	\sum_{m=1}^{M_k^{\mathcal P,\mathcal C}}
	\left(
	\bar{\bm y}_{\pi_k}^{(j,\mathcal P,\mathcal C),[\ell]}
	-
	\dot{\bm H}_k^{(j)}
	\bm r^{1:n_{\omega_k}^{(j)}}
	\right)^{\mathrm T}
	\bar{\bm G}_{\pi_k}^{(j,\mathcal P,\mathcal C),[\ell]}
	\left(
	\hat{\bm y}_{\pi_k}^{(j,\mathcal P,\mathcal C),m,[\ell]}
	-
	\bar{\bm y}_{\pi_k}^{(j,\mathcal P,\mathcal C),[\ell]}
	\right)
}_{0}\\
&\qquad+
\sum_{m=1}^{M_k^{\mathcal P,\mathcal C}}
\left(
\bar{\bm y}_{\pi_k}^{(j,\mathcal P,\mathcal C),[\ell]}
-
\dot{\bm H}_k^{(j)}
\bm r^{1:n_{\omega_k}^{(j)}}
\right)^{\mathrm T}
\bar{\bm G}_{\pi_k}^{(j,\mathcal P,\mathcal C),[\ell]}
\left(
\bar{\bm y}_{\pi_k}^{(j,\mathcal P,\mathcal C),[\ell]}
-
\dot{\bm H}_k^{(j)}
\bm r^{1:n_{\omega_k}^{(j)}}
\right)\\
&\quad=
\sum_{m=1}^{M_k^{\mathcal P,\mathcal C}}
\left(
\hat{\bm y}_{\pi_k}^{(j,\mathcal P,\mathcal C),m,[\ell]}
-
\bar{\bm y}_{\pi_k}^{(j,\mathcal P,\mathcal C),[\ell]}
\right)^{\mathrm T}
\bar{\bm G}_{\pi_k}^{(j,\mathcal P,\mathcal C),[\ell]}
\left(
\hat{\bm y}_{\pi_k}^{(j,\mathcal P,\mathcal C),m,[\ell]}
-
\bar{\bm y}_{\pi_k}^{(j,\mathcal P,\mathcal C),[\ell]}
\right)\\
&\qquad+
M_k^{\mathcal P,\mathcal C}
\left(
\bar{\bm y}_{\pi_k}^{(j,\mathcal P,\mathcal C),[\ell]}
-
\dot{\bm H}_k^{(j)}
\bm r^{1:n_{\omega_k}^{(j)}}
\right)^{\mathrm T}
\bar{\bm G}_{\pi_k}^{(j,\mathcal P,\mathcal C),[\ell]}
\left(
\bar{\bm y}_{\pi_k}^{(j,\mathcal P,\mathcal C),[\ell]}
-
\dot{\bm H}_k^{(j)}
\bm r^{1:n_{\omega_k}^{(j)}}
\right)
\end{aligned}
\label{eq:C6_proof_r_mean_expansion}
\end{equation}

\begingroup

Substituting Eq.~\eqref{eq:C6_proof_r_source_expectation} into the
measurement-dependent term in
Eq.~\eqref{eq:C6_proof_r_dependent_terms} and then applying
Eq.~\eqref{eq:C6_proof_r_mean_expansion} gives
\endgroup
\begin{equation}
\begin{aligned}
&-\frac{1}{2}
\sum_{m=1}^{M_k^{\mathcal P,\mathcal C}}
\mathbb E_{q_k^{\bm y^m,(j,\mathcal P,\mathcal C),[\ell]}}
\!\Bigg[
\left(
\bm y_k^{\mathcal P,\mathcal C,m}
-
\dot{\bm H}_k^{(j)}
\bm r^{1:n_{\omega_k}^{(j)}}
\right)^{\mathrm T}
\bar{\bm G}_{\pi_k}^{(j,\mathcal P,\mathcal C),[\ell]}
\left(
\bm y_k^{\mathcal P,\mathcal C,m}
-
\dot{\bm H}_k^{(j)}
\bm r^{1:n_{\omega_k}^{(j)}}
\right)
\Bigg]\\
&\quad=
-\frac{1}{2}
\sum_{m=1}^{M_k^{\mathcal P,\mathcal C}}
\Bigg[
\left(
\hat{\bm y}_{\pi_k}^{(j,\mathcal P,\mathcal C),m,[\ell]}
-
\dot{\bm H}_k^{(j)}
\bm r^{1:n_{\omega_k}^{(j)}}
\right)^{\mathrm T}
\bar{\bm G}_{\pi_k}^{(j,\mathcal P,\mathcal C),[\ell]}
\left(
\hat{\bm y}_{\pi_k}^{(j,\mathcal P,\mathcal C),m,[\ell]}
-
\dot{\bm H}_k^{(j)}
\bm r^{1:n_{\omega_k}^{(j)}}
\right)+
\operatorname{tr}\!\left[
\bar{\bm G}_{\pi_k}^{(j,\mathcal P,\mathcal C),[\ell]}
\bm\varXi_{\pi_k}^{\bm y,(j,\mathcal P,\mathcal C),m,[\ell]}
\right]
\Bigg]\\
&\quad=
-\frac{1}{2}
\sum_{m=1}^{M_k^{\mathcal P,\mathcal C}}
\left(
\hat{\bm y}_{\pi_k}^{(j,\mathcal P,\mathcal C),m,[\ell]}
-
\bar{\bm y}_{\pi_k}^{(j,\mathcal P,\mathcal C),[\ell]}
\right)^{\mathrm T}
\bar{\bm G}_{\pi_k}^{(j,\mathcal P,\mathcal C),[\ell]}
\left(
\hat{\bm y}_{\pi_k}^{(j,\mathcal P,\mathcal C),m,[\ell]}
-
\bar{\bm y}_{\pi_k}^{(j,\mathcal P,\mathcal C),[\ell]}
\right)\\
&\qquad-
\frac{1}{2}
M_k^{\mathcal P,\mathcal C}
\left(
\bar{\bm y}_{\pi_k}^{(j,\mathcal P,\mathcal C),[\ell]}
-
\dot{\bm H}_k^{(j)}
\bm r^{1:n_{\omega_k}^{(j)}}
\right)^{\mathrm T}
\bar{\bm G}_{\pi_k}^{(j,\mathcal P,\mathcal C),[\ell]}
\left(
\bar{\bm y}_{\pi_k}^{(j,\mathcal P,\mathcal C),[\ell]}
-
\dot{\bm H}_k^{(j)}
\bm r^{1:n_{\omega_k}^{(j)}}
\right)\\
&\qquad-
\frac{1}{2}
\sum_{m=1}^{M_k^{\mathcal P,\mathcal C}}
\operatorname{tr}\!\left[
\bar{\bm G}_{\pi_k}^{(j,\mathcal P,\mathcal C),[\ell]}
\bm\varXi_{\pi_k}^{\bm y,(j,\mathcal P,\mathcal C),m,[\ell]}
\right]
\end{aligned}
\label{eq:C6_proof_r_measurement_contribution}
\end{equation}

\begingroup

The first and third terms on the right-hand side of
Eq.~\eqref{eq:C6_proof_r_measurement_contribution} are independent of
$\bm r^{1:n_{\omega_k}^{(j)}}$. Substituting
Eq.~\eqref{eq:C6_proof_r_measurement_contribution} into
Eq.~\eqref{eq:C6_proof_r_dependent_terms} and absorbing all terms
independent of~$\bm r^{1:n_{\omega_k}^{(j)}}$ into~$\mathtt c_{\bm r}$
gives the following equivalent quadratic form:
\endgroup
\begin{equation}
\begin{aligned}
\ln q_k^{\bm r,(j,\mathcal P,\mathcal C),[\ell+1]}
\!\left(\bm r^{1:n_{\pi_k}^{(j)}}\right)
&=-\frac{1}{2}
\left(\bm r^{1:n_{\omega_k}^{(j)}}-
\hat{\bm\varGamma}_{\omega_k}^{(j)}\right)^{\mathrm T}
\left(\bm\varXi_{\omega_k}^{\bm\varGamma,(j)}\right)^{-1}
\left(\bm r^{1:n_{\omega_k}^{(j)}}-
\hat{\bm\varGamma}_{\omega_k}^{(j)}\right)\\
&\quad-\frac{1}{2}M_k^{\mathcal P,\mathcal C}
\left(\bar{\bm y}_{\pi_k}^{(j,\mathcal P,\mathcal C),[\ell]}-
\dot{\bm H}_k^{(j)}\bm r^{1:n_{\omega_k}^{(j)}}\right)^{\mathrm T}
\bar{\bm G}_{\pi_k}^{(j,\mathcal P,\mathcal C),[\ell]}
\left(\bar{\bm y}_{\pi_k}^{(j,\mathcal P,\mathcal C),[\ell]}-
\dot{\bm H}_k^{(j)}\bm r^{1:n_{\omega_k}^{(j)}}\right)
+\mathtt c_{\bm r}
\end{aligned}
\label{eq:C6_proof_r_mean_quadratic_form}
\end{equation}

Expanding the two quadratic forms in
Eq.~\eqref{eq:C6_proof_r_mean_quadratic_form} and retaining the terms
that depend on~$\bm r^{1:n_{\pi_k}^{(j)}}$ gives
\begin{equation}
\begin{aligned}
\ln q_k^{\bm r,(j,\mathcal P,\mathcal C),[\ell+1]}
\!\left(\bm r^{1:n_{\pi_k}^{(j)}}\right)&=-\frac{1}{2}
\left(\bm r^{1:n_{\pi_k}^{(j)}}\right)^{\mathrm T}
\Bigg[
\left(\bm\varXi_{\omega_k}^{\bm\varGamma,(j)}\right)^{-1}
+M_k^{\mathcal P,\mathcal C}
(\dot{\bm H}_k^{(j)})^{\mathrm T}
\bar{\bm G}_{\pi_k}^{(j,\mathcal P,\mathcal C),[\ell]}
\dot{\bm H}_k^{(j)}
\Bigg]
\bm r^{1:n_{\pi_k}^{(j)}}\\
&\quad+
\left(\bm r^{1:n_{\pi_k}^{(j)}}\right)^{\mathrm T}
\Bigg[
\left(\bm\varXi_{\omega_k}^{\bm\varGamma,(j)}\right)^{-1}
\hat{\bm\varGamma}_{\omega_k}^{(j)}
+M_k^{\mathcal P,\mathcal C}
(\dot{\bm H}_k^{(j)})^{\mathrm T}
\bar{\bm G}_{\pi_k}^{(j,\mathcal P,\mathcal C),[\ell]}
\bar{\bm y}_{\pi_k}^{(j,\mathcal P,\mathcal C),[\ell]}
\Bigg]
+\mathtt c_{\bm r}
\end{aligned}
\label{eq:C6_proof_r_quadratic_form}
\end{equation}

Equation~\eqref{eq:C6_proof_r_quadratic_form} is quadratic in the
trajectory kinematic state, and its quadratic-term coefficient matrix is
positive definite. Therefore, its variational posterior is Gaussian. Since the measurement cell is nonempty,
$M_k^{\mathcal P,\mathcal C}>0$. The inverse-Gamma parameters are
positive, so
$\iota_{u,\pi_k}^{(j,\mathcal P,\mathcal C),
	n_{\pi_k}^{(j)},[\ell]}>0$ for~$u\in\{1,2\}$. Hence,
$\bar{\bm G}_{\pi_k}^{(j,\mathcal P,\mathcal C),[\ell]}$ is positive
definite. Therefore,
$M_k^{\mathcal P,\mathcal C}(\dot{\bm H}_k^{(j)})^{\mathrm T}
\bar{\bm G}_{\pi_k}^{(j,\mathcal P,\mathcal C),[\ell]}
\dot{\bm H}_k^{(j)}$ is positive semidefinite, while
$\left(\bm\varXi_{\omega_k}^{\bm\varGamma,(j)}\right)^{-1}$ is positive
definite.

The information form of the posterior Gaussian density is
\begin{equation}
\begin{aligned}
\ln q_k^{\bm r,(j,\mathcal P,\mathcal C),[\ell+1]}
\!\left(\bm r^{1:n_{\pi_k}^{(j)}}\right)&=
-\frac{1}{2}
\left(\bm r^{1:n_{\pi_k}^{(j)}}\right)^{\mathrm T}
\left(
\bm\varXi_{\pi_k}^{\bm\varGamma,
	(j,\mathcal P,\mathcal C),[\ell+1]}
\right)^{-1}
\bm r^{1:n_{\pi_k}^{(j)}}\\
&\quad+
\left(\bm r^{1:n_{\pi_k}^{(j)}}\right)^{\mathrm T}
\left(
\bm\varXi_{\pi_k}^{\bm\varGamma,
	(j,\mathcal P,\mathcal C),[\ell+1]}
\right)^{-1}
\hat{\bm\varGamma}_{\pi_k}^{(j,\mathcal P,\mathcal C),[\ell+1]}
+\mathtt c_{\bm r}
\end{aligned}
\label{eq:C6_proof_r_posterior_information_form}
\end{equation}

\begingroup

Comparing Eq.~\eqref{eq:C6_proof_r_quadratic_form} with the Gaussian
information form in Eq.~\eqref{eq:C6_proof_r_posterior_information_form}
gives
\endgroup
\begin{equation}
\begin{aligned}
\left(
\bm\varXi_{\pi_k}^{\bm\varGamma,
	(j,\mathcal P,\mathcal C),[\ell+1]}
\right)^{-1}=\left(\bm\varXi_{\omega_k}^{\bm\varGamma,(j)}\right)^{-1}
+M_k^{\mathcal P,\mathcal C}
(\dot{\bm H}_k^{(j)})^{\mathrm T}
\bar{\bm G}_{\pi_k}^{(j,\mathcal P,\mathcal C),[\ell]}
\dot{\bm H}_k^{(j)}
\end{aligned}
\label{eq:C6_proof_r_information_matrix}
\end{equation}
\begin{equation}
\begin{aligned}
\left(
\bm\varXi_{\pi_k}^{\bm\varGamma,
	(j,\mathcal P,\mathcal C),[\ell+1]}
\right)^{-1}
\hat{\bm\varGamma}_{\pi_k}^{(j,\mathcal P,\mathcal C),[\ell+1]}
&=
\left(\bm\varXi_{\omega_k}^{\bm\varGamma,(j)}\right)^{-1}
\hat{\bm\varGamma}_{\omega_k}^{(j)}+
M_k^{\mathcal P,\mathcal C}
(\dot{\bm H}_k^{(j)})^{\mathrm T}
\bar{\bm G}_{\pi_k}^{(j,\mathcal P,\mathcal C),[\ell]}
\bar{\bm y}_{\pi_k}^{(j,\mathcal P,\mathcal C),[\ell]}
\end{aligned}
\label{eq:C6_proof_r_information_vector}
\end{equation}

Applying the matrix inversion lemma to
Eq.~\eqref{eq:C6_proof_r_information_matrix} gives
\begin{equation}
\begin{aligned}
\bm\varXi_{\pi_k}^{\bm\varGamma,
	(j,\mathcal P,\mathcal C),[\ell+1]}
&=
\Bigg[
\left(
\bm\varXi_{\omega_k}^{\bm\varGamma,(j)}
\right)^{-1}
+
(\dot{\bm H}_k^{(j)})^{\mathrm T}
\left[
M_k^{\mathcal P,\mathcal C}
\bar{\bm G}_{\pi_k}^{(j,\mathcal P,\mathcal C),[\ell]}
\right]
\dot{\bm H}_k^{(j)}
\Bigg]^{-1}\\
&\quad=
\bm\varXi_{\omega_k}^{\bm\varGamma,(j)}
-
\bm\varXi_{\omega_k}^{\bm\varGamma,(j)}
(\dot{\bm H}_k^{(j)})^{\mathrm T}
\Bigg[
\left[
M_k^{\mathcal P,\mathcal C}
\bar{\bm G}_{\pi_k}^{(j,\mathcal P,\mathcal C),[\ell]}
\right]^{-1}+
\dot{\bm H}_k^{(j)}
\bm\varXi_{\omega_k}^{\bm\varGamma,(j)}
(\dot{\bm H}_k^{(j)})^{\mathrm T}
\Bigg]^{-1}
\dot{\bm H}_k^{(j)}
\bm\varXi_{\omega_k}^{\bm\varGamma,(j)}\\
&\quad=
\bm\varXi_{\omega_k}^{\bm\varGamma,(j)}
-
\bm\varXi_{\omega_k}^{\bm\varGamma,(j)}
(\dot{\bm H}_k^{(j)})^{\mathrm T}
\Bigg[
\dot{\bm H}_k^{(j)}
\bm\varXi_{\omega_k}^{\bm\varGamma,(j)}
(\dot{\bm H}_k^{(j)})^{\mathrm T}+
\left[
M_k^{\mathcal P,\mathcal C}
\bar{\bm G}_{\pi_k}^{(j,\mathcal P,\mathcal C),[\ell]}
\right]^{-1}
\Bigg]^{-1}
\dot{\bm H}_k^{(j)}
\bm\varXi_{\omega_k}^{\bm\varGamma,(j)}
\end{aligned}
\label{eq:C6_proof_r_covariance_result_expanded}
\end{equation}

The two terms in Eq.~\eqref{eq:C6_proof_r_covariance_result_expanded},
$\dot{\bm H}_k^{(j)}
\bm\varXi_{\omega_k}^{\bm\varGamma,(j)}
(\dot{\bm H}_k^{(j)})^{\mathrm T}$ and
$\left[
M_k^{\mathcal P,\mathcal C}
\bar{\bm G}_{\pi_k}^{(j,\mathcal P,\mathcal C),[\ell]}
\right]^{-1}$, are respectively the covariance matrices of the
predicted measurement and the equivalent measurement error. Therefore, their sum is the
covariance matrix of the equivalent measurement innovation, defined as

\begin{equation}
\begin{aligned}
\bm\varXi_{\pi_k}^{\bar{\bm y},
	(j,\mathcal P,\mathcal C),[\ell]}
=
\dot{\bm H}_k^{(j)}
\bm\varXi_{\omega_k}^{\bm\varGamma,(j)}
(\dot{\bm H}_k^{(j)})^{\mathrm T}+
\left[
M_k^{\mathcal P,\mathcal C}
\bar{\bm G}_{\pi_k}^{(j,\mathcal P,\mathcal C),[\ell]}
\right]^{-1}
\end{aligned}
\label{eq:C6_proof_r_innovation_cov}
\end{equation}

Substituting Eq.~\eqref{eq:C6_proof_r_innovation_cov} into
Eq.~\eqref{eq:C6_proof_r_covariance_result_expanded} gives
\begin{equation}
\begin{aligned}
\bm\varXi_{\pi_k}^{\bm\varGamma,
	(j,\mathcal P,\mathcal C),[\ell+1]}=
\bm\varXi_{\omega_k}^{\bm\varGamma,(j)}
-
\bm\varXi_{\omega_k}^{\bm\varGamma,(j)}
(\dot{\bm H}_k^{(j)})^{\mathrm T}
\left(
\bm\varXi_{\pi_k}^{\bar{\bm y},
	(j,\mathcal P,\mathcal C),[\ell]}
\right)^{-1}
\dot{\bm H}_k^{(j)}
\bm\varXi_{\omega_k}^{\bm\varGamma,(j)}
\end{aligned}
\label{eq:C6_proof_r_covariance_result}
\end{equation}

Define the trajectory kinematic-state gain matrix as
\begin{equation}
\bm K_{\pi_k}^{(j,\mathcal P,\mathcal C),[\ell+1]}
=
\bm\varXi_{\omega_k}^{\bm\varGamma,(j)}
(\dot{\bm H}_k^{(j)})^{\mathrm T}
\left(
\bm\varXi_{\pi_k}^{\bar{\bm y},
	(j,\mathcal P,\mathcal C),[\ell]}
\right)^{-1}
\label{eq:C6_proof_r_gain_result}
\end{equation}

Right-multiplying Eq.~\eqref{eq:C6_proof_r_covariance_result} by
$\left(\bm\varXi_{\omega_k}^{\bm\varGamma,(j)}\right)^{-1}$ and using
Eq.~\eqref{eq:C6_proof_r_gain_result} gives

\begin{equation}
\begin{aligned}
\bm\varXi_{\pi_k}^{\bm\varGamma,
	(j,\mathcal P,\mathcal C),[\ell+1]}
\left(
\bm\varXi_{\omega_k}^{\bm\varGamma,(j)}
\right)^{-1}=
\bm I
-\bm K_{\pi_k}^{(j,\mathcal P,\mathcal C),[\ell+1]}
\dot{\bm H}_k^{(j)}
\end{aligned}
\label{eq:C6_proof_r_information_identity_1}
\end{equation}

The matrix inversion lemma applied to
Eq.~\eqref{eq:C6_proof_r_information_matrix} also gives

\begin{equation}
\begin{aligned}
\bm\varXi_{\pi_k}^{\bm\varGamma,
	(j,\mathcal P,\mathcal C),[\ell+1]}
(\dot{\bm H}_k^{(j)})^{\mathrm T}
M_k^{\mathcal P,\mathcal C}
\bar{\bm G}_{\pi_k}^{(j,\mathcal P,\mathcal C),[\ell]}&=
\bm\varXi_{\omega_k}^{\bm\varGamma,(j)}
(\dot{\bm H}_k^{(j)})^{\mathrm T}
\left(
\bm\varXi_{\pi_k}^{\bar{\bm y},
	(j,\mathcal P,\mathcal C),[\ell]}
\right)^{-1}\\
&\quad=
\bm K_{\pi_k}^{(j,\mathcal P,\mathcal C),[\ell+1]}
\end{aligned}
\label{eq:C6_proof_r_information_identity_2}
\end{equation}

Left-multiplying Eq.~\eqref{eq:C6_proof_r_information_vector} by
$\bm\varXi_{\pi_k}^{\bm\varGamma,
	(j,\mathcal P,\mathcal C),[\ell+1]}$ gives
\begin{equation}
\begin{aligned}
\hat{\bm\varGamma}_{\pi_k}^{
	(j,\mathcal P,\mathcal C),[\ell+1]}
&=
\bm\varXi_{\pi_k}^{\bm\varGamma,
	(j,\mathcal P,\mathcal C),[\ell+1]}
\left(\bm\varXi_{\omega_k}^{\bm\varGamma,(j)}\right)^{-1}
\hat{\bm\varGamma}_{\omega_k}^{(j)}+
\bm\varXi_{\pi_k}^{\bm\varGamma,
	(j,\mathcal P,\mathcal C),[\ell+1]}
(\dot{\bm H}_k^{(j)})^{\mathrm T}
M_k^{\mathcal P,\mathcal C}
\bar{\bm G}_{\pi_k}^{(j,\mathcal P,\mathcal C),[\ell]}
\bar{\bm y}_{\pi_k}^{(j,\mathcal P,\mathcal C),[\ell]}\\
&\quad=
\left[
\bm I-
\bm K_{\pi_k}^{(j,\mathcal P,\mathcal C),[\ell+1]}
\dot{\bm H}_k^{(j)}
\right]
\hat{\bm\varGamma}_{\omega_k}^{(j)}
+
\bm K_{\pi_k}^{(j,\mathcal P,\mathcal C),[\ell+1]}
\bar{\bm y}_{\pi_k}^{(j,\mathcal P,\mathcal C),[\ell]}\\
&\quad=
\hat{\bm\varGamma}_{\omega_k}^{(j)}
+
\bm K_{\pi_k}^{(j,\mathcal P,\mathcal C),[\ell+1]}
\left[
\bar{\bm y}_{\pi_k}^{(j,\mathcal P,\mathcal C),[\ell]}
-
\dot{\bm H}_k^{(j)}
\hat{\bm\varGamma}_{\omega_k}^{(j)}
\right]
\end{aligned}
\label{eq:C6_proof_r_mean_result}
\end{equation}

Substituting Eq.~\eqref{eq:C6_proof_r_gain_result} into
Eq.~\eqref{eq:C6_proof_r_covariance_result} gives
\begin{equation}
\bm\varXi_{\pi_k}^{\bm\varGamma,
	(j,\mathcal P,\mathcal C),[\ell+1]}
=
\bm\varXi_{\omega_k}^{\bm\varGamma,(j)}
-
\bm K_{\pi_k}^{(j,\mathcal P,\mathcal C),[\ell+1]}
\dot{\bm H}_k^{(j)}
\bm\varXi_{\omega_k}^{\bm\varGamma,(j)}
\label{eq:C6_proof_r_covariance_compact}
\end{equation}

Therefore, Eqs.~\eqref{eq:C6_proof_r_mean_result} and
\eqref{eq:C6_proof_r_covariance_compact} prove
Proposition~\ref{prop:C6_r_update}.

\newpage
\section{Proof of Proposition~\ref{prop:C6_axis_update}}
\label{app:C6_axis_update_proof}

The derivation is presented for the first posterior SSAL state sequence.
The derivation for the second posterior SSAL state sequence is identical.
According to the structured variational factorization in
Eq.~\eqref{eq:C6_traj_vb_factor} and the CAVI coordinate-update rule, the
variational factor of the first posterior SSAL state sequence satisfies
\begin{equation}
\begin{aligned}
&\ln q_k^{\ddot l_1,(j,\mathcal P,\mathcal C),[\ell+1]}
\!\left(\ddot l_1^{1:n_{\pi_k}^{(j)}}\right)=\mathbb E_{\substack{
q_k^{\bm r,(j,\mathcal P,\mathcal C),[\ell+1]}
q_k^{\theta,(j,\mathcal P,\mathcal C),[\ell]}
q_k^{\ddot l_2,(j,\mathcal P,\mathcal C),[\ell]}\\
{}\times\prod_{m=1}^{M_k^{\mathcal P,\mathcal C}}
q_k^{\bm y^m,(j,\mathcal P,\mathcal C),[\ell]}}}
\!\Bigg[
\ln p_k\!\left(
\begin{gathered}
\bm r^{1:n_{\omega_k}^{(j)}},
\bm\theta^{1:n_{\omega_k}^{(j)}},
\ddot l_1^{1:n_{\omega_k}^{(j)}},
\ddot l_2^{1:n_{\omega_k}^{(j)}},\\
\mathcal Y_k^{(\mathcal P,\mathcal C)},
\mathcal Z_k^{(\mathcal P,\mathcal C)}
\end{gathered}
\right)
\Bigg]
+\mathtt c_{l_1}
\end{aligned}
\label{eq:C6_proof_axis_CAVI}
\end{equation}
\noindent where the expectation is taken with respect to the variational
posteriors of the trajectory kinematic state sequence at the
$(\ell+1)$-th iteration and the trajectory orientation state sequence,
second trajectory SSAL state sequence, and all measurement-source
variables at the $\ell$-th iteration, excluding the variational factor of
the first trajectory SSAL state sequence being updated;
$\mathtt c_{l_1}$ denotes a constant term independent of
$\ddot l_1^{1:n_{\pi_k}^{(j)}}$. Since the current measurements affect only the terminal
trajectory state, Eq.~\eqref{eq:C6_proof_axis_CAVI} first updates the
terminal inverse-Gamma marginal. The historical marginals are then
obtained by the Beta--Bartlett backward recursion under the Markov joint
model in Eq.~\eqref{eq:C6_SSAL_Markov_joint} and projected back onto the
DGDIG form in Eq.~\eqref{eq:C6_single_component}.

The predicted terminal marginal density of the first SSAL state is
\begin{equation}
p_{\omega_k}^{\ddot l_1,(j)}
\!\left(\ddot l_1^{n_{\omega_k}^{(j)}}\right)
=\mathcal{IG}\!\left(
\ddot l_1^{n_{\omega_k}^{(j)}};
a_{1,\omega_k}^{(j),n_{\omega_k}^{(j)}},
b_{1,\omega_k}^{(j),n_{\omega_k}^{(j)}}
\right)
\label{eq:C6_proof_axis_terminal_prior}
\end{equation}

To expose the measurement-dependent part of
Eq.~\eqref{eq:C6_proof_axis_CAVI}, separate the measurement-source
conditional densities from the remaining terms in the joint density.
The black dot is defined at its first occurrence by the complete
non-measurement term under the underbrace below. Using
Eq.~\eqref{eq:C6_proof_r_source_density2} gives

\begingroup
\allowdisplaybreaks[4]
\begin{align}
&\ln q_k^{\ddot l_1,(j,\mathcal P,\mathcal C),[\ell+1]}
\!\left(\ddot l_1^{1:n_{\pi_k}^{(j)}}\right)\notag\\
&\quad=\mathbb E_{\substack{
		q_k^{\bm r,(j,\mathcal P,\mathcal C),[\ell+1]}
		q_k^{\theta,(j,\mathcal P,\mathcal C),[\ell]}
		q_k^{\ddot l_2,(j,\mathcal P,\mathcal C),[\ell]}\\
		{}\times\prod_{m=1}^{M_k^{\mathcal P,\mathcal C}}
		q_k^{\bm y^m,(j,\mathcal P,\mathcal C),[\ell]}}}
\!\left[
\underbrace{
	\ln p_k\!\left(
	\begin{gathered}
	\bm r^{1:n_{\omega_k}^{(j)}},
	\bm\theta^{1:n_{\omega_k}^{(j)}},
	\ddot l_1^{1:n_{\omega_k}^{(j)}},
	\ddot l_2^{1:n_{\omega_k}^{(j)}},\\
	\mathcal Y_k^{(\mathcal P,\mathcal C)},
	\mathcal Z_k^{(\mathcal P,\mathcal C)}
	\end{gathered}
	\right)
	-\sum_{m=1}^{M_k^{\mathcal P,\mathcal C}}
	\ln p\!\left(
	\bm y_k^{\mathcal P,\mathcal C,m}\mid
	\bm r^{1:n_{\omega_k}^{(j)}},\theta^{n_{\omega_k}^{(j)}},
	\ddot l_{1}^{n_{\omega_k}^{(j)}},
	\ddot l_{2}^{n_{\omega_k}^{(j)}}
	\right)
}_{\mathlarger{\bullet}}
\right]\notag\\
&\qquad+
\mathbb E_{\substack{
		q_k^{\bm r,(j,\mathcal P,\mathcal C),[\ell+1]}
		q_k^{\theta,(j,\mathcal P,\mathcal C),[\ell]}
		q_k^{\ddot l_2,(j,\mathcal P,\mathcal C),[\ell]}\\
		{}\times\prod_{m=1}^{M_k^{\mathcal P,\mathcal C}}
		q_k^{\bm y^m,(j,\mathcal P,\mathcal C),[\ell]}}}
\!\Bigg[
\sum_{m=1}^{M_k^{\mathcal P,\mathcal C}}
\ln p\!\left(
\bm y_k^{\mathcal P,\mathcal C,m}\mid
\bm r^{1:n_{\omega_k}^{(j)}},\theta^{n_{\omega_k}^{(j)}},
\ddot l_{1}^{n_{\omega_k}^{(j)}},
\ddot l_{2}^{n_{\omega_k}^{(j)}}
\right)
\Bigg]
+\mathtt c_{l_1}\notag\\
&\quad=
\mathbb E_{\substack{
		q_k^{\bm r,(j,\mathcal P,\mathcal C),[\ell+1]}
		q_k^{\theta,(j,\mathcal P,\mathcal C),[\ell]}
		q_k^{\ddot l_2,(j,\mathcal P,\mathcal C),[\ell]}\\
		{}\times\prod_{m=1}^{M_k^{\mathcal P,\mathcal C}}
		q_k^{\bm y^m,(j,\mathcal P,\mathcal C),[\ell]}}}
\!\left[\mathlarger{\bullet}\right]
-\underbrace{\frac{1}{2}
	\sum_{m=1}^{M_k^{\mathcal P,\mathcal C}}
	\mathbb E_{q_k^{\ddot l_2,(j,\mathcal P,\mathcal C),[\ell]}}
	\!\left[
	\ln\det\!\left(
	\mathtt s
	\mathcal{S}\left(\theta^{n_{\omega_k}^{(j)}},\ddot l_1^{n_{\omega_k}^{(j)}},
	\ddot l_2^{n_{\omega_k}^{(j)}}\right)
	\right)
	\right]}_{\substack{
		=\frac{1}{2}
		\sum_{m=1}^{M_k^{\mathcal P,\mathcal C}}
		\mathbb E_{q_k^{\ddot l_2,(j,\mathcal P,\mathcal C),[\ell]}}
		\!\left[
		\ln\!\left(
		\mathtt s^2
		\ddot l_{1}^{n_{\omega_k}^{(j)}}
		\ddot l_{2}^{n_{\omega_k}^{(j)}}
		\right)
		\right]\\
		=\frac{M_k^{\mathcal P,\mathcal C}}{2}
		\ln\ddot l_{1}^{n_{\omega_k}^{(j)}}
		+\underbrace{\frac{M_k^{\mathcal P,\mathcal C}}{2}\ln(\mathtt s^2)
			+\frac{M_k^{\mathcal P,\mathcal C}}{2}
			\mathbb E_{q_k^{\ddot l_2,(j,\mathcal P,\mathcal C),[\ell]}}
			\!\left[\ln\ddot l_{2}^{n_{\omega_k}^{(j)}}\right]}_{
			\text{Constant, absorbed in}\ \mathtt c_{l_1}}}}
\notag\\
&\qquad-
\frac{1}{2}
\sum_{m=1}^{M_k^{\mathcal P,\mathcal C}}
\mathbb E_{\substack{
		q_k^{\bm r,(j,\mathcal P,\mathcal C),[\ell+1]}
		q_k^{\theta,(j,\mathcal P,\mathcal C),[\ell]}
		q_k^{\ddot l_2,(j,\mathcal P,\mathcal C),[\ell]}\\
		{}\times q_k^{\bm y^m,(j,\mathcal P,\mathcal C),[\ell]}}}
\!\Bigg[
\left(
\underbrace{
	\bm y_k^{\mathcal P,\mathcal C,m}
	-\dot{\bm H}_k^{(j)}
	\bm r^{1:n_{\omega_k}^{(j)}}
}_{\blacksquare}
\right)^{\mathrm T}
\left(
\mathtt s\mathcal{S}\left(\theta^{n_{\omega_k}^{(j)}},\ddot l_1^{n_{\omega_k}^{(j)}},
\ddot l_2^{n_{\omega_k}^{(j)}}\right)
\right)^{-1}
\left(\blacksquare\right)
\Bigg]
+\mathtt c_{l_1}\notag\\
&\quad=
\mathbb E_{\substack{
		q_k^{\bm r,(j,\mathcal P,\mathcal C),[\ell+1]}
		q_k^{\theta,(j,\mathcal P,\mathcal C),[\ell]}
		q_k^{\ddot l_2,(j,\mathcal P,\mathcal C),[\ell]}\\
		{}\times\prod_{m=1}^{M_k^{\mathcal P,\mathcal C}}
		q_k^{\bm y^m,(j,\mathcal P,\mathcal C),[\ell]}}}
\!\left[\mathlarger{\bullet}\right]
-\frac{M_k^{\mathcal P,\mathcal C}}{2}
\ln\ddot l_{1}^{n_{\omega_k}^{(j)}}\notag\\
&\qquad-
\frac{1}{2}
\sum_{m=1}^{M_k^{\mathcal P,\mathcal C}}
\mathbb E_{\substack{
		q_k^{\bm r,(j,\mathcal P,\mathcal C),[\ell+1]}
		q_k^{\theta,(j,\mathcal P,\mathcal C),[\ell]}
		q_k^{\ddot l_2,(j,\mathcal P,\mathcal C),[\ell]}\\
		{}\times q_k^{\bm y^m,(j,\mathcal P,\mathcal C),[\ell]}}}
\!\Bigg[
\left(\blacksquare\right)^{\mathrm T}
\left(
\mathtt s\mathcal{S}\left(\theta^{n_{\omega_k}^{(j)}},\ddot l_1^{n_{\omega_k}^{(j)}},
\ddot l_2^{n_{\omega_k}^{(j)}}\right)
\right)^{-1}
\left(\blacksquare\right)
\Bigg]
+\mathtt c_{l_1}\notag\\
&\quad=
\mathbb E_{\substack{
		q_k^{\bm r,(j,\mathcal P,\mathcal C),[\ell+1]}
		q_k^{\theta,(j,\mathcal P,\mathcal C),[\ell]}
		q_k^{\ddot l_2,(j,\mathcal P,\mathcal C),[\ell]}\\
		{}\times\prod_{m=1}^{M_k^{\mathcal P,\mathcal C}}
		q_k^{\bm y^m,(j,\mathcal P,\mathcal C),[\ell]}}}
\!\left[\mathlarger{\bullet}\right]
-\frac{M_k^{\mathcal P,\mathcal C}}{2}
\ln\ddot l_{1}^{n_{\omega_k}^{(j)}}\notag\\
&\qquad-\underbrace{\frac{1}{2}
	\sum_{m=1}^{M_k^{\mathcal P,\mathcal C}}
	\mathbb E_{\substack{
			q_k^{\theta,(j,\mathcal P,\mathcal C),[\ell]}\\
			q_k^{\ddot l_2,(j,\mathcal P,\mathcal C),[\ell]}}}
	\!\Bigg[
	\operatorname{tr}\!\Bigg(
	\underbrace{
		\mathbb E_{\substack{
				q_k^{\bm r,(j,\mathcal P,\mathcal C),[\ell+1]}
				q_k^{\bm y^m,(j,\mathcal P,\mathcal C),[\ell]}}}
		\!\Bigg[
		\left(\blacksquare\right)
		\left(\blacksquare\right)^{\mathrm T}
		\Bigg]
	}_{\mathsf A}
	\left(
	\mathtt s\mathcal{S}\left(\theta^{n_{\omega_k}^{(j)}},\ddot l_1^{n_{\omega_k}^{(j)}},
	\ddot l_2^{n_{\omega_k}^{(j)}}\right)
	\right)^{-1}
	\Bigg)
	\Bigg]}_{\mathsf B}
+\mathtt c_{l_1}
\label{eq:C6_proof_axis_dependent_terms}
\end{align}
\endgroup

The term~$\mathsf{A}$ in
Eq.~\eqref{eq:C6_proof_axis_dependent_terms} is the second-order
moment of the residual between the measurement-source variable and the
predicted measurement. We denote this residual second-order moment by
$\bm C_k^{(j,\mathcal P,\mathcal C),m,[\ell+1]}$. It can be obtained by
expanding the residual around the posterior means of
$\bm y_k^{\mathcal P,\mathcal C,m}$ and
$\bm r^{1:n_{\omega_k}^{(j)}}$. 

The detailed expansion is
\begin{equation}
\begin{aligned}
&\mathbb E_{\substack{
q_k^{\bm r,(j,\mathcal P,\mathcal C),[\ell+1]}
q_k^{\bm y^m,(j,\mathcal P,\mathcal C),[\ell]}}}
\!\left[
\left(
\bm y_k^{\mathcal P,\mathcal C,m}
-\dot{\bm H}_k^{(j)}
\bm r^{1:n_{\omega_k}^{(j)}}
\right)
\left(
\bm y_k^{\mathcal P,\mathcal C,m}
-\dot{\bm H}_k^{(j)}
\bm r^{1:n_{\omega_k}^{(j)}}
\right)^{\mathrm T}
\right]\\
&=
\mathbb E_{\substack{
q_k^{\bm r,(j,\mathcal P,\mathcal C),[\ell+1]}
q_k^{\bm y^m,(j,\mathcal P,\mathcal C),[\ell]}}}
\!\Bigg[
\Bigg(
\begin{gathered}
\hat{\bm y}_{\pi_k}^{(j,\mathcal P,\mathcal C),m,[\ell]}
-\dot{\bm H}_k^{(j)}
\hat{\bm\varGamma}_{\pi_k}^{(j,\mathcal P,\mathcal C),[\ell+1]}\\
{}+\bm y_k^{\mathcal P,\mathcal C,m}
-\hat{\bm y}_{\pi_k}^{(j,\mathcal P,\mathcal C),m,[\ell]}\\
{}-\dot{\bm H}_k^{(j)}
\left(
\bm r^{1:n_{\omega_k}^{(j)}}
-\hat{\bm\varGamma}_{\pi_k}^{(j,\mathcal P,\mathcal C),[\ell+1]}
\right)
\end{gathered}
\Bigg)
\Bigg(
\begin{gathered}
\hat{\bm y}_{\pi_k}^{(j,\mathcal P,\mathcal C),m,[\ell]}
-\dot{\bm H}_k^{(j)}
\hat{\bm\varGamma}_{\pi_k}^{(j,\mathcal P,\mathcal C),[\ell+1]}\\
{}+\bm y_k^{\mathcal P,\mathcal C,m}
-\hat{\bm y}_{\pi_k}^{(j,\mathcal P,\mathcal C),m,[\ell]}\\
{}-\dot{\bm H}_k^{(j)}
\left(
\bm r^{1:n_{\omega_k}^{(j)}}
-\hat{\bm\varGamma}_{\pi_k}^{(j,\mathcal P,\mathcal C),[\ell+1]}
\right)
\end{gathered}
\Bigg)^{\mathrm T}
\Bigg]\\
&=
\left(
\hat{\bm y}_{\pi_k}^{(j,\mathcal P,\mathcal C),m,[\ell]}
-\dot{\bm H}_k^{(j)}
\hat{\bm\varGamma}_{\pi_k}^{(j,\mathcal P,\mathcal C),[\ell+1]}
\right)
\left(
\hat{\bm y}_{\pi_k}^{(j,\mathcal P,\mathcal C),m,[\ell]}
-\dot{\bm H}_k^{(j)}
\hat{\bm\varGamma}_{\pi_k}^{(j,\mathcal P,\mathcal C),[\ell+1]}
\right)^{\mathrm T}\\
&\quad{+}
\underbrace{
\left(
\hat{\bm y}_{\pi_k}^{(j,\mathcal P,\mathcal C),m,[\ell]}
 -\dot{\bm H}_k^{(j)}
 \hat{\bm\varGamma}_{\pi_k}^{(j,\mathcal P,\mathcal C),[\ell+1]}
\right)
\mathbb E_{q_k^{\bm y^m,(j,\mathcal P,\mathcal C),[\ell]}}
\!\left[
\left(
\bm y_k^{\mathcal P,\mathcal C,m}
 -\hat{\bm y}_{\pi_k}^{(j,\mathcal P,\mathcal C),m,[\ell]}
\right)^{\mathrm T}
\right]}_{\bm 0}\\
&\quad-\underbrace{
\left(
\hat{\bm y}_{\pi_k}^{(j,\mathcal P,\mathcal C),m,[\ell]}
 -\dot{\bm H}_k^{(j)}
 \hat{\bm\varGamma}_{\pi_k}^{(j,\mathcal P,\mathcal C),[\ell+1]}
\right)
\mathbb E_{q_k^{\bm r,(j,\mathcal P,\mathcal C),[\ell+1]}}
\!\left[
\left(
\bm r^{1:n_{\omega_k}^{(j)}}
 -\hat{\bm\varGamma}_{\pi_k}^{(j,\mathcal P,\mathcal C),[\ell+1]}
\right)^{\mathrm T}
\right]
(\dot{\bm H}_k^{(j)})^{\mathrm T}}_{\bm 0}\\
&\quad+\underbrace{
\mathbb E_{q_k^{\bm y^m,(j,\mathcal P,\mathcal C),[\ell]}}
\!\left[
\bm y_k^{\mathcal P,\mathcal C,m}
 -\hat{\bm y}_{\pi_k}^{(j,\mathcal P,\mathcal C),m,[\ell]}
\right]
\left(
\hat{\bm y}_{\pi_k}^{(j,\mathcal P,\mathcal C),m,[\ell]}
 -\dot{\bm H}_k^{(j)}
 \hat{\bm\varGamma}_{\pi_k}^{(j,\mathcal P,\mathcal C),[\ell+1]}
\right)^{\mathrm T}}_{\bm 0}\\
&\quad+
\mathbb E_{q_k^{\bm y^m,(j,\mathcal P,\mathcal C),[\ell]}}
\!\left[
\left(
\bm y_k^{\mathcal P,\mathcal C,m}
 -\hat{\bm y}_{\pi_k}^{(j,\mathcal P,\mathcal C),m,[\ell]}
\right)
\left(
\bm y_k^{\mathcal P,\mathcal C,m}
 -\hat{\bm y}_{\pi_k}^{(j,\mathcal P,\mathcal C),m,[\ell]}
\right)^{\mathrm T}
\right]\\
&\quad-\underbrace{
\mathbb E_{q_k^{\bm y^m,(j,\mathcal P,\mathcal C),[\ell]}}
\!\left[
\bm y_k^{\mathcal P,\mathcal C,m}
 -\hat{\bm y}_{\pi_k}^{(j,\mathcal P,\mathcal C),m,[\ell]}
\right]
\mathbb E_{q_k^{\bm r,(j,\mathcal P,\mathcal C),[\ell+1]}}
\!\left[
\left(
\bm r^{1:n_{\omega_k}^{(j)}}
 -\hat{\bm\varGamma}_{\pi_k}^{(j,\mathcal P,\mathcal C),[\ell+1]}
\right)^{\mathrm T}
\right]
(\dot{\bm H}_k^{(j)})^{\mathrm T}}_{\bm 0}\\
&\quad-\underbrace{
\dot{\bm H}_k^{(j)}
\mathbb E_{q_k^{\bm r,(j,\mathcal P,\mathcal C),[\ell+1]}}
\!\left[
\bm r^{1:n_{\omega_k}^{(j)}}
 -\hat{\bm\varGamma}_{\pi_k}^{(j,\mathcal P,\mathcal C),[\ell+1]}
\right]
\left(
\hat{\bm y}_{\pi_k}^{(j,\mathcal P,\mathcal C),m,[\ell]}
 -\dot{\bm H}_k^{(j)}
 \hat{\bm\varGamma}_{\pi_k}^{(j,\mathcal P,\mathcal C),[\ell+1]}
\right)^{\mathrm T}}_{\bm 0}\\
&\quad-\underbrace{
\dot{\bm H}_k^{(j)}
\mathbb E_{q_k^{\bm r,(j,\mathcal P,\mathcal C),[\ell+1]}}
\!\left[
\bm r^{1:n_{\omega_k}^{(j)}}
 -\hat{\bm\varGamma}_{\pi_k}^{(j,\mathcal P,\mathcal C),[\ell+1]}
\right]
\mathbb E_{q_k^{\bm y^m,(j,\mathcal P,\mathcal C),[\ell]}}
\!\left[
\left(
\bm y_k^{\mathcal P,\mathcal C,m}
 -\hat{\bm y}_{\pi_k}^{(j,\mathcal P,\mathcal C),m,[\ell]}
\right)^{\mathrm T}
\right]}_{\bm 0}\\
&\quad+
\dot{\bm H}_k^{(j)}
\mathbb E_{q_k^{\bm r,(j,\mathcal P,\mathcal C),[\ell+1]}}
\!\left[
\left(
\bm r^{1:n_{\omega_k}^{(j)}}
 -\hat{\bm\varGamma}_{\pi_k}^{(j,\mathcal P,\mathcal C),[\ell+1]}
\right)
\left(
\bm r^{1:n_{\omega_k}^{(j)}}
 -\hat{\bm\varGamma}_{\pi_k}^{(j,\mathcal P,\mathcal C),[\ell+1]}
\right)^{\mathrm T}
\right]
(\dot{\bm H}_k^{(j)})^{\mathrm T}\\
&= 
\left(
\hat{\bm y}_{\pi_k}^{(j,\mathcal P,\mathcal C),m,[\ell]}
 -\dot{\bm H}_k^{(j)}
 \hat{\bm\varGamma}_{\pi_k}^{(j,\mathcal P,\mathcal C),[\ell+1]}
\right)
\left(
\hat{\bm y}_{\pi_k}^{(j,\mathcal P,\mathcal C),m,[\ell]}
 -\dot{\bm H}_k^{(j)}
 \hat{\bm\varGamma}_{\pi_k}^{(j,\mathcal P,\mathcal C),[\ell+1]}
\right)^{\mathrm T}\\
&\quad+
\mathbb E_{q_k^{\bm y^m,(j,\mathcal P,\mathcal C),[\ell]}}
\!\left[
\left(
\bm y_k^{\mathcal P,\mathcal C,m}
-\hat{\bm y}_{\pi_k}^{(j,\mathcal P,\mathcal C),m,[\ell]}
\right)
\left(
\bm y_k^{\mathcal P,\mathcal C,m}
-\hat{\bm y}_{\pi_k}^{(j,\mathcal P,\mathcal C),m,[\ell]}
\right)^{\mathrm T}
\right]\\
&\quad+
\dot{\bm H}_k^{(j)}
\mathbb E_{q_k^{\bm r,(j,\mathcal P,\mathcal C),[\ell+1]}}
\!\left[
\left(
\bm r^{1:n_{\omega_k}^{(j)}}
-\hat{\bm\varGamma}_{\pi_k}^{(j,\mathcal P,\mathcal C),[\ell+1]}
\right)
\left(
\bm r^{1:n_{\omega_k}^{(j)}}
-\hat{\bm\varGamma}_{\pi_k}^{(j,\mathcal P,\mathcal C),[\ell+1]}
\right)^{\mathrm T}
\right]
(\dot{\bm H}_k^{(j)})^{\mathrm T}\\
&=
\left(
\hat{\bm y}_{\pi_k}^{(j,\mathcal P,\mathcal C),m,[\ell]}
-\dot{\bm H}_k^{(j)}
\hat{\bm\varGamma}_{\pi_k}^{(j,\mathcal P,\mathcal C),[\ell+1]}
\right)
\left(
\hat{\bm y}_{\pi_k}^{(j,\mathcal P,\mathcal C),m,[\ell]}
-\dot{\bm H}_k^{(j)}
\hat{\bm\varGamma}_{\pi_k}^{(j,\mathcal P,\mathcal C),[\ell+1]}
\right)^{\mathrm T}+
\bm\varXi_{\pi_k}^{\bm y,(j,\mathcal P,\mathcal C),m,[\ell]}
+
\dot{\bm H}_k^{(j)}
\bm\varXi_{\pi_k}^{\bm\varGamma,(j,\mathcal P,\mathcal C),[\ell+1]}
(\dot{\bm H}_k^{(j)})^{\mathrm T}\\
&=\bm C_k^{(j,\mathcal P,\mathcal C),m,[\ell+1]}
\end{aligned}
\label{eq:C6_proof_axis_residual_moment}
\end{equation}

The zero-valued terms in
Eq.~\eqref{eq:C6_proof_axis_residual_moment} follow from the zero
first-order central moments of the measurement-source variable and the
trajectory kinematic state:
\begin{equation}
\mathbb E_{q_k^{\bm y^m,(j,\mathcal P,\mathcal C),[\ell]}}
\!\left[
\bm y_k^{\mathcal P,\mathcal C,m}
-\hat{\bm y}_{\pi_k}^{(j,\mathcal P,\mathcal C),m,[\ell]}
\right]
=\bm 0
\label{eq:C6_proof_axis_source_centered_mean}
\end{equation}
and
\begin{equation}
\mathbb E_{q_k^{\bm r,(j,\mathcal P,\mathcal C),[\ell+1]}}
\!\left[
\bm r^{1:n_{\omega_k}^{(j)}}
-\hat{\bm\varGamma}_{\pi_k}^{(j,\mathcal P,\mathcal C),[\ell+1]}
\right]
=\bm 0
\label{eq:C6_proof_axis_state_centered_mean}
\end{equation}

The mixed terms also vanish because the two centered variables are
independent under the structured variational factorization.

Accumulating the second-order residual moments of all measurement-source
variables in the measurement cell and multiplying the result by
$(2\mathtt s)^{-1}$ give
\begin{equation}
\bar{\bm C}_k^{(j,\mathcal P,\mathcal C),[\ell+1]}
=\frac{1}{2\mathtt s}
\sum_{m=1}^{M_k^{\mathcal P,\mathcal C}}
\bm C_k^{(j,\mathcal P,\mathcal C),m,[\ell+1]}
\label{eq:C6_proof_axis_Cbar}
\end{equation}
Taking the expectation with respect to the variational posterior of the
trajectory orientation state at the $\ell$-th iteration gives the
residual-statistics matrix in the principal-axis coordinate system
\begin{equation}
\widetilde{\bm C}_k^{(j,\mathcal P,\mathcal C),[\ell+1]}
=\mathbb E_{q_k^{\theta,(j,\mathcal P,\mathcal C),[\ell]}}
\!\left[
\mathcal R\left(-\theta^{n_{\omega_k}^{(j)}}\right)
\bar{\bm C}_k^{(j,\mathcal P,\mathcal C),[\ell+1]}
\mathcal R\left(\theta^{n_{\omega_k}^{(j)}}\right)
\right]
\label{eq:C6_proof_axis_Bbar}
\end{equation}

Then, substituting Eq.~\eqref{eq:C6_proof_axis_residual_moment} into the
term \(\mathsf{B}\) in Eq.~\eqref{eq:C6_proof_axis_dependent_terms} gives
\begin{equation}
\begin{aligned}
&-\frac{1}{2}
\sum_{m=1}^{M_k^{\mathcal P,\mathcal C}}
\mathbb E_{\substack{
		q_k^{\theta,(j,\mathcal P,\mathcal C),[\ell]}\\
		q_k^{\ddot l_2,(j,\mathcal P,\mathcal C),[\ell]}}}
\!\Bigg[
\operatorname{tr}\!\Bigg(
\bm C_k^{(j,\mathcal P,\mathcal C),m,[\ell+1]}
\left(
\mathtt s\mathcal{S}\left(\theta^{n_{\omega_k}^{(j)}},\ddot l_1^{n_{\omega_k}^{(j)}},
\ddot l_2^{n_{\omega_k}^{(j)}}\right)
\right)^{-1}
\Bigg)
\Bigg]\\
&\quad=
-\frac{1}{2}
\mathbb E_{\substack{
		q_k^{\theta,(j,\mathcal P,\mathcal C),[\ell]}\\
		q_k^{\ddot l_2,(j,\mathcal P,\mathcal C),[\ell]}}}
\!\Bigg[
\sum_{m=1}^{M_k^{\mathcal P,\mathcal C}}
\operatorname{tr}\!\Bigg(
\bm C_k^{(j,\mathcal P,\mathcal C),m,[\ell+1]}
\left(
\mathtt s\mathcal{S}\left(\theta^{n_{\omega_k}^{(j)}},\ddot l_1^{n_{\omega_k}^{(j)}},
\ddot l_2^{n_{\omega_k}^{(j)}}\right)
\right)^{-1}
\Bigg)
\Bigg]\\
&\quad=
-\frac{1}{2\mathtt s}
\mathbb E_{\substack{
		q_k^{\theta,(j,\mathcal P,\mathcal C),[\ell]}\\
		q_k^{\ddot l_2,(j,\mathcal P,\mathcal C),[\ell]}}}
\!\Bigg[
\sum_{m=1}^{M_k^{\mathcal P,\mathcal C}}
\operatorname{tr}\!\Bigg(
\bm C_k^{(j,\mathcal P,\mathcal C),m,[\ell+1]}
\mathcal R\left(\theta^{n_{\omega_k}^{(j)}}\right)
\operatorname{diag}\!\left(
(\ddot l_{1}^{n_{\omega_k}^{(j)}})^{-1},
(\ddot l_{2}^{n_{\omega_k}^{(j)}})^{-1}
\right)
\mathcal R\left(-\theta^{n_{\omega_k}^{(j)}}\right)
\Bigg)
\Bigg]\\
&\quad=
-\frac{1}{2\mathtt s}
\mathbb E_{\substack{
		q_k^{\theta,(j,\mathcal P,\mathcal C),[\ell]}\\
		q_k^{\ddot l_2,(j,\mathcal P,\mathcal C),[\ell]}}}
\!\Bigg[
\operatorname{tr}\!\Bigg(
\operatorname{diag}\!\left(
(\ddot l_{1}^{n_{\omega_k}^{(j)}})^{-1},
(\ddot l_{2}^{n_{\omega_k}^{(j)}})^{-1}
\right)
\sum_{m=1}^{M_k^{\mathcal P,\mathcal C}}
\mathcal R\left(-\theta^{n_{\omega_k}^{(j)}}\right)
\bm C_k^{(j,\mathcal P,\mathcal C),m,[\ell+1]}
\mathcal R\left(\theta^{n_{\omega_k}^{(j)}}\right)
\Bigg)
\Bigg]\\
&\quad=
-\mathbb E_{\substack{
		q_k^{\theta,(j,\mathcal P,\mathcal C),[\ell]}\\
		q_k^{\ddot l_2,(j,\mathcal P,\mathcal C),[\ell]}}}
\!\Bigg[
\operatorname{tr}\!\Bigg(
\operatorname{diag}\!\left(
(\ddot l_{1}^{n_{\omega_k}^{(j)}})^{-1},
(\ddot l_{2}^{n_{\omega_k}^{(j)}})^{-1}
\right)
\mathcal R\left(-\theta^{n_{\omega_k}^{(j)}}\right)
\bar{\bm C}_k^{(j,\mathcal P,\mathcal C),[\ell+1]}
\mathcal R\left(\theta^{n_{\omega_k}^{(j)}}\right)
\Bigg)
\Bigg]\\
&\quad=
-\mathbb E_{\substack{
		q_k^{\theta,(j,\mathcal P,\mathcal C),[\ell]}\\
		q_k^{\ddot l_2,(j,\mathcal P,\mathcal C),[\ell]}}}
\!\Bigg[
\frac{
	\left[
	\mathcal R\left(-\theta^{n_{\omega_k}^{(j)}}\right)
	\bar{\bm C}_k^{(j,\mathcal P,\mathcal C),[\ell+1]}
	\mathcal R\left(\theta^{n_{\omega_k}^{(j)}}\right)
	\right]_{11}
}{
\ddot l_{1}^{n_{\omega_k}^{(j)}}
}
+
\frac{
	\left[
	\mathcal R\left(-\theta^{n_{\omega_k}^{(j)}}\right)
	\bar{\bm C}_k^{(j,\mathcal P,\mathcal C),[\ell+1]}
	\mathcal R\left(\theta^{n_{\omega_k}^{(j)}}\right)
	\right]_{22}
}{
\ddot l_{2}^{n_{\omega_k}^{(j)}}
}
\Bigg]\\
&\quad=
-\mathbb E_{q_k^{\ddot l_2,(j,\mathcal P,\mathcal C),[\ell]}}
\!\Bigg[
\frac{
	[\widetilde{\bm C}_k^{(j,\mathcal P,\mathcal C),[\ell+1]}]_{11}
}{
\ddot l_{1}^{n_{\omega_k}^{(j)}}
}
+
\frac{
	[\widetilde{\bm C}_k^{(j,\mathcal P,\mathcal C),[\ell+1]}]_{22}
}{
\ddot l_{2}^{n_{\omega_k}^{(j)}}
}
\Bigg]\\
&\quad=
-\frac{
	[\widetilde{\bm C}_k^{(j,\mathcal P,\mathcal C),[\ell+1]}]_{11}
}{
\ddot l_{1}^{n_{\omega_k}^{(j)}}
}
-\underbrace{
	[\widetilde{\bm C}_k^{(j,\mathcal P,\mathcal C),[\ell+1]}]_{22}
	\mathbb E_{q_k^{\ddot l_2,(j,\mathcal P,\mathcal C),[\ell]}}
	\!\left[
	(\ddot l_{2}^{n_{\omega_k}^{(j)}})^{-1}
	\right]
}_{\substack{
	\text{Independent of the current update variable}\
	\ddot l_{1}^{n_{\omega_k}^{(j)}}}}
\end{aligned}
\label{eq:C6_proof_axis_quadratic_contribution}
\end{equation}

The second term in the last line of
Eq.~\eqref{eq:C6_proof_axis_quadratic_contribution} is independent of
the current update variable $\ddot l_{1}^{n_{\omega_k}^{(j)}}$. It is therefore absorbed
into $\mathtt c_{l_1}$ in the CAVI kernel. The remaining terms in
Eq.~\eqref{eq:C6_proof_axis_quadratic_contribution} constitute the
$-\mathsf B$ contribution in
Eq.~\eqref{eq:C6_proof_axis_dependent_terms}. Using this contribution in
Eq.~\eqref{eq:C6_proof_axis_dependent_terms} gives
\begin{equation}
\begin{aligned}
&\ln q_k^{\ddot l_1,(j,\mathcal P,\mathcal C),[\ell+1]}
\!\left(
\ddot l_1^{1:n_{\pi_k}^{(j)}}
\right)\\
&\quad=
\underbrace{
	\mathbb E_{\substack{
			q_k^{\bm r,(j,\mathcal P,\mathcal C),[\ell+1]}
			q_k^{\theta,(j,\mathcal P,\mathcal C),[\ell]}
			q_k^{\ddot l_2,(j,\mathcal P,\mathcal C),[\ell]}\\
			{}\times\prod_{m=1}^{M_k^{\mathcal P,\mathcal C}}
			q_k^{\bm y^m,(j,\mathcal P,\mathcal C),[\ell]}}}
	\!\left[
	\ln p_k\!\left(
	\begin{gathered}
	\bm r^{1:n_{\omega_k}^{(j)}},
	\bm\theta^{1:n_{\omega_k}^{(j)}},
	\ddot l_1^{1:n_{\omega_k}^{(j)}},
	\ddot l_2^{1:n_{\omega_k}^{(j)}},\\
	\mathcal Y_k^{(\mathcal P,\mathcal C)},
	\mathcal Z_k^{(\mathcal P,\mathcal C)}
	\end{gathered}
	\right)
	-\sum_{m=1}^{M_k^{\mathcal P,\mathcal C}}
	\ln p\!\left(
	\bm y_k^{\mathcal P,\mathcal C,m}\mid
	\bm r^{1:n_{\omega_k}^{(j)}},\theta^{n_{\omega_k}^{(j)}},
	\ddot l_{1}^{n_{\omega_k}^{(j)}},
	\ddot l_{2}^{n_{\omega_k}^{(j)}}
	\right)
	\right]}_{\mathsf{C}}\\
&\qquad
-\frac{M_k^{\mathcal P,\mathcal C}}{2}
\ln\ddot l_{1}^{n_{\omega_k}^{(j)}}
-\frac{
	[\widetilde{\bm C}_k^{(j,\mathcal P,\mathcal C),[\ell+1]}]_{11}
}{
\ddot l_{1}^{n_{\omega_k}^{(j)}}
}
+\mathtt c_{l_1}
\end{aligned}
\label{eq:C6_proof_axis_terminal_likelihood}
\end{equation}

The term $\mathsf C$ in
Eq.~\eqref{eq:C6_proof_axis_terminal_likelihood} contains the structured
Beta--Bartlett prior in Eq.~\eqref{eq:C6_SSAL_Markov_joint}.
Combining this prior with the terminal measurement contribution gives
the following sequence-level variational factor:
\begin{equation}
\begin{aligned}
&q_k^{\ddot l_1,(j,\mathcal P,\mathcal C),[\ell+1]}
\!\left(\ddot l_1^{1:n_{\pi_k}^{(j)}}\right)\propto
p\!\left(
\ddot l_1^{1:n_{\omega_k}^{(j)}}
\right)
\left(
\ddot l_1^{n_{\omega_k}^{(j)}}
\right)^{-M_k^{\mathcal P,\mathcal C}/2}
\exp\!\left(
-\frac{
[\widetilde{\bm C}_k^{(j,\mathcal P,\mathcal C),[\ell+1]}]_{11}
}{
\ddot l_1^{n_{\omega_k}^{(j)}}
}
\right)
\end{aligned}
\label{eq:C6_proof_axis_structured_factor}
\end{equation}

The predicted terminal marginal of this Markov joint is the
inverse-Gamma density in Eq.~\eqref{eq:C6_proof_axis_terminal_prior}.
Combining that marginal with the terminal contribution in
Eq.~\eqref{eq:C6_proof_axis_terminal_likelihood} gives
\begin{equation}
\begin{aligned}
&\ln q_k^{\ddot l_1,(j,\mathcal P,\mathcal C),[\ell+1]}
\!\left(
\ddot l_1^{n_{\pi_k}^{(j)}}
\right)=-\left(
a_{1,\omega_k}^{(j),n_{\omega_k}^{(j)}}
+\frac{M_k^{\mathcal P,\mathcal C}}{2}+1
\right)
\ln\ddot l_1^{n_{\pi_k}^{(j)}}-
\frac{
b_{1,\omega_k}^{(j),n_{\omega_k}^{(j)}}
+[\widetilde{\bm C}_k^{(j,\mathcal P,\mathcal C),[\ell+1]}]_{11}
}{
\ddot l_1^{n_{\pi_k}^{(j)}}
}
+\mathtt c_{l_1}
\end{aligned}
\label{eq:C6_proof_axis_11}
\end{equation}

The logarithmic kernel of the terminal posterior inverse-Gamma density is
\begin{equation}
\begin{aligned}
\ln q_k^{\ddot l_1,(j,\mathcal P,\mathcal C),[\ell+1]}
\!\left(\ddot l_1^{n_{\pi_k}^{(j)}}\right)
&=\ln\mathcal{IG}\!\left(
\ddot l_1^{n_{\pi_k}^{(j)}};
 a_{1,\pi_k}^{(j,\mathcal P,\mathcal C),n_{\pi_k}^{(j)}},
 b_{1,\pi_k}^{(j,\mathcal P,\mathcal C),n_{\pi_k}^{(j)},[\ell+1]}
\right)\\
&\quad=-\left(
 a_{1,\pi_k}^{(j,\mathcal P,\mathcal C),n_{\pi_k}^{(j)}}+1
\right)\ln\ddot l_1^{n_{\pi_k}^{(j)}}
-\frac{
 b_{1,\pi_k}^{(j,\mathcal P,\mathcal C),n_{\pi_k}^{(j)},[\ell+1]}
}{
\ddot l_1^{n_{\pi_k}^{(j)}}
}
+\mathtt c_{l_1}
\end{aligned}
\label{eq:C6_proof_axis_terminal_posterior_IG_kernel}
\end{equation}

Comparing Eqs.~\eqref{eq:C6_proof_axis_11} and
\eqref{eq:C6_proof_axis_terminal_posterior_IG_kernel} gives the terminal
trajectory shape and scale parameters.

\begin{equation}
a_{1,\pi_k}^{(j,\mathcal P,\mathcal C),n_{\pi_k}^{(j)}}
=a_{1,\omega_k}^{(j),n_{\omega_k}^{(j)}}
+\frac{M_k^{\mathcal P,\mathcal C}}{2}
\label{eq:C6_proof_axis_terminal_shape}
\end{equation}

The corresponding scale parameter at the $(\ell+1)$-th iteration is
\begin{equation}
b_{1,\pi_k}^{(j,\mathcal P,\mathcal C),n_{\pi_k}^{(j)},[\ell+1]}
=b_{1,\omega_k}^{(j),n_{\omega_k}^{(j)}}
+[\widetilde{\bm C}_k^{(j,\mathcal P,\mathcal C),[\ell+1]}]_{11}
\label{eq:C6_proof_axis_terminal_scale}
\end{equation}

Therefore, the reciprocal terminal scale parameter satisfies
\begin{equation}
\breve b_{1,\pi_k}^{(j,\mathcal P,\mathcal C),n_{\pi_k}^{(j)},[\ell+1]}
=\left[
b_{1,\omega_k}^{(j),n_{\omega_k}^{(j)}}
+[\widetilde{\bm C}_k^{(j,\mathcal P,\mathcal C),[\ell+1]}]_{11}
\right]^{-1}
\label{eq:C6_proof_axis_terminal_invscale}
\end{equation}

The terminal update provides the last filtering marginal of the
Beta--Bartlett chain. For the historical trajectory times, the recovery
operator~$\bm U_{n_{\pi_k}^{(j)}}(\mathtt g)$ defined below recovers the
forward parameters from the stored smoothed marginals; the backward
recursion in \cite[Eqs.~(4.9) and~(4.10)]{Kartal-2022-VS-ETT} then
propagates the terminal information, after which the resulting
inverse-Gamma marginals are projected onto the DGDIG form in
Eq.~\eqref{eq:C6_single_component}. Applying this recursion to the first
trajectory SSAL state gives, for
$h=1,\ldots,n_{\pi_k}^{(j)}-1$,
\begin{equation}
a_{1,\pi_k}^{(j,\mathcal P,\mathcal C),h}
=(1-\mathtt g)
\left[
\bm U_{n_{\pi_k}^{(j)}}(\mathtt g)
\bm A_{1,\omega_k}^{(j)}
\right]_h
+\mathtt g
a_{1,\pi_k}^{(j,\mathcal P,\mathcal C),h+1}
\label{eq:C6_proof_axis_backward_shape}
\end{equation}
The corresponding reciprocal scale parameter satisfies
\begin{equation}
\begin{aligned}
\breve b_{1,\pi_k}^{(j,\mathcal P,\mathcal C),h,[\ell+1]}
=(1-\mathtt g)
\left[
\bm U_{n_{\pi_k}^{(j)}}(\mathtt g)
\breve{\bm B}_{1,\omega_k}^{(j)}
\right]_h+
\mathtt g
\breve b_{1,\pi_k}^{(j,\mathcal P,\mathcal C),h+1,[\ell+1]}
\end{aligned}
\label{eq:C6_proof_axis_backward_invscale}
\end{equation}
\noindent where $[\cdot]_h$ denotes the $h$-th element of a vector.
According to the backward recursion in
\cite[Eq.~(4.10)]{Kartal-2022-VS-ETT}, for a historical parameter vector
of length $n_{\pi_k}^{(j)}-1$, the recovery matrix of the
corresponding forward parameters is
$\bm W_{n_{\pi_k}^{(j)}-1}^{-1}(\mathtt g)
\bm T_{n_{\pi_k}^{(j)}-1}(\mathtt g)$. Since the current predicted
parameter at the terminal trajectory time does not require recovery, the
recovery matrix of the complete trajectory parameters is
\begin{equation}
\bm U_{n_{\pi_k}^{(j)}}(\mathtt g)
=\begin{cases}
\bm I_1, & n_{\pi_k}^{(j)}=1\\
\operatorname{blkdiag}\!\left(
\bm W_{n_{\pi_k}^{(j)}-1}^{-1}(\mathtt g)
\bm T_{n_{\pi_k}^{(j)}-1}(\mathtt g),1
\right), & n_{\pi_k}^{(j)}\geq2
\end{cases}
\label{eq:C6_proof_axis_recovery_operator}
\end{equation}

According to the definitions of $\bm J_n$ and $\bm T_n(\mathtt g)$,
when $\bm T_n(\mathtt g)$ acts on an arbitrary vector, each of its first
$n-1$ elements equals the current element minus $\mathtt g$ times the
subsequent element, whereas the terminal element remains unchanged.
Stacking Eq.~\eqref{eq:C6_proof_axis_backward_shape} and
Eq.~\eqref{eq:C6_proof_axis_terminal_shape} gives
\begin{equation}
\begin{aligned}
&\bm T_{n_{\pi_k}^{(j)}}(\mathtt g)
\bm A_{1,\pi_k}^{(j,\mathcal P,\mathcal C)}=\bm W_{n_{\pi_k}^{(j)}}(\mathtt g)
\bm U_{n_{\pi_k}^{(j)}}(\mathtt g)
\bm A_{1,\omega_k}^{(j)}
+\frac{M_k^{\mathcal P,\mathcal C}}{2}
\bm e_{n_{\pi_k}^{(j)}}
\end{aligned}
\label{eq:C6_proof_axis_shape_matrix}
\end{equation}

Since $\bm T_n(\mathtt g)$ is a unit upper-triangular matrix, it is
invertible. Left-multiplying both sides of
Eq.~\eqref{eq:C6_proof_axis_shape_matrix} by
$\bm T_{n_{\pi_k}^{(j)}}^{-1}(\mathtt g)$ gives the shape-parameter
vector of the first posterior trajectory SSAL state
\begin{equation}
\begin{aligned}
&\bm A_{1,\pi_k}^{(j,\mathcal P,\mathcal C)}
=\bm T_{n_{\pi_k}^{(j)}}^{-1}(\mathtt g)
\left[
\bm W_{n_{\pi_k}^{(j)}}(\mathtt g)
\bm U_{n_{\pi_k}^{(j)}}(\mathtt g)
\bm A_{1,\omega_k}^{(j)}
+\frac{M_k^{\mathcal P,\mathcal C}}{2}
\bm e_{n_{\pi_k}^{(j)}}
\right]
\end{aligned}
\label{eq:C6_proof_axis_shape_result}
\end{equation}

Then, stacking Eq.~\eqref{eq:C6_proof_axis_backward_invscale} and
Eq.~\eqref{eq:C6_proof_axis_terminal_invscale} gives
\begin{equation}
\begin{aligned}
\bm T_{n_{\pi_k}^{(j)}}(\mathtt g)
\breve{\bm B}_{1,\pi_k}^{(j,\mathcal P,\mathcal C),[\ell+1]}=(1-\mathtt g)
\bm\varPi_{n_{\pi_k}^{(j)}}
\bm U_{n_{\pi_k}^{(j)}}(\mathtt g)
\breve{\bm B}_{1,\omega_k}^{(j)}
+\left[
b_{1,\omega_k}^{(j),n_{\omega_k}^{(j)}}
+[\widetilde{\bm C}_k^{(j,\mathcal P,\mathcal C),[\ell+1]}]_{11}
\right]^{-1}
\bm e_{n_{\pi_k}^{(j)}}
\end{aligned}
\label{eq:C6_proof_axis_invscale_matrix}
\end{equation}

Left-multiplying both sides of
Eq.~\eqref{eq:C6_proof_axis_invscale_matrix} by
$\bm T_{n_{\pi_k}^{(j)}}^{-1}(\mathtt g)$ gives the reciprocal
scale-parameter vector of the first posterior trajectory SSAL state at
the $(\ell+1)$-th iteration
\begin{equation}
\begin{aligned}
&\breve{\bm B}_{1,\pi_k}^{(j,\mathcal P,\mathcal C),[\ell+1]}
=\bm T_{n_{\pi_k}^{(j)}}^{-1}(\mathtt g)
\left\{
(1-\mathtt g)
\bm\varPi_{n_{\pi_k}^{(j)}}
\bm U_{n_{\pi_k}^{(j)}}(\mathtt g)
\breve{\bm B}_{1,\omega_k}^{(j)}
+\left[
b_{1,\omega_k}^{(j),n_{\omega_k}^{(j)}}
+[\widetilde{\bm C}_k^{(j,\mathcal P,\mathcal C),[\ell+1]}]_{11}
\right]^{-1}
\bm e_{n_{\pi_k}^{(j)}}
\right\}
\end{aligned}
\label{eq:C6_proof_axis_invscale_result}
\end{equation}

For the second posterior trajectory SSAL state, replacing the
principal-axis index $1$ in the above derivation with $2$ and replacing
the principal-axis residual statistic
$[\widetilde{\bm C}_k^{(j,\mathcal P,\mathcal C),[\ell+1]}]_{11}$ with
$[\widetilde{\bm C}_k^{(j,\mathcal P,\mathcal C),[\ell+1]}]_{22}$ yield
the shape-parameter vector and reciprocal scale-parameter vector updates
of the same form. Therefore, Eqs.~\eqref{eq:C6_proof_axis_shape_result}
and~\eqref{eq:C6_proof_axis_invscale_result}, after replacing the
principal-axis index, hold for both $u\in\{1,2\}$.

Since $\mathtt g\in(0,1)$ and the inverse-Gamma parameters obtained from
the terminal trajectory update are positive, the backward recursion in
\cite[Eq.~(4.10)]{Kartal-2022-VS-ETT} preserves the inverse-Gamma form of
the smoothed marginal at each historical time. Taking the elementwise
reciprocal of
$\breve{\bm B}_{u,\pi_k}^{(j,\mathcal P,\mathcal C),[\ell+1]}$ gives the
scale-parameter vector
$\bm B_{u,\pi_k}^{(j,\mathcal P,\mathcal C),[\ell+1]}$.

Therefore, Eqs.~\eqref{eq:C6_proof_axis_shape_result} and
\eqref{eq:C6_proof_axis_invscale_result} prove
Proposition~\ref{prop:C6_axis_update}.

\newpage
\section{Proof of Proposition~\ref{prop:C6_theta_update}}
\label{app:C6_theta_update_proof}

According to the structured variational factorization in
Eq.~\eqref{eq:C6_traj_vb_factor} and the CAVI coordinate-update rule, the
variational factor of the posterior trajectory orientation state sequence
satisfies
\begingroup

\begin{equation}
\begin{aligned}
\ln q_k^{\theta,(j,\mathcal P,\mathcal C),[\ell+1]}
\!\left(\bm\theta^{1:n_{\pi_k}^{(j)}}\right)
&=
\mathbb E_{\substack{
q_k^{\bm r,(j,\mathcal P,\mathcal C),[\ell+1]}
\prod_{u=1}^{2}q_k^{\ddot l_u,(j,\mathcal P,\mathcal C),[\ell+1]}\\
{}\times\prod_{m=1}^{M_k^{\mathcal P,\mathcal C}}
q_k^{\bm y^m,(j,\mathcal P,\mathcal C),[\ell]}}}
\!\Bigg[
\ln p_k\!\left(
\begin{gathered}
\bm r^{1:n_{\omega_k}^{(j)}},
\bm\theta^{1:n_{\omega_k}^{(j)}},
\ddot l_1^{1:n_{\omega_k}^{(j)}},
\ddot l_2^{1:n_{\omega_k}^{(j)}},\\
\mathcal Y_k^{(\mathcal P,\mathcal C)},
\mathcal Z_k^{(\mathcal P,\mathcal C)}
\end{gathered}
\right)
\Bigg]
+\mathtt c_{\theta}
\end{aligned}
\label{eq:C6_proof_theta_CAVI}
\end{equation}
\endgroup

\begingroup

\noindent where the expectation is taken with respect to the variational
posteriors of the trajectory kinematic state sequence and the two
trajectory SSAL state sequences at the $(\ell+1)$-th iteration and all
measurement-source variables at the $\ell$-th iteration, excluding the
variational factor of the trajectory orientation state sequence being
updated; $\mathtt c_{\theta}$ denotes a constant term independent of
$\bm\theta^{1:n_{\pi_k}^{(j)}}$.
\endgroup

The trajectory orientation state density of the $j$-th predicted
trajectory component is
\begin{equation}
p_{\omega_k}^{\theta,(j)}
\!\left(\bm\theta^{1:n_{\omega_k}^{(j)}}\right)
=\mathcal N\!\left(
\bm\theta^{1:n_{\omega_k}^{(j)}};
\hat{\bm\varTheta}_{\omega_k}^{(j)},
\bm\varXi_{\omega_k}^{\bm\varTheta,(j)}
\right)
\label{eq:C6_proof_theta_pred_density}
\end{equation}

According to the variational posteriors of the SSAL states at the
$(\ell+1)$-th iteration obtained in
Proposition~\ref{prop:C6_axis_update}, the scale-weighted reciprocal
expectations of the two terminal trajectory SSAL states form the
following diagonal precision matrix:
\begin{equation}
\bm\varPhi_{\pi_k}^{(j,\mathcal P,\mathcal C),[\ell+1]}
=\operatorname{diag}\!\left(
\iota_{1,\pi_k}^{(j,\mathcal P,\mathcal C),n_{\omega_k}^{(j)},[\ell+1]},
\iota_{2,\pi_k}^{(j,\mathcal P,\mathcal C),n_{\omega_k}^{(j)},[\ell+1]}
\right)
\label{eq:C6_proof_theta_axis_precision}
\end{equation}

Combining the posterior trajectory kinematic state sequence at the
$(\ell+1)$-th iteration obtained in
Proposition~\ref{prop:C6_r_update} with the variational posteriors of the
measurement-source variables at the $\ell$-th iteration, the second-order
residual moment corresponding to the $m$-th measurement-source variable
in the measurement cell is given by Eq.~\eqref{eq:C6_axis_Cm_update}.
\begingroup

Using the measurement-source conditional density in
Eq.~\eqref{eq:C6_proof_r_source_density2},
$\det(\mathcal R(\theta))=1$,
$\bm a^{\mathrm T}\bm B\bm a
=\operatorname{tr}(\bm B\bm a\bm a^{\mathrm T})$,
Eqs.~\eqref{eq:C6_proof_theta_axis_precision} and
\eqref{eq:C6_axis_Cm_update}, and the cyclic property of the trace,
Eq.~\eqref{eq:C6_proof_theta_CAVI} can be expanded through the following
chain of equalities:
\begin{equation}
\begin{aligned}
&\ln q_k^{\theta,(j,\mathcal P,\mathcal C),[\ell+1]}
\!\left(\bm\theta^{1:n_{\pi_k}^{(j)}}\right)\\
&\quad=\ln p_{\omega_k}^{\theta,(j)}
\!\left(\bm\theta^{1:n_{\omega_k}^{(j)}}\right)
+\sum_{m=1}^{M_k^{\mathcal P,\mathcal C}}
\mathbb E_{\substack{
q_k^{\bm r,(j,\mathcal P,\mathcal C),[\ell+1]}
\prod_{u=1}^{2}q_k^{\ddot l_u,(j,\mathcal P,\mathcal C),[\ell+1]}\\
{}\times q_k^{\bm y^m,(j,\mathcal P,\mathcal C),[\ell]}}}
\!\Bigg[
\ln p\!\left(
\bm y_k^{\mathcal P,\mathcal C,m}\mid
\bm r^{1:n_{\omega_k}^{(j)}},
\theta^{n_{\omega_k}^{(j)}},
\ddot l_1^{n_{\omega_k}^{(j)}},
\ddot l_2^{n_{\omega_k}^{(j)}}
\right)
\Bigg]
+\mathtt c_{\theta}\\[1mm]
&\quad=\ln p_{\omega_k}^{\theta,(j)}
\!\left(\bm\theta^{1:n_{\omega_k}^{(j)}}\right)
-\frac{1}{2}
\sum_{m=1}^{M_k^{\mathcal P,\mathcal C}}
\mathbb E_{\substack{
q_k^{\bm r,(j,\mathcal P,\mathcal C),[\ell+1]}
\prod_{u=1}^{2}q_k^{\ddot l_u,(j,\mathcal P,\mathcal C),[\ell+1]}\\
{}\times q_k^{\bm y^m,(j,\mathcal P,\mathcal C),[\ell]}}}
\!\Bigg[
\ln\det\!\left(
\mathtt s\mathcal S\!\left(
\theta^{n_{\omega_k}^{(j)}},
\ddot l_1^{n_{\omega_k}^{(j)}},
\ddot l_2^{n_{\omega_k}^{(j)}}
\right)
\right)\\
&\qquad+
\left(
\bm y_k^{\mathcal P,\mathcal C,m}
-\dot{\bm H}_k^{(j)}\bm r^{1:n_{\omega_k}^{(j)}}
\right)^{\mathrm T}
\left[
\mathtt s\mathcal S\!\left(
\theta^{n_{\omega_k}^{(j)}},
\ddot l_1^{n_{\omega_k}^{(j)}},
\ddot l_2^{n_{\omega_k}^{(j)}}
\right)
\right]^{-1}
\left(
\bm y_k^{\mathcal P,\mathcal C,m}
-\dot{\bm H}_k^{(j)}\bm r^{1:n_{\omega_k}^{(j)}}
\right)
\Bigg]
+\mathtt c_{\theta}\\[1mm]
&\quad=\ln p_{\omega_k}^{\theta,(j)}
\!\left(\bm\theta^{1:n_{\omega_k}^{(j)}}\right)
-\frac{1}{2}
\sum_{m=1}^{M_k^{\mathcal P,\mathcal C}}
\operatorname{tr}\!\Bigg[
\mathcal R\!\left(\theta^{n_{\omega_k}^{(j)}}\right)
\bm\varPhi_{\pi_k}^{(j,\mathcal P,\mathcal C),[\ell+1]}
\mathcal R\!\left(-\theta^{n_{\omega_k}^{(j)}}\right)
\bm C_k^{(j,\mathcal P,\mathcal C),m,[\ell+1]}
\Bigg]
+\mathtt c_{\theta}\\[1mm]
&\quad=\ln p_{\omega_k}^{\theta,(j)}
\!\left(\bm\theta^{1:n_{\omega_k}^{(j)}}\right)
\underbrace{
-\frac{1}{2}
\sum_{m=1}^{M_k^{\mathcal P,\mathcal C}}
\operatorname{tr}\!\left[
\bm\varPhi_{\pi_k}^{(j,\mathcal P,\mathcal C),[\ell+1]}
\mathcal R\!\left(-\theta^{n_{\omega_k}^{(j)}}\right)
\bm C_k^{(j,\mathcal P,\mathcal C),m,[\ell+1]}
\mathcal R\!\left(\theta^{n_{\omega_k}^{(j)}}\right)
\right]
}_{\mathsf D}
+\mathtt c_{\theta}
\end{aligned}
\label{eq:C6_proof_theta_likelihood}
\end{equation}
\endgroup

The first derivative of the rotation function with respect to the
orientation state is
\begin{equation}
\dot{\mathcal R}(\theta)
=\frac{\partial\mathcal R(\theta)}{\partial\theta}
=\begin{bmatrix}
-\sin\theta&-\cos\theta\\
\cos\theta&-\sin\theta
\end{bmatrix}
\label{eq:C6_proof_rotation_derivative}
\end{equation}

Expanding the rotation function and its inverse to first order at the
terminal trajectory orientation mean
$\hat\theta_{\pi_k}^{(j,\mathcal P,\mathcal C),n_{\omega_k}^{(j)},[\ell]}$
at the $\ell$-th iteration gives
\begin{equation}
\mathcal R\!\left(\theta^{n_{\omega_k}^{(j)}}\right)
\approx\mathcal R\!\left(
\hat\theta_{\pi_k}^{(j,\mathcal P,\mathcal C),n_{\omega_k}^{(j)},[\ell]}
\right)+\dot{\mathcal R}\!\left(
\hat\theta_{\pi_k}^{(j,\mathcal P,\mathcal C),n_{\omega_k}^{(j)},[\ell]}
\right)
\left(
\theta^{n_{\omega_k}^{(j)}}-
\hat\theta_{\pi_k}^{(j,\mathcal P,\mathcal C),n_{\omega_k}^{(j)},[\ell]}
\right)
\label{eq:C6_proof_theta_linearization}
\end{equation}

\begin{equation}
\mathcal R\!\left(-\theta^{n_{\omega_k}^{(j)}}\right)
\approx\mathcal R\!\left(
-\hat\theta_{\pi_k}^{(j,\mathcal P,\mathcal C),n_{\omega_k}^{(j)},[\ell]}
\right)-\dot{\mathcal R}\!\left(
-\hat\theta_{\pi_k}^{(j,\mathcal P,\mathcal C),n_{\omega_k}^{(j)},[\ell]}
\right)
\left(
\theta^{n_{\omega_k}^{(j)}}-
\hat\theta_{\pi_k}^{(j,\mathcal P,\mathcal C),n_{\omega_k}^{(j)},[\ell]}
\right)
\label{eq:C6_proof_theta_inverse_linearization}
\end{equation}

\begingroup

Because
$\bm\varPhi_{\pi_k}^{(j,\mathcal P,\mathcal C),[\ell+1]}$ and
$\bm C_k^{(j,\mathcal P,\mathcal C),m,[\ell+1]}$ are symmetric, while
$\mathcal R(-\theta)=\mathcal R(\theta)^{\mathrm T}$ and
$-\dot{\mathcal R}(-\theta)=\dot{\mathcal R}(\theta)^{\mathrm T}$, the two
linear cross terms are identical. Substituting
Eqs.~\eqref{eq:C6_proof_theta_linearization} and
\eqref{eq:C6_proof_theta_inverse_linearization} into the
orientation-dependent measurement term in
Eq.~\eqref{eq:C6_proof_theta_likelihood}, expanding the matrix product,
and collecting the terms in $\theta^{n_{\omega_k}^{(j)}}$ gives the
following chain of equalities:
\allowdisplaybreaks[3]
\begin{align}
\mathsf{D}&=-\frac{1}{2}
\sum_{m=1}^{M_k^{\mathcal P,\mathcal C}}
\operatorname{tr}\!\left[
\bm\varPhi_{\pi_k}^{(j,\mathcal P,\mathcal C),[\ell+1]}
\mathcal R\!\left(-\theta^{n_{\omega_k}^{(j)}}\right)
\bm C_k^{(j,\mathcal P,\mathcal C),m,[\ell+1]}
\mathcal R\!\left(\theta^{n_{\omega_k}^{(j)}}\right)
\right]
\notag\\
&\quad\approx
-\frac{1}{2}
\sum_{m=1}^{M_k^{\mathcal P,\mathcal C}}
\operatorname{tr}\!\Bigg[
\bm\varPhi_{\pi_k}^{(j,\mathcal P,\mathcal C),[\ell+1]}
\left[
\mathcal R\!\left(
-\hat\theta_{\pi_k}^{(j,\mathcal P,\mathcal C),
n_{\omega_k}^{(j)},[\ell]}
\right)
-\dot{\mathcal R}\!\left(
-\hat\theta_{\pi_k}^{(j,\mathcal P,\mathcal C),
n_{\omega_k}^{(j)},[\ell]}
\right)
\underbrace{
\left(
\theta^{n_{\omega_k}^{(j)}}-
\hat\theta_{\pi_k}^{(j,\mathcal P,\mathcal C),
n_{\omega_k}^{(j)},[\ell]}
\right)
}_{\Delta\theta}
\right]
\notag\\
&\hspace{28mm}\times
\bm C_k^{(j,\mathcal P,\mathcal C),m,[\ell+1]}
\left[
\mathcal R\!\left(
\hat\theta_{\pi_k}^{(j,\mathcal P,\mathcal C),
n_{\omega_k}^{(j)},[\ell]}
\right)
+\dot{\mathcal R}\!\left(
\hat\theta_{\pi_k}^{(j,\mathcal P,\mathcal C),
n_{\omega_k}^{(j)},[\ell]}
\right)\Delta\theta
\right]
\Bigg]
\notag\\
&\quad=
-\frac{1}{2}
\sum_{m=1}^{M_k^{\mathcal P,\mathcal C}}
\Bigg\{
\operatorname{tr}\!\Bigg[
\bm\varPhi_{\pi_k}^{(j,\mathcal P,\mathcal C),[\ell+1]}
\mathcal R\!\left(
-\hat\theta_{\pi_k}^{(j,\mathcal P,\mathcal C),
n_{\omega_k}^{(j)},[\ell]}
\right)
\bm C_k^{(j,\mathcal P,\mathcal C),m,[\ell+1]}
\mathcal R\!\left(
\hat\theta_{\pi_k}^{(j,\mathcal P,\mathcal C),
n_{\omega_k}^{(j)},[\ell]}
\right)
\Bigg]
\notag\\
&\qquad+
\Delta\theta\,
\operatorname{tr}\!\Bigg[
\bm\varPhi_{\pi_k}^{(j,\mathcal P,\mathcal C),[\ell+1]}
\left[
-\dot{\mathcal R}\!\left(
-\hat\theta_{\pi_k}^{(j,\mathcal P,\mathcal C),
n_{\omega_k}^{(j)},[\ell]}
\right)
\right]
\bm C_k^{(j,\mathcal P,\mathcal C),m,[\ell+1]}
\mathcal R\!\left(
\hat\theta_{\pi_k}^{(j,\mathcal P,\mathcal C),
n_{\omega_k}^{(j)},[\ell]}
\right)
\Bigg]
\notag\\
&\qquad+
\Delta\theta\,
\operatorname{tr}\!\Bigg[
\bm\varPhi_{\pi_k}^{(j,\mathcal P,\mathcal C),[\ell+1]}
\mathcal R\!\left(
-\hat\theta_{\pi_k}^{(j,\mathcal P,\mathcal C),
n_{\omega_k}^{(j)},[\ell]}
\right)
\bm C_k^{(j,\mathcal P,\mathcal C),m,[\ell+1]}
\dot{\mathcal R}\!\left(
\hat\theta_{\pi_k}^{(j,\mathcal P,\mathcal C),
n_{\omega_k}^{(j)},[\ell]}
\right)
\Bigg]
\notag\\
&\qquad+
(\Delta\theta)^2
\operatorname{tr}\!\Bigg[
\bm\varPhi_{\pi_k}^{(j,\mathcal P,\mathcal C),[\ell+1]}
\left[
-\dot{\mathcal R}\!\left(
-\hat\theta_{\pi_k}^{(j,\mathcal P,\mathcal C),
n_{\omega_k}^{(j)},[\ell]}
\right)
\right]
\bm C_k^{(j,\mathcal P,\mathcal C),m,[\ell+1]}
\dot{\mathcal R}\!\left(
\hat\theta_{\pi_k}^{(j,\mathcal P,\mathcal C),
n_{\omega_k}^{(j)},[\ell]}
\right)
\Bigg]
\Bigg\}
\notag\\
&\quad=
-\frac{1}{2}
\sum_{m=1}^{M_k^{\mathcal P,\mathcal C}}
\operatorname{tr}\!\Bigg[
\bm\varPhi_{\pi_k}^{(j,\mathcal P,\mathcal C),[\ell+1]}
\left[
-\dot{\mathcal R}\!\left(
-\hat\theta_{\pi_k}^{(j,\mathcal P,\mathcal C),
n_{\omega_k}^{(j)},[\ell]}
\right)
\right]
\bm C_k^{(j,\mathcal P,\mathcal C),m,[\ell+1]}
\dot{\mathcal R}\!\left(
\hat\theta_{\pi_k}^{(j,\mathcal P,\mathcal C),
n_{\omega_k}^{(j)},[\ell]}
\right)
\Bigg](\Delta\theta)^2
\notag\\
&\qquad-
\sum_{m=1}^{M_k^{\mathcal P,\mathcal C}}
\operatorname{tr}\!\Bigg[
\bm\varPhi_{\pi_k}^{(j,\mathcal P,\mathcal C),[\ell+1]}
\mathcal R\!\left(
-\hat\theta_{\pi_k}^{(j,\mathcal P,\mathcal C),
n_{\omega_k}^{(j)},[\ell]}
\right)
\bm C_k^{(j,\mathcal P,\mathcal C),m,[\ell+1]}
\dot{\mathcal R}\!\left(
\hat\theta_{\pi_k}^{(j,\mathcal P,\mathcal C),
n_{\omega_k}^{(j)},[\ell]}
\right)
\Bigg]\Delta\theta
+\mathtt c_{\theta}
\notag\\
&\quad=
-\frac{1}{2}
\sum_{m=1}^{M_k^{\mathcal P,\mathcal C}}
\operatorname{tr}\!\Bigg[
\bm\varPhi_{\pi_k}^{(j,\mathcal P,\mathcal C),[\ell+1]}
\left[
-\dot{\mathcal R}\!\left(
-\hat\theta_{\pi_k}^{(j,\mathcal P,\mathcal C),
n_{\omega_k}^{(j)},[\ell]}
\right)
\right]
\bm C_k^{(j,\mathcal P,\mathcal C),m,[\ell+1]}
\dot{\mathcal R}\!\left(
\hat\theta_{\pi_k}^{(j,\mathcal P,\mathcal C),
n_{\omega_k}^{(j)},[\ell]}
\right)
\Bigg]
\left(
\theta^{n_{\omega_k}^{(j)}}-
\hat\theta_{\pi_k}^{(j,\mathcal P,\mathcal C),
n_{\omega_k}^{(j)},[\ell]}
\right)^2
\notag\\
&\qquad-
\sum_{m=1}^{M_k^{\mathcal P,\mathcal C}}
\operatorname{tr}\!\Bigg[
\bm\varPhi_{\pi_k}^{(j,\mathcal P,\mathcal C),[\ell+1]}
\mathcal R\!\left(
-\hat\theta_{\pi_k}^{(j,\mathcal P,\mathcal C),
n_{\omega_k}^{(j)},[\ell]}
\right)
\bm C_k^{(j,\mathcal P,\mathcal C),m,[\ell+1]}
\dot{\mathcal R}\!\left(
\hat\theta_{\pi_k}^{(j,\mathcal P,\mathcal C),
n_{\omega_k}^{(j)},[\ell]}
\right)
\Bigg]
\left(
\theta^{n_{\omega_k}^{(j)}}-
\hat\theta_{\pi_k}^{(j,\mathcal P,\mathcal C),
n_{\omega_k}^{(j)},[\ell]}
\right)
+\mathtt c_{\theta}
\notag\\
&\quad=
-\frac{1}{2}
\underbrace{
\sum_{m=1}^{M_k^{\mathcal P,\mathcal C}}
\operatorname{tr}\!\Bigg[
\bm\varPhi_{\pi_k}^{(j,\mathcal P,\mathcal C),[\ell+1]}
\left[
-\dot{\mathcal R}\!\left(
-\hat\theta_{\pi_k}^{(j,\mathcal P,\mathcal C),
n_{\omega_k}^{(j)},[\ell]}
\right)
\right]
\bm C_k^{(j,\mathcal P,\mathcal C),m,[\ell+1]}
\dot{\mathcal R}\!\left(
\hat\theta_{\pi_k}^{(j,\mathcal P,\mathcal C),
n_{\omega_k}^{(j)},[\ell]}
\right)
\Bigg]
}_{\zeta_{\pi_k}^{(j,\mathcal P,\mathcal C),[\ell+1]}}
\left(\theta^{n_{\omega_k}^{(j)}}\right)^2
\notag\\
&\qquad+
\underbrace{
\begin{aligned}
&\hat\theta_{\pi_k}^{(j,\mathcal P,\mathcal C),
n_{\omega_k}^{(j)},[\ell]}
\sum_{m=1}^{M_k^{\mathcal P,\mathcal C}}
\operatorname{tr}\!\Bigg[
\bm\varPhi_{\pi_k}^{(j,\mathcal P,\mathcal C),[\ell+1]}
\left[
-\dot{\mathcal R}\!\left(
-\hat\theta_{\pi_k}^{(j,\mathcal P,\mathcal C),
n_{\omega_k}^{(j)},[\ell]}
\right)
\right]
\bm C_k^{(j,\mathcal P,\mathcal C),m,[\ell+1]}
\dot{\mathcal R}\!\left(
\hat\theta_{\pi_k}^{(j,\mathcal P,\mathcal C),
n_{\omega_k}^{(j)},[\ell]}
\right)
\Bigg]\\
&\quad-
\sum_{m=1}^{M_k^{\mathcal P,\mathcal C}}
\operatorname{tr}\!\Bigg[
\bm\varPhi_{\pi_k}^{(j,\mathcal P,\mathcal C),[\ell+1]}
\mathcal R\!\left(
-\hat\theta_{\pi_k}^{(j,\mathcal P,\mathcal C),
n_{\omega_k}^{(j)},[\ell]}
\right)
\bm C_k^{(j,\mathcal P,\mathcal C),m,[\ell+1]}
\dot{\mathcal R}\!\left(
\hat\theta_{\pi_k}^{(j,\mathcal P,\mathcal C),
n_{\omega_k}^{(j)},[\ell]}
\right)
\Bigg]
\end{aligned}
}_{\eta_{\pi_k}^{(j,\mathcal P,\mathcal C),[\ell+1]}}
\theta^{n_{\omega_k}^{(j)}}
+\mathtt c_{\theta}
\label{eq:C6_proof_theta_quadratic_likelihood}
\end{align}
\endgroup

When
$\zeta_{\pi_k}^{(j,\mathcal P,\mathcal C),[\ell+1]}>0$, completing the
square for the quadratic and linear terms in
$\theta^{n_{\omega_k}^{(j)}}$ in
Eq.~\eqref{eq:C6_proof_theta_quadratic_likelihood} gives
\begingroup

\begin{equation}
\begin{aligned}
&-\frac{1}{2}
\zeta_{\pi_k}^{(j,\mathcal P,\mathcal C),[\ell+1]}
\left(\theta^{n_{\omega_k}^{(j)}}\right)^2
+\eta_{\pi_k}^{(j,\mathcal P,\mathcal C),[\ell+1]}
\theta^{n_{\omega_k}^{(j)}}\\
&\quad=
-\frac{1}{2}
\zeta_{\pi_k}^{(j,\mathcal P,\mathcal C),[\ell+1]}
\Bigg[
\left(\theta^{n_{\omega_k}^{(j)}}\right)^2
-2
\frac{
\eta_{\pi_k}^{(j,\mathcal P,\mathcal C),[\ell+1]}
}{
\zeta_{\pi_k}^{(j,\mathcal P,\mathcal C),[\ell+1]}
}
\theta^{n_{\omega_k}^{(j)}}
+\left(
\frac{
\eta_{\pi_k}^{(j,\mathcal P,\mathcal C),[\ell+1]}
}{
\zeta_{\pi_k}^{(j,\mathcal P,\mathcal C),[\ell+1]}
}
\right)^2
-\left(
\frac{
\eta_{\pi_k}^{(j,\mathcal P,\mathcal C),[\ell+1]}
}{
\zeta_{\pi_k}^{(j,\mathcal P,\mathcal C),[\ell+1]}
}
\right)^2
\Bigg]\\
&\quad=
-\frac{1}{2}
\zeta_{\pi_k}^{(j,\mathcal P,\mathcal C),[\ell+1]}
\left(
\theta^{n_{\omega_k}^{(j)}}-
\frac{
\eta_{\pi_k}^{(j,\mathcal P,\mathcal C),[\ell+1]}
}{
\zeta_{\pi_k}^{(j,\mathcal P,\mathcal C),[\ell+1]}
}
\right)^2
+\underbrace{
\frac{
(\eta_{\pi_k}^{(j,\mathcal P,\mathcal C),[\ell+1]})^2
}{
2\zeta_{\pi_k}^{(j,\mathcal P,\mathcal C),[\ell+1]}
}
}_{
\substack{
\text{Independent of }\theta^{n_{\omega_k}^{(j)}}, \
\text{absorbed into }\mathtt c_{\theta}
}
}
\end{aligned}
\label{eq:C6_proof_theta_square_completion}
\end{equation}
\endgroup

\begingroup

The quantity $\mathsf D$ underbraced in
Eq.~\eqref{eq:C6_proof_theta_likelihood} is the orientation-dependent
contribution of the measurement-source conditional densities to the log
variational posterior. Its exponential is therefore the corresponding
local orientation measurement factor. Using
Eqs.~\eqref{eq:C6_proof_theta_quadratic_likelihood} and
\eqref{eq:C6_proof_theta_square_completion} and absorbing all terms
independent of $\theta^{n_{\omega_k}^{(j)}}$ into the proportionality
constant gives
\begin{equation}
\begin{aligned}
\exp\!\left(\mathsf D\right)
&\approx
\exp\!\Bigg[
-\frac{1}{2}
\zeta_{\pi_k}^{(j,\mathcal P,\mathcal C),[\ell+1]}
\left(\theta^{n_{\omega_k}^{(j)}}\right)^2
+\eta_{\pi_k}^{(j,\mathcal P,\mathcal C),[\ell+1]}
\theta^{n_{\omega_k}^{(j)}}
+\mathtt c_{\theta}
\Bigg]\\
&\quad\propto
\exp\!\Bigg[
-\frac{1}{2}
\zeta_{\pi_k}^{(j,\mathcal P,\mathcal C),[\ell+1]}
\left(
\theta^{n_{\omega_k}^{(j)}}-
\frac{
\eta_{\pi_k}^{(j,\mathcal P,\mathcal C),[\ell+1]}
}{
\zeta_{\pi_k}^{(j,\mathcal P,\mathcal C),[\ell+1]}
}
\right)^2
\Bigg]\\
&\quad=
\sqrt{
\frac{2\pi}{
\zeta_{\pi_k}^{(j,\mathcal P,\mathcal C),[\ell+1]}
}}
\mathcal N\!\left(
\theta^{n_{\omega_k}^{(j)}};
\frac{
\eta_{\pi_k}^{(j,\mathcal P,\mathcal C),[\ell+1]}
}{
\zeta_{\pi_k}^{(j,\mathcal P,\mathcal C),[\ell+1]}
},
(\zeta_{\pi_k}^{(j,\mathcal P,\mathcal C),[\ell+1]})^{-1}
\right)\\
&\quad\propto
\mathcal N\!\left(
\theta^{n_{\omega_k}^{(j)}};
\underbrace{
\frac{
\eta_{\pi_k}^{(j,\mathcal P,\mathcal C),[\ell+1]}
}{
\zeta_{\pi_k}^{(j,\mathcal P,\mathcal C),[\ell+1]}
}
}_{\widetilde\theta_{\pi_k}^{(j,\mathcal P,\mathcal C),[\ell+1]}},
(\zeta_{\pi_k}^{(j,\mathcal P,\mathcal C),[\ell+1]})^{-1}
\right)
\end{aligned}
\label{eq:C6_proof_theta_local_gaussian}
\end{equation}
\endgroup

Since the local Gaussian measurement term acts only on the terminal
trajectory orientation state, the terminal trajectory orientation can be
expressed as
$\bm e_{n_{\omega_k}^{(j)}}^{\mathrm T}
\bm\theta^{1:n_{\omega_k}^{(j)}}$. Multiplying the local measurement
term in Eq.~\eqref{eq:C6_proof_theta_local_gaussian} by the predicted
trajectory orientation state density in
Eq.~\eqref{eq:C6_proof_theta_pred_density} gives
\begingroup

\begin{equation}
\begin{aligned}
&\ln q_k^{\theta,(j,\mathcal P,\mathcal C),[\ell+1]}
\!\left(\bm\theta^{1:n_{\omega_k}^{(j)}}\right)\\
&\quad=-\frac{1}{2}
\left(
\bm\theta^{1:n_{\omega_k}^{(j)}}-
\hat{\bm\varTheta}_{\omega_k}^{(j)}
\right)^{\mathrm T}
(\bm\varXi_{\omega_k}^{\bm\varTheta,(j)})^{-1}
\left(
\bm\theta^{1:n_{\omega_k}^{(j)}}-
\hat{\bm\varTheta}_{\omega_k}^{(j)}
\right)-
\frac{1}{2}
\zeta_{\pi_k}^{(j,\mathcal P,\mathcal C),[\ell+1]}
\left(
\widetilde\theta_{\pi_k}^{(j,\mathcal P,\mathcal C),[\ell+1]}-
\bm e_{n_{\omega_k}^{(j)}}^{\mathrm T}
\bm\theta^{1:n_{\omega_k}^{(j)}}
\right)^2
+\mathtt c_{\theta}\\
&\quad=-\frac{1}{2}
\Bigg[
\left(\bm\theta^{1:n_{\omega_k}^{(j)}}\right)^{\mathrm T}
(\bm\varXi_{\omega_k}^{\bm\varTheta,(j)})^{-1}
\bm\theta^{1:n_{\omega_k}^{(j)}}
-2
\left(\bm\theta^{1:n_{\omega_k}^{(j)}}\right)^{\mathrm T}
(\bm\varXi_{\omega_k}^{\bm\varTheta,(j)})^{-1}
\hat{\bm\varTheta}_{\omega_k}^{(j)}+
(\hat{\bm\varTheta}_{\omega_k}^{(j)})^{\mathrm T}
(\bm\varXi_{\omega_k}^{\bm\varTheta,(j)})^{-1}
\hat{\bm\varTheta}_{\omega_k}^{(j)}
\Bigg]\\
&\qquad-
\frac{1}{2}
\zeta_{\pi_k}^{(j,\mathcal P,\mathcal C),[\ell+1]}
\Bigg[
\left(
\widetilde\theta_{\pi_k}^{(j,\mathcal P,\mathcal C),[\ell+1]}
\right)^2
-2
\widetilde\theta_{\pi_k}^{(j,\mathcal P,\mathcal C),[\ell+1]}
\bm e_{n_{\omega_k}^{(j)}}^{\mathrm T}
\bm\theta^{1:n_{\omega_k}^{(j)}}+
\left(\bm\theta^{1:n_{\omega_k}^{(j)}}\right)^{\mathrm T}
\bm e_{n_{\omega_k}^{(j)}}
\bm e_{n_{\omega_k}^{(j)}}^{\mathrm T}
\bm\theta^{1:n_{\omega_k}^{(j)}}
\Bigg]
+\mathtt c_{\theta}\\
&\quad=-\frac{1}{2}
\left(\bm\theta^{1:n_{\omega_k}^{(j)}}\right)^{\mathrm T}
\Bigg[
(\bm\varXi_{\omega_k}^{\bm\varTheta,(j)})^{-1}
+\zeta_{\pi_k}^{(j,\mathcal P,\mathcal C),[\ell+1]}
\bm e_{n_{\omega_k}^{(j)}}
\bm e_{n_{\omega_k}^{(j)}}^{\mathrm T}
\Bigg]
\bm\theta^{1:n_{\omega_k}^{(j)}} \\
&\qquad+\left(\bm\theta^{1:n_{\omega_k}^{(j)}}\right)^{\mathrm T}
\Bigg[
(\bm\varXi_{\omega_k}^{\bm\varTheta,(j)})^{-1}
\hat{\bm\varTheta}_{\omega_k}^{(j)}
+\zeta_{\pi_k}^{(j,\mathcal P,\mathcal C),[\ell+1]}
\bm e_{n_{\omega_k}^{(j)}}
\widetilde\theta_{\pi_k}^{(j,\mathcal P,\mathcal C),[\ell+1]}
\Bigg]
+\mathtt c_{\theta}
\end{aligned}
\label{eq:C6_proof_theta_trajectory_quadratic}
\end{equation}

Since $(\bm\varXi_{\omega_k}^{\bm\varTheta,(j)})^{-1}$ is positive
definite and
$\zeta_{\pi_k}^{(j,\mathcal P,\mathcal C),[\ell+1]}
\bm e_{n_{\omega_k}^{(j)}}\bm e_{n_{\omega_k}^{(j)}}^{\mathrm T}$ is
positive semidefinite, the quadratic coefficient in
Eq.~\eqref{eq:C6_proof_theta_trajectory_quadratic} is positive definite.
Hence, the posterior trajectory orientation density is Gaussian:
\begin{equation}
\begin{aligned}
&q_k^{\theta,(j,\mathcal P,\mathcal C),[\ell+1]}
\!\left(\bm\theta^{1:n_{\omega_k}^{(j)}}\right)=
\mathcal N\!\left(
\bm\theta^{1:n_{\omega_k}^{(j)}};
\hat{\bm\varTheta}_{\pi_k}^{(j,\mathcal P,\mathcal C),[\ell+1]},
\bm\varXi_{\pi_k}^{\bm\varTheta,
(j,\mathcal P,\mathcal C),[\ell+1]}
\right)
\end{aligned}
\label{eq:C6_proof_theta_posterior_density}
\end{equation}

The logarithmic information form of
Eq.~\eqref{eq:C6_proof_theta_posterior_density} is
\begin{equation}
\begin{aligned}
\ln q_k^{\theta,(j,\mathcal P,\mathcal C),[\ell+1]}
\!\left(\bm\theta^{1:n_{\omega_k}^{(j)}}\right)
&=-\frac{1}{2}
\left(
\bm\theta^{1:n_{\omega_k}^{(j)}}-
\hat{\bm\varTheta}_{\pi_k}^{(j,\mathcal P,\mathcal C),[\ell+1]}
\right)^{\mathrm T}
\left(
\bm\varXi_{\pi_k}^{\bm\varTheta,
(j,\mathcal P,\mathcal C),[\ell+1]}
\right)^{-1}
\left(
\bm\theta^{1:n_{\omega_k}^{(j)}}-
\hat{\bm\varTheta}_{\pi_k}^{(j,\mathcal P,\mathcal C),[\ell+1]}
\right)
+\mathtt c_{\theta}\\
&\quad=-\frac{1}{2}
\left(\bm\theta^{1:n_{\omega_k}^{(j)}}\right)^{\mathrm T}
\left(
\bm\varXi_{\pi_k}^{\bm\varTheta,
(j,\mathcal P,\mathcal C),[\ell+1]}
\right)^{-1}
\bm\theta^{1:n_{\omega_k}^{(j)}}\\
&\qquad+
\left(\bm\theta^{1:n_{\omega_k}^{(j)}}\right)^{\mathrm T}
\left(
\bm\varXi_{\pi_k}^{\bm\varTheta,
(j,\mathcal P,\mathcal C),[\ell+1]}
\right)^{-1}
\hat{\bm\varTheta}_{\pi_k}^{(j,\mathcal P,\mathcal C),[\ell+1]}
+\mathtt c_{\theta}
\end{aligned}
\label{eq:C6_proof_theta_posterior_information_form}
\end{equation}

Comparing Eqs.~\eqref{eq:C6_proof_theta_trajectory_quadratic} and
\eqref{eq:C6_proof_theta_posterior_information_form} gives
\endgroup
\begin{equation}
\left(
\bm\varXi_{\pi_k}^{\bm\varTheta,(j,\mathcal P,\mathcal C),[\ell+1]}
\right)^{-1}
=\left(
\bm\varXi_{\omega_k}^{\bm\varTheta,(j)}
\right)^{-1}
+\zeta_{\pi_k}^{(j,\mathcal P,\mathcal C),[\ell+1]}
\bm e_{n_{\omega_k}^{(j)}}
\bm e_{n_{\omega_k}^{(j)}}^{\mathrm T}
\label{eq:C6_proof_theta_information_matrix}
\end{equation}

\begin{equation}
\left(
\bm\varXi_{\pi_k}^{\bm\varTheta,(j,\mathcal P,\mathcal C),[\ell+1]}
\right)^{-1}
\hat{\bm\varTheta}_{\pi_k}^{(j,\mathcal P,\mathcal C),[\ell+1]}
=\left(
\bm\varXi_{\omega_k}^{\bm\varTheta,(j)}
\right)^{-1}
\hat{\bm\varTheta}_{\omega_k}^{(j)}
+\zeta_{\pi_k}^{(j,\mathcal P,\mathcal C),[\ell+1]}
\bm e_{n_{\omega_k}^{(j)}}
\widetilde\theta_{\pi_k}^{(j,\mathcal P,\mathcal C),[\ell+1]}
\label{eq:C6_proof_theta_information_vector}
\end{equation}

Adding the predicted terminal trajectory orientation variance and the
local Gaussian measurement variance gives the equivalent orientation
measurement innovation variance
\begin{equation}
\varXi_{\pi_k}^{\widetilde\theta,(j,\mathcal P,\mathcal C),[\ell+1]}
=\bm e_{n_{\omega_k}^{(j)}}^{\mathrm T}
\bm\varXi_{\omega_k}^{\bm\varTheta,(j)}
\bm e_{n_{\omega_k}^{(j)}}
+\left(
\zeta_{\pi_k}^{(j,\mathcal P,\mathcal C),[\ell+1]}
\right)^{-1}
\label{eq:C6_proof_theta_innovation_variance}
\end{equation}

Substituting Eq.~\eqref{eq:C6_proof_theta_innovation_variance} into
Eq.~\eqref{eq:C6_proof_theta_information_matrix}, taking the inverse, and
rearranging gives
\begin{equation}
\bm\varXi_{\pi_k}^{\bm\varTheta,(j,\mathcal P,\mathcal C),[\ell+1]}
=\bm\varXi_{\omega_k}^{\bm\varTheta,(j)}-
\bm\varXi_{\omega_k}^{\bm\varTheta,(j)}
\bm e_{n_{\omega_k}^{(j)}}
\left(
\varXi_{\pi_k}^{\widetilde\theta,(j,\mathcal P,\mathcal C),[\ell+1]}
\right)^{-1}
\bm e_{n_{\omega_k}^{(j)}}^{\mathrm T}
\bm\varXi_{\omega_k}^{\bm\varTheta,(j)}
\label{eq:C6_proof_theta_covariance_result}
\end{equation}

From the coefficient corresponding to the equivalent orientation
measurement innovation in Eq.~\eqref{eq:C6_proof_theta_covariance_result},
the trajectory orientation state gain vector is
\begin{equation}
\bm\varPsi_{\pi_k}^{(j,\mathcal P,\mathcal C),[\ell+1]}
=\bm\varXi_{\omega_k}^{\bm\varTheta,(j)}
\bm e_{n_{\omega_k}^{(j)}}
\left(
\varXi_{\pi_k}^{\widetilde\theta,(j,\mathcal P,\mathcal C),[\ell+1]}
\right)^{-1}
\label{eq:C6_proof_theta_gain_result}
\end{equation}

Substituting Eq.~\eqref{eq:C6_proof_theta_covariance_result} into
Eq.~\eqref{eq:C6_proof_theta_information_vector} and rearranging gives the
posterior trajectory orientation state mean vector
\begin{equation}
\hat{\bm\varTheta}_{\pi_k}^{(j,\mathcal P,\mathcal C),[\ell+1]}
=\hat{\bm\varTheta}_{\omega_k}^{(j)}+
\bm\varPsi_{\pi_k}^{(j,\mathcal P,\mathcal C),[\ell+1]}
\left[
\widetilde\theta_{\pi_k}^{(j,\mathcal P,\mathcal C),[\ell+1]}-
\bm e_{n_{\omega_k}^{(j)}}^{\mathrm T}
\hat{\bm\varTheta}_{\omega_k}^{(j)}
\right]
\label{eq:C6_proof_theta_mean_result}
\end{equation}

Substituting Eq.~\eqref{eq:C6_proof_theta_gain_result} into
Eq.~\eqref{eq:C6_proof_theta_covariance_result} gives the posterior
trajectory orientation state covariance matrix
\begin{equation}
\bm\varXi_{\pi_k}^{\bm\varTheta,(j,\mathcal P,\mathcal C),[\ell+1]}
=\bm\varXi_{\omega_k}^{\bm\varTheta,(j)}
-\bm\varPsi_{\pi_k}^{(j,\mathcal P,\mathcal C),[\ell+1]}
\bm e_{n_{\omega_k}^{(j)}}^{\mathrm T}
\bm\varXi_{\omega_k}^{\bm\varTheta,(j)}
\label{eq:C6_proof_theta_covariance_compact}
\end{equation}

Therefore, Eqs.~\eqref{eq:C6_proof_theta_mean_result} and
\eqref{eq:C6_proof_theta_covariance_compact} prove
Proposition~\ref{prop:C6_theta_update}.

\newpage
\section{Proof of Proposition~\ref{prop:C6_y_update}}
\label{app:C6_y_update_proof}

\begingroup

The derivation below considers the $m$-th measurement-source variable in
measurement cell $(\mathcal P,\mathcal C)$. According to the structured
variational factorization in Eq.~\eqref{eq:C6_traj_vb_factor} and the
CAVI coordinate-update rule, its variational factor satisfies
\begin{equation}
\begin{aligned}
&\ln q_k^{\bm y^m,(j,\mathcal P,\mathcal C),[\ell+1]}
\!\left(\bm y_k^{\mathcal P,\mathcal C,m}\right)=\mathbb E_{\substack{
q_k^{\bm r,(j,\mathcal P,\mathcal C),[\ell+1]}
q_k^{\theta,(j,\mathcal P,\mathcal C),[\ell+1]}\\
{}\times\prod_{u=1}^{2}
q_k^{\ddot l_u,(j,\mathcal P,\mathcal C),[\ell+1]}}}
\!\Bigg[
\ln p_k\!\left(
\begin{gathered}
\bm r^{1:n_{\omega_k}^{(j)}},
\bm\theta^{1:n_{\omega_k}^{(j)}},
\ddot l_1^{1:n_{\omega_k}^{(j)}},
\ddot l_2^{1:n_{\omega_k}^{(j)}},\\
\bm y_k^{\mathcal P,\mathcal C,m},
\bm z_k^{\mathcal P,\mathcal C,m}
\end{gathered}
\right)
\Bigg]
+\mathtt c_{\bm y_k^{\mathcal P,\mathcal C,m}}
\end{aligned}
\label{eq:C6_proof_y_CAVI}
\end{equation}
\noindent where the expectation is taken with respect to the variational
posteriors of the posterior trajectory kinematic state sequence,
posterior trajectory orientation state sequence, and two posterior
trajectory SSAL state sequences at the $(\ell+1)$-th iteration, excluding
the variational factor of the $m$-th measurement-source variable being
updated; $\mathtt c_{\bm y_k^{\mathcal P,\mathcal C,m}}$ denotes a
constant term independent of
$\bm y_k^{\mathcal P,\mathcal C,m}$.

Given the $m$-th measurement-source variable, the conditional density of
the actual measurement is
\begin{equation}
p\!\left(
\bm z_k^{\mathcal P,\mathcal C,m}\mid
\bm y_k^{\mathcal P,\mathcal C,m}
\right)
=\mathcal N\!\left(
\bm z_k^{\mathcal P,\mathcal C,m};
\bm y_k^{\mathcal P,\mathcal C,m},
\bm R_k
\right)
\label{eq:C6_proof_y_measurement_density}
\end{equation}

Given the terminal trajectory state, the conditional density of the
$m$-th measurement-source variable is
\begin{equation}
\begin{aligned}
&p\!\left(
\bm y_k^{\mathcal P,\mathcal C,m}\mid
\bm r^{1:n_{\omega_k}^{(j)}},
\theta^{n_{\omega_k}^{(j)}},
\ddot l_1^{n_{\omega_k}^{(j)}},
\ddot l_2^{n_{\omega_k}^{(j)}}
\right)=\mathcal N\!\left(
\bm y_k^{\mathcal P,\mathcal C,m};
\dot{\bm H}_k^{(j)}\bm r^{1:n_{\omega_k}^{(j)}},
\mathtt s
\mathcal R\!\left(\theta^{n_{\omega_k}^{(j)}}\right)
\mathcal W\!\left(
\ddot l_1^{n_{\omega_k}^{(j)}},
\ddot l_2^{n_{\omega_k}^{(j)}}
\right)
\mathcal R\!\left(-\theta^{n_{\omega_k}^{(j)}}\right)
\right)
\end{aligned}
\label{eq:C6_proof_y_source_density}
\end{equation}

All remaining terms in the joint density are independent of
$\bm y_k^{\mathcal P,\mathcal C,m}$ and can be absorbed into
$\mathtt c_{\bm y_k^{\mathcal P,\mathcal C,m}}$. Therefore,
Eq.~\eqref{eq:C6_proof_y_CAVI} can be expanded as
\begin{equation}
\begin{aligned}
&\ln q_k^{\bm y^m,(j,\mathcal P,\mathcal C),[\ell+1]}
\!\left(\bm y_k^{\mathcal P,\mathcal C,m}\right)\\
&=\ln p\!\left(
\bm z_k^{\mathcal P,\mathcal C,m}\mid
\bm y_k^{\mathcal P,\mathcal C,m}
\right)+
\mathbb E_{\substack{
q_k^{\bm r,(j,\mathcal P,\mathcal C),[\ell+1]}
q_k^{\theta,(j,\mathcal P,\mathcal C),[\ell+1]}\\
{}\times\prod_{u=1}^{2}
q_k^{\ddot l_u,(j,\mathcal P,\mathcal C),[\ell+1]}}}
\!\Bigg[
\ln p\!\left(
\bm y_k^{\mathcal P,\mathcal C,m}\mid
\bm r^{1:n_{\omega_k}^{(j)}},
\theta^{n_{\omega_k}^{(j)}},
\ddot l_1^{n_{\omega_k}^{(j)}},
\ddot l_2^{n_{\omega_k}^{(j)}}
\right)
\Bigg]
+\mathtt c_{\bm y_k^{\mathcal P,\mathcal C,m}}
\end{aligned}
\label{eq:C6_proof_y_CAVI_decomposition}
\end{equation}

According to the
variational posteriors at the $(\ell+1)$-th iteration obtained in
Propositions~\ref{prop:C6_axis_update} and~\ref{prop:C6_theta_update}, the
expected inverse of the terminal trajectory extent-spread covariance
matrix is
\begin{equation}
\bar{\bm G}_{\pi_k}^{(j,\mathcal P,\mathcal C),[\ell+1]}
=\mathbb E_{q_k^{\theta,(j,\mathcal P,\mathcal C),[\ell+1]}}
\!\left[
\mathcal R\!\left(\theta^{n_{\omega_k}^{(j)}}\right)
\operatorname{diag}\!\left(
\iota_{1,\pi_k}^{(j,\mathcal P,\mathcal C),n_{\pi_k}^{(j)},[\ell+1]},
\iota_{2,\pi_k}^{(j,\mathcal P,\mathcal C),n_{\pi_k}^{(j)},[\ell+1]}
\right)
\mathcal R\!\left(-\theta^{n_{\omega_k}^{(j)}}\right)
\right]
\label{eq:C6_proof_y_expected_precision}
\end{equation}

According to Proposition~\ref{prop:C6_r_update}, the mean vector of the
posterior trajectory kinematic state sequence at the $(\ell+1)$-th
iteration is
$\hat{\bm\varGamma}_{\pi_k}^{(j,\mathcal P,\mathcal C),[\ell+1]}$.
Substituting Eq.~\eqref{eq:C6_proof_y_source_density} into the second term
of Eq.~\eqref{eq:C6_proof_y_CAVI_decomposition} and expanding its Gaussian
log kernel gives
\begin{equation}
\begin{aligned}
&\mathbb E_{\substack{
q_k^{\bm r,(j,\mathcal P,\mathcal C),[\ell+1]}
q_k^{\theta,(j,\mathcal P,\mathcal C),[\ell+1]}\\
{}\times\prod_{u=1}^{2}
q_k^{\ddot l_u,(j,\mathcal P,\mathcal C),[\ell+1]}}}
\!\Bigg[
\ln p\!\left(
\bm y_k^{\mathcal P,\mathcal C,m}\mid
\bm r^{1:n_{\omega_k}^{(j)}},
\theta^{n_{\omega_k}^{(j)}},
\ddot l_1^{n_{\omega_k}^{(j)}},
\ddot l_2^{n_{\omega_k}^{(j)}}
\right)
\Bigg]\\
&\quad=-\frac{1}{2}
\mathbb E_{\substack{
q_k^{\bm r,(j,\mathcal P,\mathcal C),[\ell+1]}
q_k^{\theta,(j,\mathcal P,\mathcal C),[\ell+1]}\\
{}\times\prod_{u=1}^{2}
q_k^{\ddot l_u,(j,\mathcal P,\mathcal C),[\ell+1]}}}
\!\Bigg[
\ln\det\!\Bigg(
\mathtt s
\mathcal R\!\left(\theta^{n_{\omega_k}^{(j)}}\right)
\mathcal W\!\left(
\ddot l_1^{n_{\omega_k}^{(j)}},
\ddot l_2^{n_{\omega_k}^{(j)}}
\right)
\mathcal R\!\left(-\theta^{n_{\omega_k}^{(j)}}\right)
\Bigg)\\
&\qquad+
\left(
\bm y_k^{\mathcal P,\mathcal C,m}-
\dot{\bm H}_k^{(j)}\bm r^{1:n_{\omega_k}^{(j)}}
\right)^{\mathrm T}
\left(
\mathtt s
\mathcal R\!\left(\theta^{n_{\omega_k}^{(j)}}\right)
\mathcal W\!\left(
\ddot l_1^{n_{\omega_k}^{(j)}},
\ddot l_2^{n_{\omega_k}^{(j)}}
\right)
\mathcal R\!\left(-\theta^{n_{\omega_k}^{(j)}}\right)
\right)^{-1}
\left(
\bm y_k^{\mathcal P,\mathcal C,m}-
\dot{\bm H}_k^{(j)}\bm r^{1:n_{\omega_k}^{(j)}}
\right)
\Bigg]
+\mathtt c_{\bm y_k^{\mathcal P,\mathcal C,m}}\\
&\quad=-\frac{1}{2}
\mathbb E_{\substack{
q_k^{\bm r,(j,\mathcal P,\mathcal C),[\ell+1]}
q_k^{\theta,(j,\mathcal P,\mathcal C),[\ell+1]}\\
{}\times\prod_{u=1}^{2}
q_k^{\ddot l_u,(j,\mathcal P,\mathcal C),[\ell+1]}}}
\!\Bigg[
(\bm y_k^{\mathcal P,\mathcal C,m})^{\mathrm T}
\left(
\mathtt s
\mathcal R\!\left(\theta^{n_{\omega_k}^{(j)}}\right)
\mathcal W\!\left(
\ddot l_1^{n_{\omega_k}^{(j)}},
\ddot l_2^{n_{\omega_k}^{(j)}}
\right)
\mathcal R\!\left(-\theta^{n_{\omega_k}^{(j)}}\right)
\right)^{-1}
\bm y_k^{\mathcal P,\mathcal C,m}\\
&\qquad-2
(\bm y_k^{\mathcal P,\mathcal C,m})^{\mathrm T}
\left(
\mathtt s
\mathcal R\!\left(\theta^{n_{\omega_k}^{(j)}}\right)
\mathcal W\!\left(
\ddot l_1^{n_{\omega_k}^{(j)}},
\ddot l_2^{n_{\omega_k}^{(j)}}
\right)
\mathcal R\!\left(-\theta^{n_{\omega_k}^{(j)}}\right)
\right)^{-1}
\dot{\bm H}_k^{(j)}\bm r^{1:n_{\omega_k}^{(j)}}\\
&\qquad+
\left(\dot{\bm H}_k^{(j)}\bm r^{1:n_{\omega_k}^{(j)}}\right)^{\mathrm T}
\left(
\mathtt s
\mathcal R\!\left(\theta^{n_{\omega_k}^{(j)}}\right)
\mathcal W\!\left(
\ddot l_1^{n_{\omega_k}^{(j)}},
\ddot l_2^{n_{\omega_k}^{(j)}}
\right)
\mathcal R\!\left(-\theta^{n_{\omega_k}^{(j)}}\right)
\right)^{-1}
\dot{\bm H}_k^{(j)}\bm r^{1:n_{\omega_k}^{(j)}}
\Bigg]
+\mathtt c_{\bm y_k^{\mathcal P,\mathcal C,m}}\\
&\quad=-\frac{1}{2}
(\bm y_k^{\mathcal P,\mathcal C,m})^{\mathrm T}
\bar{\bm G}_{\pi_k}^{(j,\mathcal P,\mathcal C),[\ell+1]}
\bm y_k^{\mathcal P,\mathcal C,m}
+(\bm y_k^{\mathcal P,\mathcal C,m})^{\mathrm T}
\bar{\bm G}_{\pi_k}^{(j,\mathcal P,\mathcal C),[\ell+1]}
\dot{\bm H}_k^{(j)}
\hat{\bm\varGamma}_{\pi_k}^{(j,\mathcal P,\mathcal C),[\ell+1]}
+\mathtt c_{\bm y_k^{\mathcal P,\mathcal C,m}}
\end{aligned}
\label{eq:C6_proof_y_source_quadratic}
\end{equation}

The log-determinant term and the state-only quadratic term in
Eq.~\eqref{eq:C6_proof_y_source_quadratic} are independent of
$\bm y_k^{\mathcal P,\mathcal C,m}$ and have been absorbed into
$\mathtt c_{\bm y_k^{\mathcal P,\mathcal C,m}}$. On the other hand,
Eq.~\eqref{eq:C6_proof_y_measurement_density} gives the following actual
measurement terms involving $\bm y_k^{\mathcal P,\mathcal C,m}$:
\begin{equation}
\begin{aligned}
\ln p\!\left(
\bm z_k^{\mathcal P,\mathcal C,m}\mid
\bm y_k^{\mathcal P,\mathcal C,m}
\right)
&=-\frac{1}{2}
(\bm y_k^{\mathcal P,\mathcal C,m})^{\mathrm T}
\bm R_k^{-1}
\bm y_k^{\mathcal P,\mathcal C,m}
+(\bm y_k^{\mathcal P,\mathcal C,m})^{\mathrm T}
\bm R_k^{-1}
\bm z_k^{\mathcal P,\mathcal C,m}
+\mathtt c_{\bm y_k^{\mathcal P,\mathcal C,m}}
\end{aligned}
\label{eq:C6_proof_y_measurement_quadratic}
\end{equation}

Substituting Eqs.~\eqref{eq:C6_proof_y_source_quadratic} and
\eqref{eq:C6_proof_y_measurement_quadratic} into
Eq.~\eqref{eq:C6_proof_y_CAVI_decomposition} gives
\begin{equation}
\begin{aligned}
\ln q_k^{\bm y^m,(j,\mathcal P,\mathcal C),[\ell+1]}
\!\left(\bm y_k^{\mathcal P,\mathcal C,m}\right)&=-\frac{1}{2}
(\bm y_k^{\mathcal P,\mathcal C,m})^{\mathrm T}
\left[
\bm R_k^{-1}
+\bar{\bm G}_{\pi_k}^{(j,\mathcal P,\mathcal C),[\ell+1]}
\right]
\bm y_k^{\mathcal P,\mathcal C,m}\\
&\quad+
(\bm y_k^{\mathcal P,\mathcal C,m})^{\mathrm T}
\left[
\bm R_k^{-1}\bm z_k^{\mathcal P,\mathcal C,m}
+\bar{\bm G}_{\pi_k}^{(j,\mathcal P,\mathcal C),[\ell+1]}
\dot{\bm H}_k^{(j)}
\hat{\bm\varGamma}_{\pi_k}^{(j,\mathcal P,\mathcal C),[\ell+1]}
\right]
+\mathtt c_{\bm y_k^{\mathcal P,\mathcal C,m}}
\end{aligned}
\label{eq:C6_proof_y_quadratic_form}
\end{equation}

Eq.~\eqref{eq:C6_proof_y_quadratic_form} is quadratic in
$\bm y_k^{\mathcal P,\mathcal C,m}$. Since $\bm R_k$ and
$\bar{\bm G}_{\pi_k}^{(j,\mathcal P,\mathcal C),[\ell+1]}$ are positive
definite, the coefficient matrix of the quadratic term is positive
definite. Therefore, the variational posterior of the $m$-th
measurement-source variable is Gaussian:
\begin{equation}
\begin{aligned}
q_k^{\bm y^m,(j,\mathcal P,\mathcal C),[\ell+1]}
\!\left(\bm y_k^{\mathcal P,\mathcal C,m}\right)=\mathcal N\!\left(
\bm y_k^{\mathcal P,\mathcal C,m};
\hat{\bm y}_{\pi_k}^{(j,\mathcal P,\mathcal C),m,[\ell+1]},
\bm\varXi_{\pi_k}^{\bm y,(j,\mathcal P,\mathcal C),m,[\ell+1]}
\right)
\end{aligned}
\label{eq:C6_proof_y_posterior_density}
\end{equation}

The logarithmic information form of
Eq.~\eqref{eq:C6_proof_y_posterior_density} is
\begin{equation}
\begin{aligned}
&\ln q_k^{\bm y^m,(j,\mathcal P,\mathcal C),[\ell+1]}
\!\left(\bm y_k^{\mathcal P,\mathcal C,m}\right)\\
&\quad=-\frac{1}{2}
\left(
\bm y_k^{\mathcal P,\mathcal C,m}-
\hat{\bm y}_{\pi_k}^{(j,\mathcal P,\mathcal C),m,[\ell+1]}
\right)^{\mathrm T}
\left(
\bm\varXi_{\pi_k}^{\bm y,(j,\mathcal P,\mathcal C),m,[\ell+1]}
\right)^{-1}
\left(
\bm y_k^{\mathcal P,\mathcal C,m}-
\hat{\bm y}_{\pi_k}^{(j,\mathcal P,\mathcal C),m,[\ell+1]}
\right)
+\mathtt c_{\bm y_k^{\mathcal P,\mathcal C,m}}\\
&\quad=-\frac{1}{2}
(\bm y_k^{\mathcal P,\mathcal C,m})^{\mathrm T}
\left(
\bm\varXi_{\pi_k}^{\bm y,(j,\mathcal P,\mathcal C),m,[\ell+1]}
\right)^{-1}
\bm y_k^{\mathcal P,\mathcal C,m}\\
&\qquad+
(\bm y_k^{\mathcal P,\mathcal C,m})^{\mathrm T}
\left(
\bm\varXi_{\pi_k}^{\bm y,(j,\mathcal P,\mathcal C),m,[\ell+1]}
\right)^{-1}
\hat{\bm y}_{\pi_k}^{(j,\mathcal P,\mathcal C),m,[\ell+1]}
+\mathtt c_{\bm y_k^{\mathcal P,\mathcal C,m}}
\end{aligned}
\label{eq:C6_proof_y_posterior_information_form}
\end{equation}

Comparing Eqs.~\eqref{eq:C6_proof_y_quadratic_form} and
\eqref{eq:C6_proof_y_posterior_information_form} gives the posterior
precision matrix
\begin{equation}
\left(
\bm\varXi_{\pi_k}^{\bm y,(j,\mathcal P,\mathcal C),m,[\ell+1]}
\right)^{-1}
=\bm R_k^{-1}
+\bar{\bm G}_{\pi_k}^{(j,\mathcal P,\mathcal C),[\ell+1]}
\label{eq:C6_proof_y_information_matrix}
\end{equation}

Taking the inverse of both sides of
Eq.~\eqref{eq:C6_proof_y_information_matrix} gives the posterior
covariance matrix of the $m$-th measurement-source variable
\begin{equation}
\bm\varXi_{\pi_k}^{\bm y,(j,\mathcal P,\mathcal C),m,[\ell+1]}
=\left[
\bm R_k^{-1}
+\bar{\bm G}_{\pi_k}^{(j,\mathcal P,\mathcal C),[\ell+1]}
\right]^{-1}
\label{eq:C6_proof_y_covariance_result}
\end{equation}

The corresponding information vector satisfies
\begin{equation}
\begin{aligned}
\left(
\bm\varXi_{\pi_k}^{\bm y,(j,\mathcal P,\mathcal C),m,[\ell+1]}
\right)^{-1}
\hat{\bm y}_{\pi_k}^{(j,\mathcal P,\mathcal C),m,[\ell+1]}
&=\bm R_k^{-1}
\bm z_k^{\mathcal P,\mathcal C,m}
+\bar{\bm G}_{\pi_k}^{(j,\mathcal P,\mathcal C),[\ell+1]}
\dot{\bm H}_k^{(j)}
\hat{\bm\varGamma}_{\pi_k}^{(j,\mathcal P,\mathcal C),[\ell+1]}
\end{aligned}
\label{eq:C6_proof_y_information_vector}
\end{equation}

Left-multiplying both sides of Eq.~\eqref{eq:C6_proof_y_information_vector}
by
$\bm\varXi_{\pi_k}^{\bm y,(j,\mathcal P,\mathcal C),m,[\ell+1]}$
gives the posterior mean of the $m$-th measurement-source variable
\begin{equation}
\begin{aligned}
\hat{\bm y}_{\pi_k}^{(j,\mathcal P,\mathcal C),m,[\ell+1]}
&=\bm\varXi_{\pi_k}^{\bm y,(j,\mathcal P,\mathcal C),m,[\ell+1]}
\left[
\bm R_k^{-1}
\bm z_k^{\mathcal P,\mathcal C,m}
+\bar{\bm G}_{\pi_k}^{(j,\mathcal P,\mathcal C),[\ell+1]}
\dot{\bm H}_k^{(j)}
\hat{\bm\varGamma}_{\pi_k}^{(j,\mathcal P,\mathcal C),[\ell+1]}
\right]
\end{aligned}
\label{eq:C6_proof_y_mean_result}
\end{equation}

Since each measurement-source variable depends only on its corresponding
actual measurement and the common terminal trajectory state, no
cross-quadratic term exists between any two distinct measurement-source
variables. Therefore, their variational posteriors are mutually
independent. The above derivation holds for
$m=1,\ldots,M_k^{\mathcal P,\mathcal C}$. Consequently,
Eqs.~\eqref{eq:C6_proof_y_mean_result} and
\eqref{eq:C6_proof_y_covariance_result} prove
Proposition~\ref{prop:C6_y_update}.

\endgroup

\newpage
\section{Proof of Proposition~\ref{prop:C6_component_weight}}
\label{app:C6_component_weight_proof}

\begin{lemma}[Gaussian Moment Matching for Extended-Target Measurements]
	\label{lem:C6_predicted_measurement_moments}
	Let~$\bm\xi_k$ denote an extent state and suppose that the conditional
	measurement density is
	\begin{equation}
	p(\bm z_k\mid\bm r_k,\bm\xi_k)
	=
	\mathcal N\!\left(
	\bm z_k;
	\bm H_k\bm r_k,
	\bm S_k^{\bm z}(\bm\xi_k)+\bm R_k
	\right)
	\label{eq:C6_lemma_conditional_measurement}
	\end{equation}
	where~$\bm r_k$ has predicted mean~$\hat{\bm r}_{\omega_k}$ and
	covariance matrix~$\bm\varXi_{\omega_k}^{\bm r}$. Then the first two
	moments of the predicted measurement are
	\begin{equation}
	\mathbb E[\bm z_k]
	=
	\bm H_k\hat{\bm r}_{\omega_k}
	\label{eq:C6_lemma_measurement_mean}
	\end{equation}
	\begin{equation}
	\begin{aligned}
	\operatorname{Cov}(\bm z_k)
	=
	\bm H_k
	\bm\varXi_{\omega_k}^{\bm r}
	\bm H_k^{\mathrm T}+
	\mathbb E\!\left[
	\bm S_k^{\bm z}(\bm\xi_k)
	\right]
	+\bm R_k
	\end{aligned}
	\label{eq:C6_lemma_measurement_covariance}
	\end{equation}
	
	Consequently, the predicted measurement density can be approximated by
	a Gaussian distribution with the moments in
	Eqs.~\eqref{eq:C6_lemma_measurement_mean} and
	\eqref{eq:C6_lemma_measurement_covariance}.
\end{lemma}

\begin{IEEEproof}
	The law of total expectation gives
	\begin{equation}
	\begin{aligned}
	\mathbb E[\bm z_k]
	&=
	\mathbb E\!\left[
	\mathbb E[\bm z_k\mid\bm r_k,\bm\xi_k]
	\right]\\
	&\quad=
	\mathbb E[\bm H_k\bm r_k]
	=
	\bm H_k\hat{\bm r}_{\omega_k}
	\end{aligned}
	\label{eq:C6_lemma_proof_mean}
	\end{equation}
	
	The law of total covariance gives
	\begin{equation}
	\begin{aligned}
	\operatorname{Cov}(\bm z_k)
	&=
	\mathbb E\!\left[
	\operatorname{Cov}
	(\bm z_k\mid\bm r_k,\bm\xi_k)
	\right]+
	\operatorname{Cov}\!\left(
	\mathbb E[\bm z_k\mid\bm r_k,\bm\xi_k]
	\right)\\
	&\quad=
	\mathbb E\!\left[
	\bm S_k^{\bm z}(\bm\xi_k)
	\right]
	+\bm R_k
	+\bm H_k
	\bm\varXi_{\omega_k}^{\bm r}
	\bm H_k^{\mathrm T}
	\end{aligned}
	\label{eq:C6_lemma_proof_covariance}
	\end{equation}
	
	Matching these two moments gives the stated Gaussian approximation.
	This is the same predicted-measurement moment structure used for the MEM
	in~\cite[Eq.~(21)]{2019YangShishanMEM}, but it follows directly from the
	conditional Gaussian measurement model and therefore does not depend on
	a particular extent parameterization.
\end{IEEEproof}

\begin{corollary}[DOAM Predicted Measurement Likelihood]
	\label{cor:C6_DOAM_component_likelihood}
	For the terminal predicted DGDIG density in
	Eq.~\eqref{eq:C6_terminal_pred_density}, let
	\begingroup
	
	\begin{equation}
	\bm S_k^{\bm z}(\bm\xi_k)
	=
	\mathtt s
	\mathcal R\!\left(\theta^{n_{\omega_k}^{(j)}}\right)
	\mathcal W\!\left(
	\ddot l_1^{n_{\omega_k}^{(j)}},
	\ddot l_2^{n_{\omega_k}^{(j)}}
	\right)
	\mathcal R\!\left(-\theta^{n_{\omega_k}^{(j)}}\right)
	\label{eq:C6_corollary_DOAM_conditional_extent}
	\end{equation}
	\endgroup
	
	Then the Gaussian moment-matched covariance matrix of a target-generated
	measurement is given by Eq.~\eqref{eq:C6_DOAM_measurement_cov}, where the
	expected extent-spread covariance is given by
	Eq.~\eqref{eq:C6_DOAM_expected_extent}. Applying the same predicted
	measurement moments independently to all measurements in
	cell~$(\mathcal P,\mathcal C)$ yields the likelihood approximation in
	Eq.~\eqref{eq:C6_DOAM_component_likelihood}.
\end{corollary}

\begin{IEEEproof}
	For the $u$-th terminal predicted SSAL state, the mean of the
	inverse-Gamma distribution exists when
	$a_{u,\omega_k}^{(j),n_{\omega_k}^{(j)}}>1$ and is
	\begingroup
	
	\begin{equation}
	\overline{\ddot l}_{u,\omega_k}^{(j),n_{\omega_k}^{(j)}}
	=
	\frac{
		b_{u,\omega_k}^{(j),n_{\omega_k}^{(j)}}
	}{
	a_{u,\omega_k}^{(j),n_{\omega_k}^{(j)}}-1
},
\qquad u\in\{1,2\}
\label{eq:C6_corollary_axis_mean}
\end{equation}
\endgroup

The characteristic function of a Gaussian distribution gives
\begingroup

\begin{equation}
\begin{aligned}
\mathbb E\!\left[
\cos\!\left(2\theta^{n_{\omega_k}^{(j)}}\right)
\right]
&=
\cos\!\left(
2\hat\theta_{\omega_k}^{(j),n_{\omega_k}^{(j)}}
\right)
\mathrm e^{-2\varXi_{\omega_k}^{\theta,(j),n_{\omega_k}^{(j)}}}
\end{aligned}
\label{eq:C6_corollary_cosine_moment}
\end{equation}
\begin{equation}
\begin{aligned}
\mathbb E\!\left[
\sin\!\left(2\theta^{n_{\omega_k}^{(j)}}\right)
\right]
&=
\sin\!\left(
2\hat\theta_{\omega_k}^{(j),n_{\omega_k}^{(j)}}
\right)
\mathrm e^{-2\varXi_{\omega_k}^{\theta,(j),n_{\omega_k}^{(j)}}}
\end{aligned}
\label{eq:C6_corollary_sine_moment}
\end{equation}
\endgroup

Since the orientation and the two SSAL states are mutually independent
under Eq.~\eqref{eq:C6_terminal_pred_density}, substituting
Eqs.~\eqref{eq:C6_corollary_axis_mean}--\eqref{eq:C6_corollary_sine_moment}
into the expectation of Eq.~\eqref{eq:C6_corollary_DOAM_conditional_extent}
gives Eq.~\eqref{eq:C6_DOAM_expected_extent}.

\begingroup

Applying Lemma~\ref{lem:C6_predicted_measurement_moments} independently
to all measurements in the cell and adopting the independent Gaussian
moment-matching approximation gives
\endgroup
\begingroup

\begin{equation}
\begin{aligned}
&\mathcal L_{\omega_k}^{(j,\mathcal P,\mathcal C)}\approx
\prod_{m=1}^{M_k^{\mathcal P,\mathcal C}}
\mathcal N\!\left(
\bm z_k^{\mathcal P,\mathcal C,m};
\bm H_k
\hat{\bm r}_{\omega_k}^{(j),n_{\omega_k}^{(j)}},
\bm\varXi_{\omega_k}^{\bm z,(j),n_{\omega_k}^{(j)}}
\right)
\end{aligned}
\label{eq:C6_corollary_cell_likelihood}
\end{equation}
\endgroup
which is identical to Eq.~\eqref{eq:C6_DOAM_component_likelihood} and
thereby proves the corollary.
\end{IEEEproof}

\begingroup

The factorized terminal density in
Eq.~\eqref{eq:C6_terminal_pred_density} follows by extracting the terminal
marginals of the predicted DGDIG trajectory component, and
Corollary~\ref{cor:C6_DOAM_component_likelihood} gives its
measurement-cell likelihood. Substituting the normalized predicted weight
in Eq.~\eqref{eq:C6_normalized_pred_weight} and the \DOAM likelihood in
Eq.~\eqref{eq:C6_DOAM_component_likelihood} into the detected-component
weight update of the TCPHD recursion in \cite{Cheng2026ExplicitExtent},
whose partition product, cardinality normalizer, and component coefficient
are given by \cite[Eqs.~(42), (44), and~(46)]
{Cheng2026ExplicitExtent}, directly yields
Eq.~\eqref{eq:C6_post_component_weight}, thereby establishing
Proposition~\ref{prop:C6_component_weight}.
\endgroup

\renewcommand{\refname}{Supplementary References}
\putbib[bibtex/references]
\end{bibunit}

\end{document}